\documentclass{article}[11pt]

\usepackage{fullpage}
\usepackage{cite,color,xspace}
\usepackage{todonotes}
\usepackage{amsmath}
\usepackage{amsfonts}
\usepackage{enumitem}

\usepackage[T1]{fontenc}
\usepackage{fourier}

\usepackage{bibunits}

\usepackage{subfig}

\usepackage{tikz,tikz-qtree,tikz-qtree-compat,tikz-3dplot}

\usepackage{xcolor}
\usepackage{tabu}

\usepackage{amsthm}
\usepackage{cite}
\usepackage{url}
\usepackage{graphicx}
\usepackage{algorithm}
\usepackage[noend]{algpseudocode}

\usepackage{appendix}

\newcommand{\yell}[1]{{\color{red} \textbf{#1}}}
\newcommand{\floor}[1]{\left\lfloor #1 \right\rfloor}
\newcommand{\norm}[1]{\left\| #1 \right\|}
\newcommand{\idim}{\mathsf{idim}}
\newcommand{\spw}{\mathsf{spw}}
\newcommand{\sylv}[2]{\nabla_{#1,#2}}
\newcommand{\stein}[2]{\Delta_{#1,#2}}

\newcommand{\mvec}[1]{\mathbf{#1}}
\newcommand{\va}{\mvec{a}}
\newcommand{\vb}{\mvec{b}}
\newcommand{\vc}{\mvec{c}}
\newcommand{\vd}{\mvec{d}}
\newcommand{\ve}{\mvec{e}}
\newcommand{\vf}{\mvec{f}}
\newcommand{\vg}{\mvec{g}}
\newcommand{\vh}{\mvec{h}}
\newcommand{\vi}{\mvec{i}}
\newcommand{\vim}{\mvec{\imath}}
\newcommand{\vj}{\mvec{j}}
\newcommand{\vjm}{\mvec{\jmath}}
\newcommand{\vq}{\mvec{q}}
\newcommand{\vp}{\mvec{p}}

\newcommand{\vu}{\mvec{u}}
\newcommand{\vv}{\mvec{v}}
\newcommand{\vw}{\mvec{w}}
\newcommand{\vx}{\mvec{x}}
\newcommand{\vy}{\mvec{y}}
\newcommand{\vz}{\mvec{z}}
\newcommand{\vM}{\mvec{M}}
\newcommand{\vA}{\mvec{A}}
\newcommand{\vB}{\mvec{B}}
\newcommand{\vC}{\mvec{C}}
\newcommand{\vD}{\mvec{D}}
\newcommand{\vE}{\mvec{E}}
\newcommand{\vF}{\mvec{F}}
\newcommand{\cF}{\mvec{\mathcal{F}}}
\newcommand{\vG}{\mvec{G}}
\newcommand{\vH}{\mvec{H}}
\newcommand{\vcH}{\mvec{\mathcal{H}}}
\newcommand{\vI}{\mvec{I}}
\newcommand{\vJ}{\mvec{J}}
\newcommand{\vK}{\mvec{K}}
\newcommand{\vP}{\mvec{P}}
\newcommand{\vQ}{\mvec{Q}}
\newcommand{\vL}{\mvec{L}}
\newcommand{\vR}{\mvec{R}}
\newcommand{\vS}{\mvec{S}}
\newcommand{\vT}{\mvec{T}}
\newcommand{\vV}{\mvec{V}}
\newcommand{\vU}{\mvec{U}}
\newcommand{\vW}{\mvec{W}}
\newcommand{\vX}{\mvec{X}}
\newcommand{\vY}{\mvec{Y}}
\newcommand{\vZ}{\mvec{Z}}
\newcommand{\vzero}{\mvec{0}}
\newcommand{\vone}{\mvec{1}}

\newcommand{\ar}[1]{}
\newcommand{\cmr}[1]{}
\newcommand{\rpu}[1]{}
\newcommand{\agu}[1]{}

\newcommand{\calH}{\mathcal H}

\newcommand{\calR}{\mathcal R}

\newcommand{\calC}{\mathcal C}

\newcommand{\eat}[1]{}

\newcommand{\poly}{\mathrm{poly}}

\newcommand{\eps}{\epsilon}

\newcommand{\cA}{\mathcal{A}}
\newcommand{\cM}{\mathcal{M}}
\newcommand{\cR}{\mathcal{R}}

\newcommand{\F}{\mathbb{F}}
\newcommand{\N}{\mathbb{N}}

\newcommand{\R}{\mathbb{R}}
\newcommand{\C}{\mathbb{C}}
\newcommand{\Z}{\mathbb{Z}}

\newcommand{\be}{\begin{enumerate}}
\newcommand{\ee}{\end{enumerate}}
\newcommand{\bi}{\begin{itemize}}
\newcommand{\ei}{\end{itemize}}
\newcommand{\beq}{\begin{equation}}
\newcommand{\eeq}{\end{equation}}

\newcommand{\bp}{\begin{proof}}
\newcommand{\ep}{\end{proof}}
\newcommand{\bcor}{\begin{cor}}
\newcommand{\ecor}{\end{cor}}
\newcommand{\bthm}{\begin{thm}}
\newcommand{\ethm}{\end{thm}}
\newcommand{\blmm}{\begin{lmm}}
\newcommand{\elmm}{\end{lmm}}
\newcommand{\bdefn}{\begin{defn}}
\newcommand{\edefn}{\end{defn}}
\newcommand{\bprop}{\begin{prop}}
\newcommand{\eprop}{\end{prop}}
\newcommand{\bconj}{\begin{conj}}
\newcommand{\econj}{\end{conj}}
\newcommand{\bopm}{\begin{opm}}
\newcommand{\eopm}{\end{opm}}
\newcommand{\brmk}{\begin{rmk}}
\newcommand{\ermk}{\end{rmk}}

\newcommand{\diag}[1]{\mathop{\textnormal{diag}} #1}
\newcommand{\rank}[1]{\mathop{\textnormal{rank}} #1}

\theoremstyle{plain}                   % default
\newtheorem{thm}{Theorem}[section]
\newtheorem{lmm}[thm]{Lemma}
\newtheorem{prop}[thm]{Proposition}
\newtheorem{cor}[thm]{Corollary}

\theoremstyle{definition}              % Examples and all

\newtheorem{opm}[thm]{Open Problem}
\newtheorem{conj}[thm]{Conjecture}

\newtheorem{defn}[thm]{Definition}

\newtheorem{rmk}[thm]{Remark}
\newtheorem{claim}{Claim}

\newcommand{\bbox}{
\begin{center}
\begin{tabular}{|c|}
\hline
}
\newcommand{\ebox}{
\\
\hline
\end{tabular}
\end{center}
}

\newlength{\toppush}
\newcommand{\defeq}{\stackrel{\mathrm{def}}{=}}

\algrenewcommand\algorithmicrequire{\textbf{Input:}}
\algrenewcommand\algorithmicensure{\textbf{Output:}}
\algrenewcommand\algorithmicwhile{\textbf{While}}
\algrenewcommand\algorithmicfor{\textbf{For}}
\algrenewcommand\algorithmicreturn{\textbf{Return}}
\algrenewcommand\algorithmicif{\textbf{If}}

\newcommand{\tO}{\tilde{O}}

\definecolor{ao}{rgb}{0,0.5,0}

\newcommand{\marked}[1]{\textcolor{red}{#1}}

\DeclareMathOperator{\coef}{Coeff}

\newcommand{\eqdef}{\stackrel{\text{def}}{=}}

\newcommand{\ind}[1]{\mathbb{1}_{{#1}}}

\newcommand{\ATmult}{\textsc{TransposeMult}}
\newcommand{\Amult}{\textsc{MatrixVectMult}}
\newcommand{\Amultrat}{\textsc{MatrixVectMultRational}}

\newcommand{\kry}{\mathcal{K}}
\newcommand{\K}{\mathsf{K}}

\newcommand{\ip}[2]{\left\langle #1,#2\right\rangle}

\newcommand{\width}{recurrence width}
\newcommand{\ewidth}{recurrence error width}

\newcommand{\barN}{\overline{N}}
\newcommand{\bard}{\bar{d}}

\usepackage{arydshln}
\usepackage{comment}
\usepackage[colorlinks, linkcolor=blue, citecolor=red]{hyperref}
\usepackage{multirow}
\usepackage{tablefootnote}

\begin{document}
%\title{\textbf{Recurrence Width for Structured Dense Matrix Vector Multiplication}}
\title{\textbf{A two-pronged progress in structured dense matrix vector multiplication}}
\author{\textsc{Christopher De Sa}\footnotemark[3] \and \textsc{Albert Gu}\footnotemark[1] \and \textsc{Rohan Puttagunta}\footnotemark[1] \and \textsc{Christopher R\'{e}}\footnotemark[1] \and \textsc{Atri Rudra}\footnotemark[2]}
\date{\footnotemark[1]~~Department of Computer Science\\
Stanford University\\
\texttt{\{albertgu,rohanp,chrismre\}@stanford.edu}\\
\vspace*{2mm}
\footnotemark[2]~~Department of Computer Science and Engineering\\
University at Buffalo, SUNY\\
\texttt{atri@buffalo.edu}\\
\vspace*{2mm}
\footnotemark[3]~~Department of Computer Science\\
Cornell University\\
\texttt{cdesa@cs.cornell.edu}
}

\maketitle

\setcounter{page}{0}
\thispagestyle{empty}

\begin{abstract}
Matrix-vector multiplication is one of the most fundamental computing primitives. Given a matrix $\vA\in\F^{N\times N}$ and a vector $\vb\in\F^N$, it is known that in the worst case $\Theta(N^2)$ operations over $\F$ are needed to compute $\vA\vb$. Many types of structured matrices do admit faster multiplication. However, even given a matrix $\vA$ that is known to have this property, it is hard in general to recover a representation of $\vA$ exposing the actual fast multiplication algorithm. Additionally, it is not known in general whether the inverses of such structured matrices can be computed or multiplied quickly. A broad question is thus to identify classes of {\em structured dense} matrices that can be represented with $O(N)$ parameters, and for which matrix-vector multiplication (and ideally other operations such as solvers) can be performed in a sub-quadratic number of operations. 

One such class of structured matrices that admit near-linear matrix-vector multiplication are the {\em orthogonal polynomial transforms} whose rows correspond to a family of orthogonal polynomials. Other well known classes include the Toeplitz, Hankel, Vandermonde, Cauchy matrices and their extensions (e.g.\ confluent Cauchy-like matrices) that are all special cases of a low {\em displacement rank} property.

In this paper, we make progress on two fronts:
\begin{enumerate}
\item We introduce the notion of {\em \width} of matrices. For matrices $\vA$ with constant \width, we design algorithms to compute both $\vA\vb$ and $\vA^T\vb$ with a near-linear number of operations. This notion of width is finer than all the above classes of structured matrices and thus we can compute near-linear matrix-vector multiplication for all of them using the same core algorithm.
Furthermore, we show that it is possible to solve the harder problems of {\em recovering} the structured parameterization of a matrix with low recurrence width, and computing matrix-vector product with its {\em inverse} in near-linear time.
  %We consider extensions and variants of this width to other notions that can also be reduced to the main algorithms. These reductions are captured through operations on certain types of Krylov matrices, which have been used in many previous works such as Lanczos' and Wiedemann's algorithms. Furthermore, we show that the subclasses of Krylov matrices we use also exhibit a simple \width\ structure.
\item We additionally adapt our algorithm to a matrix-vector multiplication algorithm for a much more general class of matrices with displacement structure: those with low displacement rank with respect to quasiseparable matrices. This result is a novel connection between matrices with displacement structure and those with rank structure, two large but previously separate classes of structured matrices. This class includes Toeplitz-plus-Hankel-like matrices, the Discrete Trigonometric Transforms, and more, and captures all previously known matrices with displacement structure under a unified parametrization and algorithm.
\end{enumerate}
Our work unifies, generalizes, and simplifies existing state-of-the-art results in structured matrix-vector multiplication. Finally, we show how applications in areas such as multipoint evaluations of multivariate polynomials can be reduced to problems involving low \width\ matrices. % and solving system of linear equations.
\end{abstract}

\newpage

\section{Introduction}
\label{sec:intro}

%\subsection{Background and Overview of Our Results}

Given a generic matrix $\vA \in \F^{N \times N}$ over any field $\F$, the problem of matrix-vector multiplication by $\vA$ has a clear $\Omega(N^2)$ lower bound in general.\footnote{This is in terms of operations over $\F$ in exact arithmetic, which will be our primary focus throughout this paper. Also we will focus exclusively on computing the matrix vector product exactly as opposed to approximately. We leave the study of approximation and numerical stability (which are important for computation over real/complex numbers) for future work.} Many special classes of matrices, however, admit multiplication algorithms that only require near linear (in $N$) operations. In general, any matrix $\vA$ can be identified with the smallest linear circuit that computes the linear function induced by $\vA$.
%$(x_0,\dots,x_{N-1}) \mapsto (\sum_j \vA[0,j]x_j,\dots, \sum_j \vA[N-1,j]x_j)$.
This is a tight characterization of the best possible arithmetic complexity of any matrix-vector multiplication algorithm for $\vA$ that uses linear operations\footnote{Furthermore over any infinite field $\F$, non-linear operations do not help, i.e.\ the smallest linear circuit is within a constant factor size of the smallest arithmetic circuit~\cite{burgisser2013algebraic}.} and captures all known structured matrix vector multiplication algorithms. Additionally, it implies the classical \emph{transposition principle}~\cite{bordewijk1956, burgisser2013algebraic}, which states that the number of linear operations needed for matrix-vector multiplication by $\vA$ and $\vA^T$ are within constant factors of each other. (Appendix~\ref{sec:spw} presents an overview of other known results.) Thus, this quantity is a very general characterization of the complexity of a matrix. However, it has several shortcomings. 
Most importantly, given a matrix, the problem of finding the minimum circuit size is APX-hard~\cite{boyar2013logic}, and the best known upper bound on the approximation ratio is only $O(N/\log{N})$~\cite{lupanov}. 
Finally, this characterization does not say anything about the inverse of a structured matrix, even though $\vA^{-1}$ is often also structured if $\vA$ is. Thus, much work in the structured matrix vector multiplication literature has focused on the following problem:
\begin{quote}
\textit{Identify the most general class of structured matrices $\vA$ so that one can in near-linear operations compute both $\vA\vb$ and $\vA^{-1}\vb$. In addition given an arbitrary matrix $\vA$, we would like to efficiently recover the representation of the matrix in the chosen class.}
\end{quote}

One of most general classes studied so far is the structure of low {\em displacement rank}.  The notion of displacement rank, which was introduced in the seminal work of Kailath et al.~\cite{KKM79}, is defined as follows. Given any pair of matrices $(\vL,\vR)$, the displacement rank of $\vA$ with respect to $(\vL,\vR)$ is the rank of the {\em error matrix}\footnote{This defines the Sylvester type displacement operator. Our algorithms work equally well for the Stein type displacement operator $\vA-\vL\vA\vR$. We also note that the terminology of error matrix to talk about the displacement operator is non-standard: we make this change to be consistent with our general framework where this terminology makes sense.}:
\begin{equation}
  \label{eq:intro-dr}
  \vE=\vL\vA-\vA\vR.
\end{equation}
Until very recently, the most general structure in this framework was studied by Olshevsky and Shokrollahi~\cite{OS00}, who show that any matrix with a displacement rank of $r$ with respect to {\em Jordan form} matrices (see Section~\ref{subsec:known}) $\vL$ and $\vR$ supports near-linear operations matrix-vector multiplication. A very recent pre-print extended this result to the case when both $\vL$ and $\vR$ are {\em block companion matrices}~\cite{bostan2017}. In our first main result, 
\begin{quote}
\emph{we substantially generalize these results to the case when $\vL$ and $\vR$ are {\em quasiseparable} matrices.}
\end{quote}
Quasiseparable matrices are a type of \emph{rank-structured} matrix, defined by imposing low-rank constraints on certain submatrices, which have become widely used in efficient numerical computations~\cite{VVM08}. This result represents a new unification of two large, important, and previously separate classes of structured matrices, namely those with displacement structure and those with rank structure.

Another general class of matrices are {\em orthogonal polynomial transforms}~\cite{chihara,cheb-approx,jacobi,zernike}. We first observe that known results on such polynomials~\cite{driscoll,potts1998fast, bostan2010op} can be easily extended to polynomial recurrences of bounded width. However, these results and those for displacement rank matrices tend to be proved using seemingly disparate algorithms and techniques.
In our second main result, 
\begin{quote}
{\em we introduce the notion of \emph{recurrence width}, which captures the class of orthogonal polynomial transforms as well as matrices of low displacement rank with respect to Jordan form matrices.}
\end{quote}
We design a simple and general near-linear-operation matrix-vector multiplication algorithm for low recurrence width matrices, hence capturing these previously disparate classes. Moreover, we observe that we can solve the harder problems of recovery (i.e.\ recovering the recurrence width parameterization given a matrix) and inverse for the polynomials recurrences of bounded width in polynomial time. Figure~\ref{fig:res} shows the relationship between the various classes of matrices that we have discussed, and collects the relevant known results and our new results. (All the algorithms result in linear circuits and hence by the transposition principle, we get algorithms for $\vA^T\vb$ and $\vA^{-T}\vb$ with the same complexity as $\vA\vb$ and $\vA^{-1}\vb$ respectively.)

%\begin{figure}
%\centering
%\includegraphics[scale=.4]{matrix-classes.jpg}
%\caption{Relationships between various matrix classes considered in this paper.}
%\label{fig:matrix-classes}
%\end{figure}
\taburulecolor{lightgray}
\begin{figure*}[ht]
\begin{center}
\begin{tikzpicture}[scale = 3]

%%%%%% Legend
\draw (2,-1.1) rectangle (3.9,-1.6);
\node[below] at (3,-1.1) {{\small \textcolor{gray}{Operation count legend}}};
\node[below right] at (2,-1.25) {%
{\footnotesize
\begin{tabu}{|c|c|}
\hline
$\vA\vb$ compute & $\vA^{-1}\vb$ compute\\
\hline
\textcolor{gray}{
$\vA\vb$ pre-processing} & \textcolor{gray}{$\vA^{-1}\vb$ pre-processing}\\
\hline
\end{tabu}
}
};

%%%%%%% The matrix classes

%% Vandermonde
\node at (4,-4.5) {Vandermonde};

%%% poly style

% Recurrrence width (poly)
\path [fill=blue, opacity=.04] (2.5,-4.75) -- (4.5,-4.75) -- (3.625,-2.5) -- (1.625, -2.5) -- (2.5, -4.75);
\node[below right] at (2, -2.5) {\textbf{Recurrence width (poly)}};
\node[below right] at (2.25, -2.6) {{\small Theorem~\ref{thm:mainATb},\ref{thm:recur-inverse}}};
%Runtime 
\node[below right] at (2.125,-2.8) {%
{\footnotesize
\begin{tabu}{|c|c|}
\hline
$\tO(t^2N)$ & $\tO(t^2N)$\\
\hline
\textcolor{gray}{
$\tO(t^{\omega}N)$} & \textcolor{gray}{$\tO(t^{\omega}N)$}\\
\hline
\end{tabu}
}
};

%Recurrence width (matrix)
\path [fill=red, opacity=0.06] (2.5, -4.75) -- (5.5, -4.75) -- (6.55,-2) -- (1.425, -2) -- (2.5, -4.75);
\node[below right] at (3.25, -2) {\textbf{Recurrence width (matrix)}};
\node[below] at (4,-2.1) {{\small Theorem~\ref{cor:disp-rank},\ref{cor:disp-inverse}}};
%Runtime 
\node[below right] at (3.35,-2.25) {%
{\footnotesize
\begin{tabu}{|c|c|}
\hline
$\tO(rt^2N)$ & $\tO(rt^2N)$\\
\hline
\textcolor{gray}{
$\tO(t^{\omega}N)$} & \textcolor{gray}{$\tO((t^{\omega}+r^2t^2)N)$}\\
\hline
\end{tabu}
}
};

%Orthogonal polys
\path [fill=blue, opacity=.08] (2.5,-4.75) -- (4.5,-4.75) -- (3.9,-3.25) -- (1.92, -3.25) -- (2.5, -4.75);
\node [below right] at (2.2, -3.5) {Orthogonal polynomials};
\node [below right] at (2.7,-3.62) {{\small \cite{driscoll,bostan2010op}}};
%Runtime 
\node[below right] at (2.5,-3.75) {%
{\footnotesize
\begin{tabu}{|c|c|}
\hline
$\tO(N)$ & $\tO(N)$\\
\hline
\textcolor{gray}{
--} & \textcolor{gray}{--}\\
\hline
\end{tabu}
}
};

%%%%% DR style

%Displacement rank (quasi)
\path [fill=green, opacity=0.04] (3.5,-4.75) -- (5.5, -4.75) -- (7.125, -0.5) -- (5.125,-0.5) -- (3.5, -4.75);
\node [below right] at (5.25, -.5) {\textbf{Quasiseparable}};
\node [below right] at (5.25, -.62) {{\small \textbf{(Displacement rank)}~Theorem~\ref{thm:intro-dr-quasi}}};
%Runtime 
\node[below right] at (5.35,-.75) {%
{\footnotesize
\begin{tabu}{|c|c|}
\hline
$\tO(rt^{\omega}N)$ & $\tO(rt^{\omega}N)$\\
\hline
\textcolor{gray}{
--} & \textcolor{gray}{$\tO((r+t)^2t^{\omega}N)$}\\
\hline
\end{tabu}
}
};

%Displacement rank (confluent Cauchy)
\path [fill=green, opacity=0.08] (3.5,-4.75) -- (5.5, -4.75) -- (6.85, -1.25) -- (4.85,-1.25) -- (3.5, -4.75);
\node [below right] at (4.4, -3.5) {Confluent Cauchy-like};
\node [below] at (5,-3.6) {{\small \cite{OS00,bostan2017}}};
%Runtime 
\node[below right] at (4.5,-3.75) {%
{\footnotesize
\begin{tabu}{|c|c|}
\hline
$\tO(rN)$ & $\tO(rN)$\\
\hline
\textcolor{gray}{
--} & \textcolor{gray}{$\tO(r^{\omega-1}N)$}\\
\hline
\end{tabu}
}
};

%Displacement rank (companion matrix)
\path [fill=green, opacity=0.12] (3.5,-4.75) -- (5.5, -4.75) -- (6.08, -3.25) -- (4.08,-3.25) -- (3.5, -4.75);
\node [below right] at (5, -1.25) {Block companion matrix};
\node [below right] at (5, -1.37) {{\small (Displacement rank)~\cite{bostan2017}}};
%Runtime 
\node[below right] at (5.25,-1.49) {%
{\footnotesize
\begin{tabu}{|c|c|}
\hline
$\tO(rN)$ & $\tO(rN)$ \\
\hline
\textcolor{gray}{
--} & \textcolor{gray}{$\tO(r^{\omega-1}N)$}\\
\hline
\end{tabu}
}
};

\end{tikzpicture}
\end{center}
\caption{Overview of results and hierarchy of matrix classes. Operation count includes pre-processing (on $\vA$ only) and computation (on $\vA$ and $\vb$) where $\tO(\cdot)$ hides polylog factors in $N$. Each class has a parameter controlling the degree of structure: $t$ refers to the recurrence width or quasiseparability degree, and $r$ indicates the displacement rank.}
\label{fig:res}
\end{figure*}
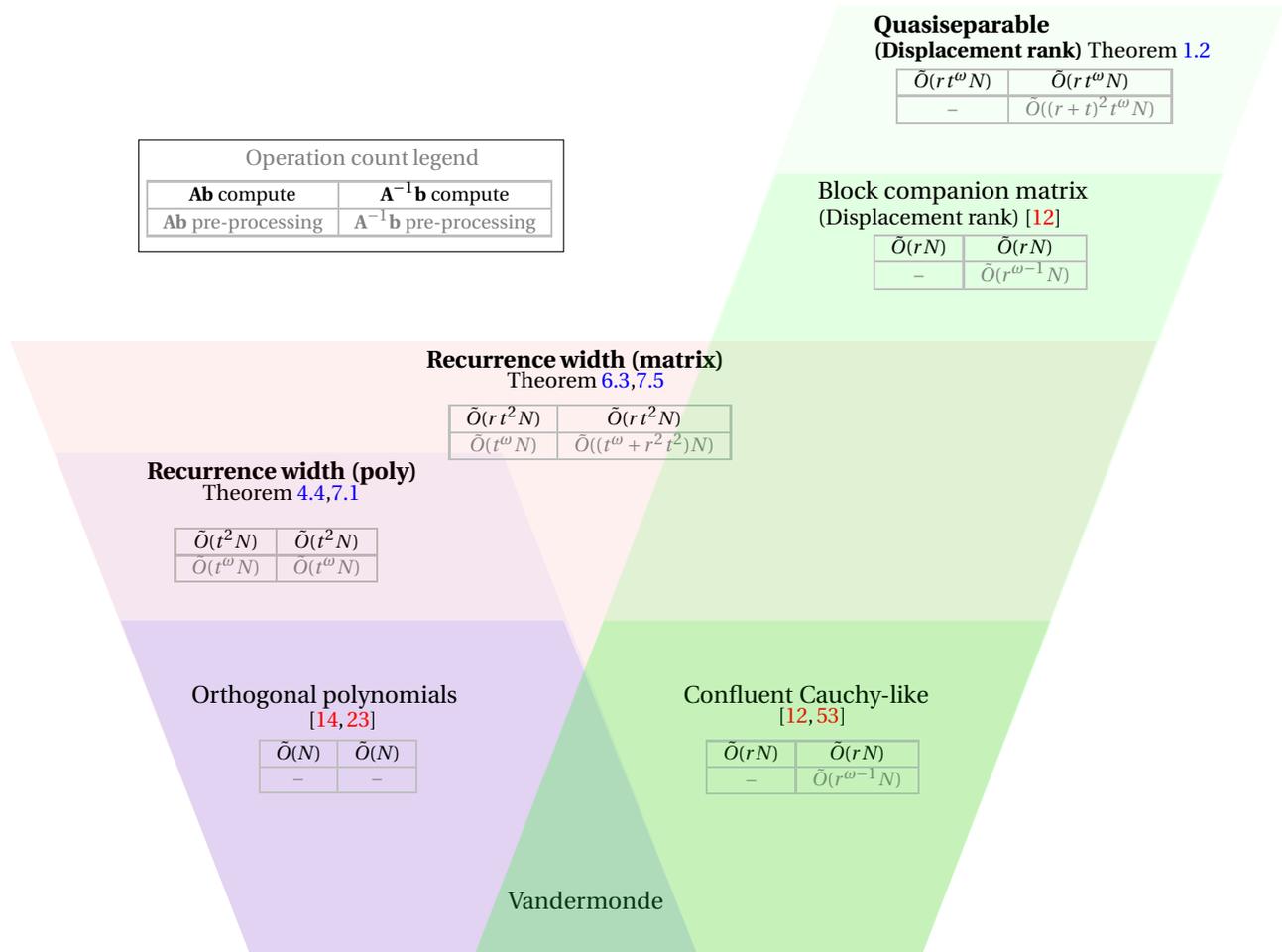

We now focus on the two strands of work that inspired much of our work, and describe how we capture (and generalize) previous results in these areas. These strands have been well-studied and have rich and varied applications, which we summarize at the end of the section. Throughout this section we describe square matrices for simplicity, but we emphasize that the recurrence width concept applies to general matrices (see Definition~\ref{defn:recurrence-width}).

\paragraph{Displacement rank.} The approach of expressing structured matrices via their displacements and using this to define unified algorithms for multiplication and inversion has traditionally been limited to the four most popular classes of Toeplitz-like, Hankel-like, Vandermonde-like, and Cauchy-like matrices\cite{GO94,OP98}. These classic results on displacement rank constrain $\vL,\vR$ to diagonal or shift matrices. Until very recently, the most powerful results on matrix-vector multiplication for matrices with low displacement rank are in the work of Olshevsky and Shokrollahi~\cite{OS00}, who show that any matrix with a displacement rank of $r$ with respect to Jordan form matrices $\vL$ and $\vR$ can be multiplied by an arbitrary vector with $\tO(rN)$ operations. (They use these results and matrices to solve the Nevalinna-Pick problem as well as solve the interpolation step in some list decoding algorithms in a previous work~\cite{OS99}.) Recall that Jordan normal form matrices are a special case of $2$-{\em band upper triangular matrices}. In this work, by applying the recurrence width concept (Definition~\ref{defn:intro-width}), we show that when both $\vL$ and $\vR$ are any triangular $t$-band matrices, matrices with displacement rank of $r$ with respect to them admit fast multiplication. Furthermore, in the restricted case of $t=2$ our result matches the previous bound~\cite{OS00}.

In our following theorem and for the rest of the paper, let $\cM(N)$ denote the cost of multiplying two polynomials of degree $N$ over $\F$ (ranging from $N\log N$ to $N\log{N}\log\log{N}$) and $\omega$ be the matrix-matrix multiplication exponent.
\begin{thm}
  \label{thm:intro-dr}
  Let $\vL$ and $\vR$ be triangular $t$-band matrices sharing no eigenvalues, and let $\vA$ be a matrix such that $\vL\vA - \vA\vR$ has rank $r$. Then with $O(t^{\omega}\cM(N)\log{N})$ pre-processing operations, $\vA$ and $\vA^T$ can be multiplied by any vector $\vb$ in $O(rt^2\cM(N)\log{N})$ operations. With $O(t^{\omega}\cM(N)\log{N} + r^2t^2\cM(N)\log^2{N})$ pre-processing operations, $\vA^{-1}$ and $\vA^{-T}$ can be multiplied by any vector in $O(rt^2\cM(N)\log{N})$ operations.
\end{thm}

In March 2017, a pre-print~\cite{bostan2017} generalized the result of~\cite{OS00} in a different direction, to the case where $\vL$ and $\vR$ are block companion matrices.
We show an alternative technique that is slower than Theorem~\ref{thm:intro-dr} by polylog factors, but works for even more general $\vL$ and $\vR$ that subsumes the classes of both Theorem~\ref{thm:intro-dr} and~\cite{bostan2017}:
%The main idea of this algorithm is short and general, and we hypothesize that it should be possible to improve to match state-of-the-art bounds~\cite{bostan2017}. 
%\ar{Albert: Please fill in a para stating the new result and our new result.}

\begin{thm}
  \label{thm:intro-dr-quasi}
  Let $\vL$ and $\vR$ be $t$-quasiseparable matrices and $\vE$ be rank $r$ such that $\vA$ is uniquely defined by $\vL\vA-\vA\vR = \vE$. Then $\vA$ and $\vA^T$ admit matrix-vector multiplication in $O( rt^\omega \cM(N)\log^2{N} + rt^2 \cM(N) \log^3{N})$ operations. Furthermore, letting this cost be denoted $M_t(N,r)$, then with $O((r+t\log{N}) M_t(N, r+t\log{N})\log{N})$ operations, $\vA^{-1}$ and $\vA^{-T}$ also admit matrix-vector multiplication in $M_t(N,r)$ operations.
\end{thm}
The formal definition of quasiseparability is deferred to Section~\ref{subsec:disp-rank-quasi}, but essentially refers to a type of rank-structured matrix where every submatrix not intersecting the diagonal has rank at most $t$~\cite{eidelman1999}.
$t$-quasiseparability refers to a type of rank-structured matrix where every submatrix not intersecting the diagonal has low rank~\cite{eidelman1999} (see Definition~\ref{defn:quasi}).
This class includes both the triangular banded matrices in Theorem~\ref{thm:intro-dr} which are $t$-quasiseparable, and the block companion matrices of~\cite{bostan2017} which are $1$-quasiseparable. Other well-known classes of matrices with displacement structure include Toeplitz-plus-Hankel-like matrices, all variants of Discrete Cosine/Sine Transforms, Chebyshev-Vandermonde-like matrices, and so on~\cite{olshevsky2003}, and the displacement operators for all of these are tridiagonal which are $1$-quasiseparable.\footnote{Some variants allow the top right and bottom left corner elements to be non-zero; these are still $2$-quasiseparable.} To the best of our knowledge, Theorem~\ref{thm:intro-dr-quasi} is the most general result on near-linear time matrix vector multiplication for matrices with displacement rank structure, and in particular captures all currently known matrices with such structure.

Aside from strongly extending these previous important classes of structured matrices, this class of matrices is well-behaved and nice in its own right. They naturally inherit traits of displacement rank such as certain multiplicative and inversion closure properties~\cite{pan2002structured}. Furthermore, the class of quasiseparable matrices, which was introduced relatively recently and are still under active study~\cite{eidelman1999}, generalize band matrices and their inverses and satisfy properties such as gradation under addition and multiplication\footnote{For example, the sum of a $p$-quasiseparable matrix and $q$-quasiseparable matrix is $(p+q)$-quasiseparable.} and closure under inversion. As one consequence, the Sylvester and Stein displacement classes with respect to quasiseparable operators are essentially identical.\footnote{A matrix has low Sylvester displacement with respect to $\vL,\vR$ if and only if it has low Stein displacement with respect to $\vL^{-1},\vR$.}

\paragraph{Orthogonal polynomial transforms.} The second strand of work that inspired our results relates to \emph{orthogonal polynomial transforms}~\cite{cheb-approx, jacobi, zernike, chihara}. Any matrix $\vA$ can be represented as a polynomial transform, defined as follows. %The transforms can be expressed as matrix-vector multiplication involving matrices whose entries are the evaluations of a family of orthogonal polynomials.
\begin{defn}
  \label{defn:polynomial-transform}
  Let $a_0(X), \dots, a_{N-1}(X)$ be a collection of polynomials over a field $\F$ and $z_0, \dots, z_{N-1}$ be a set of points. The discrete polynomial transform matrix $\vA$ is defined by $\vA[i,j] = a_i(z_j)$. The discrete polynomial transform of a vector $\vb$ with respect to the $a_i$ and $z_j$ is given by the product $\vA\vb$.
\end{defn}
%In particular, the $i^{th}$ row of $\vA$ contains the evaluation at points $\{z_1, \dots, z_n\}$ of the $i^{th}$ polynomial $f_i(X)$.
When the $a_i$ are a family of orthogonal polynomials~\cite{chihara}, we are left with an orthogonal polynomial transform. Orthogonal polynomials can be characterized by the following two term recurrence\footnote{This is often called a three-term recurrence, but we will consider the number of terms on the RHS to be the recurrence length.}:
\begin{equation}
\label{eq:intro-recur-op}
a_{i}(X)= (\alpha_iX+\beta_i)a_{i-1}(X)+\gamma_ia_{i-2}(X),
\end{equation}
where $\alpha_i,\beta_i,\gamma_i\in\F$. Driscoll, Healy and Rockmore present an algorithm to perform the orthogonal polynomial transform over $\F=\R$ in $O(N\log^2{N})$ operations~\cite{driscoll}. Later it was shown how to perform the transposed transform and inverse transform (i.e.\ multiplication by $\vA^{T}$ and $\vA^{-1}$) with the same complexity~\cite{potts1998fast, bostan2010op}. We observe that this class of transforms can be extended to polynomials that satisfy a more general recurrence. We introduce the notion of recurrence width which describes the degree of this type of structure in a matrix:
\begin{defn}
  \label{defn:intro-width}
  An $N\times N$ matrix $\vA$ has \width\ $t$ if the polynomials $a_i(X) = \sum_{j=0}^{N-1} \vA[i,j]X^j$ satisfy $\deg(a_i) \leq i$ for $i < t$, and
  \begin{equation}
    \label{eq:intro-recur-poly}
    a_{i}(X)=\sum_{j=1}^t g_{i,j}(X)a_{i-j}(X)
  \end{equation}
  for $i \geq t$, where the polynomials $g_{i,j}\in\F[X]$ have degree at most $j$.
\end{defn}
%\agu{perhaps should just introduce error terms here since that is what our basic algorithm does}
%If this is the smallest recurrence that the polynomial family satisfies, then the matrix $\vA$ whose rows contain the coefficients of these polynomials has a \width\ of $t$.

%This is the most basic notion of recurrence width and can be kept in mind as a prototypical example. 
Note that under this definition, an orthogonal polynomial transform matrix has recurrence width $2$.\footnote{There is a subtlety in that the orthogonal polynomial transform in Definition~\ref{defn:polynomial-transform} can be factored as the product of the matrix in Definition~\ref{defn:intro-width} times a Vandermonde matrix on the $z_j$. This is covered by a generalization of Definition~\ref{defn:intro-width}, see Definition~\ref{defn:recurrence-width} and Remark~\ref{rmk:rw-examples}.}
%Definition~\ref{defn:intro-width} is a special case of a more general definition (\ref{defn:recurrence-width}) that captures the other strands of work and applications, .
The recurrence structure allows us to generate matrix vector multiplication algorithms for matrices $\vA^T$ and $\vA^{-1}$ (and for $\vA$ and $\vA^{-T}$ by the transposition principle) in a simple and general way. %Promisingly, our algorithms are optimal in the sense of matching the size of its parametrization, up to poly-log factors.

%\ar{State that we also get $\vA^{-1}$ and $\vA^{-T}$ in the Thm statement below. We need to play up the fact that at least for this setting we get everything with pretty much the same ideas.}
%\agu{I am unsure how to implement the second part. We have repeatedly emphasized that $\vA\vb$ follows from the transposition principle; to me the story treats it more as a reduction than ``the same idea''. If you have ideas for implementing this, feel free to modify the text directly.}
\begin{thm}
  \label{thm:intro-width}
  Let $\vA \in \F^{N \times N}$ be a matrix with recurrence width $t$. With $O(t^\omega \cM(N)\log{N})$ pre-processing operations, %(where $\omega$ is the matrix-matrix multiplication exponent), 
the products $\vA^T\vb$ and $\vA\vb$ can be computed in $O(t^2\cM(N)\log{N})$ operations for any vector $\vb$, and the products $\vA^{-1}\vb$ and $\vA^{-T}\vb$ can be computed in $O(t^2\cM(N)\log^2{N})$ operations for any vector $\vb$.
\end{thm}
In Section~\ref{sec:transpose} we provide an algorithm for $\vA^T\vb$, and bound its complexity in Theorem~\ref{thm:mainATb}. The statement for $\vA\vb$ then follows from the transposition principle (and we provide more detail in Appendix~\ref{sec:Ab}). The statement for $\vA^{-1}\vb$ (hence $\vA^{-T}\vb$) is proved in Theorem~\ref{thm:recur-inverse}. %We remark that these algorithms all essentially follow from the same main ideas.
Our algorithms are optimal in the sense that their complexity is equal to the worst-case input size of $\Theta(t^2N)$ for a matrix with \width\ $t$, up to log factors (and if $\omega=2$, so is the pre-processing). In particular, we recover the bounds of Driscoll et al. of $O(\cM(N)\log{N})$ in the orthogonal polynomial case~\cite{driscoll}.

\paragraph{The connection.} Theorem~\ref{thm:intro-dr} and~\ref{thm:intro-width} are special cases of our results on the most general notion of recurrence width. This connection is compelling because the set of matrices with low \width\ and those with low displacement rank seem to be widely different. Indeed the existing algorithms for the class of orthogonal polynomials~\cite{driscoll} and low displacement rank~\cite{OS00} look very different. Specifically, the algorithm of Driscoll et al.~\cite{driscoll} is a divide and conquer algorithm, while that of Olshevsky and Shokrollahi~\cite{OS00} (and other works) heavily exploits structural algebraic properties on matrices with low displacement rank. Despite this, we show that both of these classes of matrices can be captured by a more abstract recurrence and handled with a single class of superfast algorithms.
Our definition is arguably more natural since it centers around polynomials, while existing structured matrices are often defined without polynomials yet use them heavily in the algorithms. Thus, our definition makes this connection more explicit.
%In particular, it turns out that if equation~\eqref{eq:intro-recur-poly} is extended to use {\em error polynomials}, then this notion of recurrence also includes the matrices of Theorem~\ref{thm:intro-dr}.
%In Section~\ref{sec:prelim} we present the more abstract recurrence that captures both of these classes of matrices, and present efficient algorithms that work for these general recurrences.

Similarly, both classic matrices with displacement structure~\cite{P01} and those with rank structure~\cite{EGH13} have large but separate bodies of literature, and Theorem~\ref{thm:intro-dr-quasi} shows that their structure can be combined in a way that still admits efficient algorithms. We believe that unifying these existing threads of disparate work is interesting in its own right and our most important contribution. In Appendix~\ref{sec:related}, we provide a more detailed history of these threads.

\paragraph{Some Applications.}
Orthogonal polynomials and low displacement rank matrices have applications in numerous areas from signal processing to machine learning. 
Indeed orthogonal polynomials have their own dedicated conference~\cite{OP-conf}. For matrices with low displacement rank, the survey by Kailath and Sayed~\cite{disp-rank-survey} covers more details and applications, which our matrices naturally inherit.
%\agu{can try to put a line about how every displacement rank equation is closely tied to some interpolation type problem}
%\ar{We need to make some noise on how this is really crucial for practical applications - Chris can add here}
The matrix structure we discuss is also related to notions in control theory arising from different motivations. The notion of displacement rank is a special case of the Sylvester equation which arises in the stability analysis of linear dynamical systems~\cite{simoncini2016computational}. Additionally, the more general type of recurrence width matrices we study can be viewed as satisfying a type of higher-order Sylvester equation which has become increasingly more important in this area. Such equations are very difficult to analyze in full generality~\cite{simoncini2016computational}, and our notions represent one line of attack for efficient solutions to these problems.

As a more concrete application, there has been recent interest in the use of structured matrices as components of neural networks, intended to replace the expensive fully-connected layers (which are essentially linear transformations represented as matrix-vector products) with fast and compressible structures. Many ad-hoc paradigms exist for representing structured matrices in neural networks~\cite{moczulski2015acdc}, and it turns out several of them are instances of a displacement rank structure. More explicitly, a recent paper by Sindhwani et al. found that Toeplitz-like matrices (which have low displacement rank with respect to shift matrices) were much more effective than standard low-rank encodings for mobile speech recognition~\cite{sindhwani2015structured}. Additionally, the theoretical guarantees of using displacement rank to compress neural networks are beginning to be understood~\cite{zhao2017theoretical}. For the same number of parameters, our generalized classes of dense, full-rank structured matrices can be even more expressive than the standard structures explored so far, so they may be suitable for these applications.

\paragraph{Other results and paper organization.} Finally, we collect some other results that support our main results above. In Section~\ref{sec:hierarchy}, we show that recurrence width forms a hierarchy, i.e.\ the matrices of width $t$ cannot describe those of width $t+1$. Next, we show that our new generalizations capture matrices from coding theory that were not captured by the earlier classes. In particular, we show that the matrix corresponding to multipoint multivariate polynomial evaluation has small recurrence width. In Section~\ref{sec:multi}, we use this connection to show a {\emph barrier} result: if one could design algorithms that are efficient enough in terms of sparsity of the input, then we would improve the state-of-the-art results in multipoint multivariate polynomial evaluation. We also note that our matrices capture certain types of sequences, and show how to use the general purpose algorithm to compute them; in Appendix~\ref{sec:bernoulli}, we show how to compute the first $N$ Bernoulli numbers in $O(N\log^2 N)$ operations, which recovers the same algorithmic bit complexity (potentially with a log factor loss) as the best algorithms that are specific to Bernoulli numbers~\cite{bernoulli}. We show some preliminary experimental results in Appendix~\ref{sec:experiments}, which support the claim that our algorithm is simple and the constant factors are small.

Our paper is organized as follows. In Section~\ref{sec:technical-overview}, we provide an overview of our techniques.
In Section~\ref{sec:prelim}, we formally define the most general notion of recurrence width and describe the main structural properties needed for our algorithms.
In Section~\ref{sec:transpose}, we describe the full algorithm for $\vA^T\vb$ for a restricted class of recurrence width; the algorithm for this subclass contains all the core ideas and in particular is already more general than Definition~\ref{defn:intro-width}.
In Section~\ref{sec:general-mult}, we describe the modifications to the $\vA^T\vb$ algorithm needed to handle the fully general definition of recurrence width. We also elaborate on related multiplication operations, including multiplication by $\vA$, multiplication by Krylov matrices, and an optimization for matrix-matrix multiplication.
In Section~\ref{sec:disp-rank} we provide details on the multiplication algorithms for matrices with low displacement rank. We first prove Theorem~\ref{thm:intro-dr} by showing that those matrices have low recurrence width, and then adapt the techniques to cover the more general displacement structure of Theorem~\ref{thm:intro-dr-quasi}.
In Section~\ref{sec:inverse}, we provide details on computing inverses or solvers for several types of matrices of low recurrence width.
In Section~\ref{sec:other}, we include other results and applications of recurrence width. These include how to recover the parameterization of recurrence width $t$ matrices (Section~\ref{subsec:fitting}), that recurrence width forms a hierarchy (Section~\ref{sec:hierarchy}), and applications to coding theory (Section~\ref{sec:multi}).

\section{Technical Overview}
\label{sec:technical-overview}

%\ar{Still have to make a pass on this subsection.}

For ease of exposition, we start off with an overview of our techniques for low recurrence width matrices. We describe the basic recurrence width and a natural generalization that captures displacement rank with respect to triangular band matrices. Finally we show how to adapt these techniques to handle matrices with more general displacement structure.

As in all algorithms for structured matrices, the main challenge is in manipulating the alternate compact representation of the matrix directly. Although the recurrence~\eqref{eq:intro-recur-poly} is represented with only $O(N)$ total values (for a fixed $t$), the polynomials $a_i(X)$ it produces are large: unlike conventional linear recurrences on scalars, the total output size of this recurrence is quadratic because the polynomial degrees grow linearly (hence why the matrix produced is dense). The first hurdle we clear is compressing this output using the recurrence structure~\eqref{eq:intro-recur-poly}, which lets us characterize the polynomials $a_i(X)$ in a simple way and leads to an intuitive algorithm for transpose multiplication. At its core, the algorithms exploit the self-similar structure of the recurrence~\eqref{eq:intro-recur-poly}. By definition, all the polynomials $a_i(X)$ are generated from $a_0(X), \dots, a_{t-1}(X)$. But due to the recurrence structure, the higher-degree polynomials $a_{N/2}(X)$ through $a_{N-1}(X)$ can be thought of as being generated from $a_{N/2}(X), \dots, a_{N/2+t-1}(X)$ (and the dependence on these generators is now lower-degree). This forms the basis of the divide-and-conquer algorithms.

We then consider a natural generalization consisting of recurrences on vectors, where the transitions involve multiplying by polynomial functions of a fixed matrix. We show a reduction to the standard polynomial case; this allows us to interpret structures such as the displacement rank equation~\eqref{eq:intro-dr}, which does not seem to involve polynomials, as defining a polynomial recurrence. This reduction is represented through {\em Krylov matrices}, which appear in independently interesting settings on numerical methods such as the Lanczos eigenvalue algorithm and Wiedemann's kernel vector algorithm~\cite{trefethen, wiedemann}.
For the cases we are interested in, these Krylov matrices themselves turn out to have low recurrence width.

%Finally, we apply our characterization to address other operations such as reconstructing the parameters of a low recurrence-width matrix from its entries, 

%These generalizations can be unified by factoring out terms that capture the evaluation of certain matrices at polynomials or rational functions. We represent these as {\em Krylov matrices}, which appear in independently interesting settings on numerical methods such as the Lanczos eigenvalue algorithm and Wiedemann's kernel vector algorithm~\cite{trefethen, wiedemann}.
%The problem of evaluating a function of a matrix is itself a well-studied and useful problem~\cite{higham2008}, and certain strong representations of the matrix (such as the SVD or eigendecomposition) often play a central role. In contrast,
%We show that for the cases we are interested in, these Krylov matrices themselves have low recurrence width.

\paragraph{Transpose multiplication.}

%We achieve this with an alternate representation of these recurrences as inverses of triangular banded matrices of polynomials.

We will consider a more general form of the recurrence~\eqref{eq:intro-recur-poly} that allows some generators to be added at every step. It turns out that this generalization does not increase the asymptotic cost of the algorithm but will be necessary to capture structure such as displacement rank~\eqref{eq:intro-dr}. Suppose that the following recurrence holds for all $i$
\begin{equation}
  \label{eq:intro-recur-error}
  a_{i}(X)=\sum_{j=1}^{\min(i,t)} g_{i,j}(X)a_{i-j}(X) + \sum_{k=0}^{r-1} c_{i,k} d_k(X),
\end{equation}
for some fixed \emph{generators} $d_0(X), \dots, d_{r-1}(X)$.
Notice that equation~\eqref{eq:intro-recur-poly} is a special case with $r=t$, $d_i(X)=a_i(X)$ for $0 \le i < t$, and $c_{i,k}=\delta_{ik}$ for $0 \le i < n, 0 \le k < r$ (i.e.\ the $t$ initial polynomials are generators and are never re-added).

Consider the simplified case when $r=1$ and $d_0(X) = 1$, so that the recurrence becomes $a_{i}(X)=\sum_{j=1}^{\min(i,t)} g_{i,j}(X)a_{i-j}(X) + c_i$; this simplification captures the main difficulties. This can be written as
\[
  \begin{bmatrix}
  1 & 0 & 0 & \cdots & 0 \\
  -g_{1,1}(X) & 1 & 0 & \cdots & 0 \\
  -g_{2,2}(X) & -g_{2,1}(X) & 1 & \cdots & 0 \\
  \vdots & \vdots & \vdots & \ddots & \vdots \\
  0 & 0 & 0 & \cdots & 1
  \end{bmatrix}
  \begin{bmatrix} a_0(X) \\ a_1(X) \\ a_2(X) \\ \vdots \\ a_{N-1}(X) \end{bmatrix}
  =
  \begin{bmatrix} c_0 \\ c_1 \\ c_2 \\ \vdots \\ c_{N-1} \end{bmatrix}.
\]
Let this matrix of recurrence coefficients be $\vG$. Notice that computing $\vA^T\vb$ is equivalent to computing the coefficient vector of $\sum \vb[i] a_i(X)$. By the equation above, it suffices to compute $\vb^T \vG^{-1} \vc$.

Thus, we convert the main challenge of understanding the structure of the matrix $\vA$ into that of understanding the structure of $\vG^{-1}$, which is easier to grasp. Notice that $\vG$ is triangular and $t$-banded. The inverse of it is also structured\footnote{Inverses of banded matrices are called semiseparable~\cite{vandebril2005note}.}
and has the property that every submatrix below the diagonal has rank $t$. Thus we can partition $\vG^{-1}$ into $O(\log N)$ structured submatrices; for any such submatrix $\vG'$, we only need to be able to compute ${\vb'}^T \vG' \vc'$ for vectors $\vb',\vc'$. We provide a more explicit formula for the entries of $\vG^{-1}$ that enables us to do this. The pre-computation step of our algorithms, as mentioned in Theorem~\ref{thm:intro-width}, essentially corresponds to computing a representation of $\vG^{-1}$ via generators of its low-rank submatrices. This is formalized in Section~\ref{subsec:structure}.
We note that this algorithm automatically induces an algorithm for $\vA\vb$ with the same time complexity, by the transposition principle. % (Proposition~\ref{prop:transposition}). 
%Furthermore, it turns out that the transpose of our general $\vA^T\vb$ algorithm, for the restricted setting of Driscoll, recovers their original algorithm for $\vA\vb$. However, the direction $\vA^T\vb$ seems more natural and intuitive.

\paragraph{Matrix Recurrences, Krylov Efficiency and Displacement Rank.}
Equation~\eqref{eq:intro-recur-poly} can be thought of as equipping the vector space $\F^N$ with a $\F[X]$-module structure and defining a recurrence on vectors $\va_i = \sum g_{i,j}(X)\va_{i-j}$. In general, a $\F[X]$-module structure is defined by the action of $X$, which can be an arbitrary linear map. This motivates our most general recurrence:
\begin{equation}
  \label{eq:intro-recur-matrix}
  \va_{i}=\sum_{j=1}^{\min(i,t)} g_{i,j}(\vR)\va_{i-j} + \sum_{k=0}^{r-1} c_{i,k} \vd_k,
\end{equation}
where the row vectors are governed by a \emph{matrix recurrence}.
This structure naturally captures polynomial recurrence structure~\eqref{eq:intro-recur-poly} and~\eqref{eq:intro-recur-error} (when $\vR$ is the shift matrix, i.e.\ $1$ on the main subdiagonal), low rank matrices (when the recurrence is degenerate, i.e.\ $t=0$), and displacement rank, which we show next.

Consider a matrix $\vA$ satisfying equation~\eqref{eq:intro-dr} for lower triangular $(t+1)$-banded $\vL$ and $\vR$ and $\rank(\vE) = r$. By the rank condition, each row of $\vE$ can be expanded as a linear combination of some $r$ fixed vectors $\vd_k$, so the $i$th row of this equation can be rewritten as
\begin{align}
    &\sum_{j=0}^t \vL[i,i-j]\vA[i-j,:] - \vA[i,:]\vR = \vE[i,:], \nonumber \\
  &\vA[i,:](\vL[i,i]\vI-\vR) = \sum_{j=1}^t -\vL[i,i-j]\vA[i-j,:] + \sum_{k=0}^{r-1} c_{i,k} \vd_k.
  \label{eq:intro-recur-dp}
\end{align}
Notice the similarity to equation~\eqref{eq:intro-recur-error} if $\vA[i,:]$ is replaced with $a_i(X)$ and $\vR$ is $X$. In fact, we show that the matrix $\vA$ can be decomposed as $\sum_1^r \vA_i \vK_i$, where $\vA_i$ are matrices of recurrence width $t$ and $\vK_i$ are \emph{Krylov matrices on $\vR$}.\footnote{The $i$th column of the Krylov matrix on $\vR$ and $\vx$ is $\vR^i \vx$.} We remark that this form generalizes the well-known \emph{$\mathit{\Sigma}LU$ representation} of Toeplitz-like matrices, which is used in multiplication and inversion algorithms for them~\cite{KKM79}.

Recurrence~\eqref{eq:intro-recur-matrix} involves evaluating functions at a matrix $\vR$.
%; we focus on triangular $t$-banded matrices which is sufficient for the displacement rank application.
%The question of evaluating a function of a matrix is itself a well-studied and useful problem~\cite{higham2008}.
Classical ways of computing these matrix functions use natural decompositions of the matrix, such as its singular value/eigendecomposition, or Jordan normal form $\vR = \vA\vJ\vA^{-1}$. 
%The evaluation of an analytic function then becomes $f(\vR) = \vA f(\vJ)\vA^{-1}$, which admits superfast multiplication if each component does. We note that the Jordan decomposition is generally hard to compute and much work has been done on computing specific matrix functions without going through the Jordan form~\cite{putzer1966}. Despite this, 
In Appendix~\ref{subsec:jordan} we show that it is possible to compute the Jordan decomposition quickly for several special subclasses of the matrices we are interested in, using techniques involving low-width recurrences.
%\ar{I would recommend moving most of the above para to the related work section.}
%\agu{I can't figure out how to cut more while preserving the story Chris likes}
%This type of decomposition not only solves our problem of reducing multiplication by a low displacement-rank matrix to multiplication by a low-width matrix, but a variety of other problems involving triangular banded matrices and matrix functions. For example, perhaps the most prominent example of a matrix function is the matrix exponential, which is tied to the theory of differential equations: The basic linear system $\vx'(t) = \vM\vx(t)$ is solved by the matrix exponential $\vx(t) = e^{t\vM}\vx(0)$~\cite{hartman1960}. 

However,~\eqref{eq:intro-recur-matrix} has more structure and only requires multiplication by the aforementioned Krylov matrices.
%is easier than evaluating general matrix functions.
%Consider again the simplified case when $t=0$. Fixing the $g_{i,j}(X)$, we can define the polynomials $f_i(X)$ by~\eqref{eq:intro-recur-poly}. Note that $\vf_i = g_{i-1,0}(\vR)\cdots g_{0,0}(\vR)\vf_0 = f_i(\vR)\vf_0$. So we can factor $\vA = \vA'\vK$, where $\vA'$ is the coefficient matrix of the basic polynomial recurrence (Definition~\ref{defn:intro-width}), and $\vK$ is the matrix whose $i$-th column is $\vR^i \vf_0$. Thus this reduces to the basic recurrence~\eqref{eq:intro-recur-poly} if we can multiply by $\vK$, the {\em Krylov matrix} on $\vR$ and $\vf_0$.\footnote{The image of this matrix is the Krylov subspace on $\vR$ and $\vf_0$.} We say that $\vR$ is {\em Krylov efficient} if all of its Krylov matrices admit superfast multiplication.
%
In Section~\ref{sec:Krylov}, we show that the Krylov matrix itself satisfies a recurrence of type~\eqref{eq:intro-recur-error} of width $t$.
%the rows of the Krylov matrix, when interpreted as polynomials over $X$, satisfy a recurrence $\pmod{X^N}$, with the recurrence width equal to the bandwidth of the matrix $\vR$.
Using our established results, this gives a single $O(t^2 \cM(N)\log{N})$ algorithm that unifies the aforementioned subclasses with Jordan decompositions, and implies all triangular banded matrices are \emph{Krylov efficient} (Definition~\ref{defn:krylov}).%and solves Krylov efficiency for all triangular banded matrices with one algorithm.

%\ar{Albert: Please put in the technical text on the semi-separable stuff. For the time being do not worry too much about flow with rest of the intro. We'll smoothen it in later.}

When $\vR$ is not banded, the Krylov matrices do not have low recurrence width, but they can still be structured. The techniques of Section~\ref{sec:Krylov} show that in general, multiplying by a Krylov matrix on $\vR$ is can be reduced to a computation on the \emph{resolvent} of $\vR$. Specifically, it is enough to be able to compute the rational function $\vb^T (\vI-X\vR)^{-1}\vc$ for any $\vb,\vc$. This reduction bears similarity to results from control theory about the Sylvester equation~\eqref{eq:intro-dr} implying that manipulating $\vA$ can be done through operating on the resolvents of $\vL$ and $\vR$~\cite{simoncini2016computational}. In the case of multiplication by $\vA$, it suffices to be able to solve the above resolvent multiplication problem. In Section~\ref{subsec:disp-rank-quasi}, we show how to solve it when $\vR$ is quasiseparable, by using a recursive low-rank decomposition of $\vI-X\vR$. %Informally, quasiseparable matrices are a type of rank-structured matrix where every sub-matrix not intersecting the diagonal has low rank.

\section{Problem Definition}
\label{sec:prelim}

\subsection{Notation}

%\ar{Do we use $\cR$ in any of our applications? If not maybe we can just state everything for fields and put in a footnote saying we can extend our stuff to commutative-rings.}
We will use $\F$ to denote a field and use $\R$ and $\C$ to denote the field of real and complex numbers respectively.
%\footnote{It can be assumed that $\F$ is $\R$ or $\C$ because our applications only use these fields, but all results can be extended to general fields $\F$ and also commutative rings. We also assume that $\F$ supports the FFT; otherwise, the runtime bounds incur an extra $\log\log{N}$ factor.} The set of polynomials over $\F$ and set of rational functions over $\F$ will be denoted by $\F[X]$ and $\F(X)$ respectively.
%\ar{Are we using rational recurrences anymore? Or we down to just polynomials because of the mod business? If so, we might want to drop the rational stuff and put in a footnote saying when we can handle rationals.}
%\agu{We are still using rational functions because the coefficients of our mod recurrences are represented using a fraction; otherwise their size is too large}

\noindent
\textbf{Polynomials.}
For polynomials $p(X),q(X),s(X) \in \F[X]$, we use the notation $p(X) \equiv q(X) \pmod{s(X)}$ to indicate equivalence modulo $s(X)$, i.e.\ $s(X) | (p(X)-q(X))$, and $p(X) = q(X) \pmod{s(X)}$ to specify $q(X)$ as the unique element of $p(X)$'s equivalence class with degree less than $\deg{s(X)}$. 
%We will sometimes consider polynomials over multiple indeterminates; if $f(X_1,\dots,X_k) \in \F[X_1,\dots,X_k]$, we use $\deg(f)$ to mean the maximum degree of its monomials (i.e.\ sum of powers of all indeterminates in a term), and $\deg_{X_i}(f)$ to denote the degree of $f$ when viewed as a polynomial over $X_i$.
We use $\cM(N)$ to denote the time complexity of multiplying two polynomials of degree $N$ over the field $\F$ (and is known to be $O(N\log{N}\log\log{N})$ in worst-case). We will use $\tO(T(N))$ to denote $O\left(T(N)\cdot (\log{T(N)})^{O(1)} \right)$. 

\noindent
\textbf{Vectors and Indexing.}
For any integer $m\ge 1$, we will use $[m]$ to denote the set $\{0,\dots,m-1\}$. Unless specified otherwise, indices in the paper start from $0$. Vectors are boldset like $\vx$ and are column vectors unless specified otherwise. We will denote the $i$th element in $\vx$ by $\vx[i]$ and the vector between the positions $[\ell,r):\ell\le r$ by $\vx[\ell:r]$. For any subset $T\subseteq [N]$, $\ve_T$ denotes the characteristic vector of $T$. We will shorten $\ve_{\{i\}}$ by $\ve_i$. Sometimes a vector $\vb$ is associated with a polynomial $b(X)$ and this mapping is always $b(X) = \sum_{i \ge 0} \vb[i] X^i$ unless specified otherwise.

\noindent
\textbf{Matrices.}
Matrices will be boldset like $\vM$ and by default $\vM\in \F^{N\times N}$. We will denote the element in the $i$th row and $j$th column by $\vM[i,j]$ ($\vM[0,0]$ denotes the `top-left' element of $\vM$). $\vM[\ell_1:r_1,\ell_2:r_2]$ denotes the sub-matrix $\{\vM[i,j]\}_{\ell_1\le i < r_1, \ell_2\le j < r_2}$. In particular we will use $\vM[i,:]$ and $\vM[:,j]$ to denote the $i$th row and $j$th column of $\vM$ respectively. 
We will use $\vS$ to denote the shift matrix (i.e.\ $\vS[i,j]=1$ if $i=j+1$ and $0$ otherwise) and $\vI$ to denote the identity matrix. We let $c_{\vM}(X) = \det(X\vI-\vM)$ denote the characteristic polynomial of $\vM$.
%We will use $\vV$ to denote Vandermonde matrices, defined in Definition~\ref{defn:vandermonde}.
Given a matrix $\vA$, we denote its transpose and inverse (assuming it exists) by $\vA^T$ (so that $\vA^T[i,j]=\vA[j,i]$) and $\vA^{-1}$ (so that $\vA\cdot\vA^{-1}=\vI$). We will denote $(\vA^T)^{-1}$ by $\vA^{-T}$.

%We will denote the {\em size} of a polynomial $p(X)\in\F[X]$ as the degree of $p(X)$. %The size of a fraction $g(X)=\frac{a(X)}{b(X)}$, denoted by $\|g\|$ is the sum of sizes of $a(X)$ and $b(X)$.
%We will call a set of fractions $S\subseteq \F(X)$ to be {\em nice} if for any $g_1(X),g_2(X)\in S$, we have
%\begin{enumerate}
%\item $\| g_1+g_2\| \le \max(\|g_1\|,\|g_2\|)$. Further $g_1(X)+g_2(X)$ can be computed in $\tO(\|g_1+g_2\|)$ operations over $\F$.
%\item $\| g_1\cdot g_2\| \le \|g_1\|+\|g_2\|$. Further $g_1(X)\cdot g_2(X)$ can be computed in $\tO(\|g_1\cdot g_2\|)$ operations over $\F$.
%\end{enumerate}
%\agu{Manas points out that isn't 2 always true?}
%We note that any subset of $\F(X)$, where all the fractions have the same denominator (e.g.\ $\F[X]$), is nice.

Finally, we define the notion of Krylov efficiency which will be used throughout and addressed in Section~\ref{sec:Krylov}.
\begin{defn}
  \label{defn:krylov}
  Given a matrix $\vM\in \F^{N\times N}$ and a vector $\vy\in\F^{N}$, the {\em Krylov matrix of $\vM$ generated by $\vy$} (denoted by $\kry(\vM,\vy)$) is the $N\times N$ matrix whose $i$th column for $0\le i<N$ is $\vM^i\cdot \vy$. We say that $\vM$ is $(\alpha,\beta)$-{\em Krylov efficient} if for every $\vy\in\F^N$, we have that $\vK=\kry(\vM,\vy)$ admits the operations $\vK\vx$ and $\vK^T\vx$ (for any $\vx\in\F^N$) with $O(\beta)$ many operations (with $O(\alpha)$ pre-processing operations on $\vM$).
\end{defn}

Throughout this paper, we assume some well-known facts that are summarized in Appendix~\ref{subsec:known}.

\subsection{Our Problem}
\label{sec:our-problem}
We address structured matrices satisfying the following property.
\begin{defn}
  \label{defn:recurrence-width}
  A matrix $\vA \in \F^{M\times N}$ satisfies a \emph{$\vR$-recurrence} of \emph{width $(t,r)$} if its row vectors $\va_i = \vA[i,:]^T$ satisfy
  \begin{equation}
    \label{eq:recur}
    g_{i,0}(\vR)\va_{i} = \sum_{j=1}^{\min(t,i)} g_{i,j}(\vR)\va_{i-j} + \vf_i
  \end{equation}
  for some polynomials $g_{i,j}(X) \in \F[X]$, and the matrix $\vF\in \F^{M \times N}$ formed by stacking the $\vf_i$ (as row vectors) has rank $r$. We sometimes call the $\vf_i$ \emph{error terms} and $\vF$ the \emph{error matrix}, reflecting how they modify the base recurrence.

  Note that $\vR \in \F^{N \times N}$ and we assume $g_{i,0}(\vR)$ is invertible for all $i$.

  The recurrence has \emph{degree $(d, \bar{d})$} if $\deg(g_{i,j}) \leq dj + \bard$.

\end{defn}

For convenience in describing the algorithms, we assume that $M,N$ and $t$ are powers of $2$ throughout this paper, since it does not affect the asymptotics.

\begin{rmk}
    \label{rmk:rw-examples}
We point out some specific cases of importance, and typical assumptions.
\begin{enumerate}
  \item For shorthand, we sometimes say the width is $t$ instead of $(t,r)$ if $r \leq t$. In particular, the basic polynomial recurrence~\eqref{eq:intro-recur-poly} has width $t$. In Section~\ref{subsec:recur-error}, we show that every rank $r$ up to $t$ has the same complexity for our multiplication algorithm.

  \item When $\bard = 0$, we assume that $g_{i,0} = 1$ for all $i$ and omit it from the recurrence equation. %This can be assumed as a prototypical example throughout this paper; the general case is not much more complicated.

  \item The orthogonal polynomial transform in Definition~\ref{defn:polynomial-transform} has width $(2,1)$ and degree $(1,0)$. More specifically, $\vR=\diag(z_0,\dots,z_{N-1})$ and $\vF = \left[\begin{array}{c|c|c} \ve_1 & \cdots & \ve_1 \end{array}\right]$. In Section~\ref{subsec:recur-matrix} we show that a Krylov matrix $\kry(\vR,\vone)$ can be factored out, leaving behind a matrix satisfying a polynomial recurrence~\eqref{eq:intro-recur-poly}. Thus we think of Definition~\ref{defn:intro-width} as the prototypical example of recurrence width.
    %\ar{Shouldn't $\vF$ be all zeros?}
    %\agu{I don't believe so under this definition. For example, $\va_0 = \vf_0$. Since $\va_0$ is the all $1$s vector for Driscoll (since it is the evaluations of the first polynomial which is just $1$), so if $\vf_0$.}

  \item When $\vR$ is a companion matrix
    \[
      \begin{bmatrix} 0 & 0 & \cdots & 0 & m_0 \\ 1 & 0 & \cdots & 0 & m_1 \\ 0 & 1 & \cdots & 0 & m_2 \\ \vdots & \vdots & \ddots & \vdots & \vdots \\ 0 & 0 & \cdots & 1 & m_{N-1} \end{bmatrix}
    \]
    corresponding to the {\em characteristic polynomial} $c_\vR(X) = X^N - m_{N-1} X^{N-1} - \dots - m_0$, the recurrence is equivalent to interpreting $\va_i,\vf_i$ as the coefficient vector of a polynomial and following a polynomial recurrence $\pmod{c_\vR(X)}$. That is, if $a_i(X) = \sum_{j=0}^{N-1} \va_i[j]X^j = \sum_{j=0}^{N-1} \vA[i,j] X^j$ (and analogously for $f_i(X)$), then the $a_i(X)$ satisfy the recurrence
    \begin{equation}
      \label{eq:recur-mod}
      g_{i,0}(X)a_i(X) = \sum_{j=1}^{\min(t,i)} g_{i,j}(X)a_{i-j}(X) + f_i(X) \pmod{c_\vR(X)}
    \end{equation}
    We sometimes call this a \emph{modular recurrence}.

    As an even more special case, the basic polynomial recurrence~\eqref{eq:intro-recur-poly} can be considered an instance of~\eqref{eq:recur} with $\vR = \vS$, the shift matrix.

  \item If $t=0$, then the recurrence is degenerate and $\vA=\vF$. In other words, low rank matrices are degenerate recurrence width matrices.
    %\agu{cite Dummit. remark that it acts as F[X]-module in the vector space of mod M}

\end{enumerate}
\end{rmk}

A matrix defined by Definition~\ref{defn:recurrence-width} can be compactly represented as follows.
Let $\vG \in \F[X]^{M\times M}$ be given by $\vG_{ii} = g_{i,0}(X)$ and $\vG_{ij} = -g_{i,i-j}(X)$ (thus, $\vG$ is zero outside of the main diagonal and $t$ subdiagonals). Let $\vF \in \F^{M \times N}$ be the matrix formed by stacking the $\vf_i$ (as row vectors). Then $\vG,\vF$, and $\vR$ fully specify the recurrence and we sometimes refer to $\vA$ as a $(\vG,\vF,\vR)$-recurrence matrix.

\subsection{The Structure Lemma}
\label{subsec:structure}

In this section, we provide a characterization of the elements $\va_i$ generated by the recurrence in terms of its parameterization $\vG,\vF,\vR$. We further provide a characterization of the entries of this matrix.
We assume for now that we are working with a recurrence of degree $(1,0)$.

\begin{lmm}[Structure Lemma, (i)]
  \label{lmm:structure1}
  Let $\vH = (h_{i,j}(X))_{i,j} = \vG^{-1} \pmod{c_\vR(X)} \in \F[X]^{M \times M}$. Then
  \[
    \va_i = \sum_{j=0}^{M-1} h_{i,j}(\vR) \vf_j
  \]
\end{lmm}
\begin{proof}
  Recall that $\vg_{i,0}(\vR)$ is invertible, so $\vg_{i,0}(X)$ shares no roots with $c_\vR(X)$. Thus $\vG^{-1} \pmod{c_\vR(X)}$ is well-defined, since $\vG$ is triangular and its diagonal elements are invertible $\pmod{c_\vR(X)}$.

  Given a matrix $\vM = (m_{ij}(X))_{i,j}$, let $\vM(\vR)$ denote element-wise evaluation: $\vM(\vR) = (m_{ij}(\vR))$. Note that Definition~\ref{defn:recurrence-width} is equivalent to the equation
  \[
    \vG(\vR) \begin{bmatrix} \va_0 \\ \vdots \\ \va_{M-1} \end{bmatrix} = \begin{bmatrix} \vf_0 \\ \vdots \\ \vf_{M-1} \end{bmatrix}.
  \]
  %\[
  %  \vG(\vR) \left[ \va_0 \cdots \va_{n-1} \right]^T = \left[ \vf_0 \cdots \vf_{n-1} \right]^T.
  %\]
  However, $\vH(\vR)\vG(\vR) = (\vH\vG)(\vR) = \vI$ by the Cayley-Hamilton Theorem. This implies
  \[
    \begin{bmatrix} \va_0 \\ \vdots \\ \va_{M-1} \end{bmatrix} = \vH(\vR) \begin{bmatrix} \vf_0 \\ \vdots \\ \vf_{M-1} \end{bmatrix}.
  \]
\end{proof}

Our algorithms involve a pre-processing step that computes a compact representation of $\vH$. We use the following characterization of the elements of $\vH$.

First, for any $0 \le i < M$, we will define the $t\times t$ transition matrix:
\[\vT_i=
\begin{pmatrix}
0 & 1& \cdots &0  &0 \\
\vdots & \vdots & \cdots &\vdots &\vdots\\
0 & 0 &\cdots &1 &0\\
0 & 0 &\cdots &0 &1\\
g_{i+1,t}(X) & g_{i+1,t-1}(X) & \cdots &g_{i+1,2}(X) &g_{i+1,1}(X)\\
\end{pmatrix}
\]
(let $g_{i,j}(X) = 0$ for $j > i$).
And for any $\ell\le r$, define $\vT_{[\ell:r]}=\vT_{r-1}\times \cdots\times \vT_\ell$ (note that $\vT_{[\ell:\ell]} = \vI_t$).

\begin{lmm}[Structure Lemma, (ii)]
  \label{lmm:structure2}
  $h_{i,j}(X)$ is the bottom right element of \, $\vT_{[j:i]} \pmod{c_\vR(X)}$.
\end{lmm}

Intuitively, $\vT_i$ describes, using the recurrence, how to compute $\va_{i+1}$ in terms of $\va_{i-t+1}, \dots, \va_{i}$, and the ranged transition $\vT_{[j:i]}$ describes the dependence of $\va_{i}$ on $\va_{j-t+1}, \dots, \va_{j}$. The following lemma about the sizes of entries in $\vT_{[\ell:r]}$ will help bound the cost of computing them.
\begin{lmm}
  \label{lmm:T-range-sizes}
  Let the recurrence in~\eqref{eq:recur} have degree $(1,0)$. Then
  for any $0 \leq \ell \leq r < N$ and $0 \leq i,j < t$, %the matrix $\vT_{[\ell:r]} = D_{\ell,r}(X)^{-1}\vT'_{\ell,r}$ where $\deg(D_{\ell,r}) \leq \bar{d}(r-\ell)$ and for all $0\le i,j\le t$,
  \[\deg(\vT_{[\ell:r]}[i,j]) \le \max((r-\ell+i-j),0).\]
\end{lmm}
In particular, it is helpful to keep in mind that $\vH$ satisfies the same degree condition as $\vG$: $\deg(\vH[i,j]) \leq \max(i-j,0)$. Statements of Lemmas~\ref{lmm:structure2} and~\ref{lmm:structure1} for general $(d,\bard)$-degree recurrences, and their proofs (as well as the proof of Lemma~\ref{lmm:T-range-sizes}), are presented in Appendix~\ref{sec:structure}.

\section{Core Multiplication Algorithm}
\label{sec:transpose}
We provide an algorithm for computing $\vA^T\vb$ for a simplified setting that contains all the core ideas. Assume $\vA \in \F^{N \times N}$ is a modular recurrence of width $(t,1)$ and degree $(1,0)$. In other words, $\vA$ is square with three independent simplifications on its generators: $\vG$ has degree $(1,0)$, $\vF$ has rank $1$, and $\vR$ is a companion matrix.\footnote{These assumptions are already more general than the setting in Driscoll et al.~\cite{driscoll}; their setting corresponds to $t=2$, $\vR=\vS$, $\vc = \ve_0$, $d(X)=1$ and $c_{\vR}(X)=X^N$.}
We can express this as 
%\agu{switch order of d and $\bard$?}

\begin{equation}
  \label{eq:recur-core}
  a_i(X) = \sum_{j=1}^{\min(t,i)} g_{i,j}(X)a_{i-j}(X) + \vc[i] d(X) \pmod{c_\vR(X)}
\end{equation}

for some $\vc \in \F^N$, $d(X) \in \F[X]$, $\deg(g_{i,j}) \leq j$ and recall $c_{\vR}(X)$ is the characteristic polynomial of $\vR$. % and $\deg(c_\vR(X)) = N$.

In this context, Lemma~\ref{lmm:structure1} states that $a_i(X) = \sum_{j=0}^{N-1} h_{i,j}(X) \cdot \vc[j]d(X) \pmod{c_\vR(X)}$.
Therefore the desired vector $\vA^T\vb$ is the coefficient vector of
\begin{equation}
  \label{eq:core-answer}
  \begin{aligned}
    \sum_{i=0}^{N-1} \vb[i] a_i(X) &= \sum_{i=0}^{N-1} \sum_{j=0}^{N-1} \vb[i] h_{i,j}(X) \cdot \vc[j]d(X) \pmod{c_\vR(X)} \\
    &= \vb^T \vH \vc \cdot d(X) \pmod{c_\vR(X)}
  \end{aligned}
\end{equation}

So it suffices to compute $\vb^T \vH \vc$, and perform the multiplication by $d(X) \pmod{c_\vR(X)}$ at the end.
By Lemma~\ref{lmm:structure2}, $\vb^T\vH\vc$ is the bottom right element of the following $t\times t$ matrix:
\begin{equation}
  \label{eq:triangular-recursion}
  \begin{aligned}
    \sum_{0 \leq j \leq i < N} \vb[i] \vT_{[j:i]}\vc[j] &= \sum_{0 \leq j \leq i < N/2} \vb[i] \vT_{[j:i]}\vc[j] + \sum_{j<N/2, i\ge N/2} \vb[i] \vT_{[j:i]}\vc[j] + \sum_{N/2 \leq j \leq i < N} \vb[i] \vT_{[j:i]}\vc[j] \\
    &= \sum_{0 \leq j \leq i < N/2} \vb[i] \vT_{[j:i]}\vc[j] + \left( \sum_{i\ge N/2} \vb[i]\vT_{[N/2:i]} \right)\left( \sum_{j<N/2} \vT_{[j:N/2]}\vc[j]\right) + \sum_{N/2 \leq j \leq i < N} \vb[i] \vT_{[j:i]}\vc[j]
  \end{aligned}
\end{equation}

The first and third sums have the same form as the original. Recursively applying this decomposition, we see that it suffices to compute the last row of $ \sum_{i \in [\ell:r]} \vb[i] \vT_{[\ell:i]}$ and the last column of $\sum_{j \in [\ell:r]} \vc[j] \vT_{[j:r]}$
for all dyadic intervals $[\ell:r] = \left[ \frac{b}{2^d}N : \frac{b+1}{2^d}N \right]$. By symmetry, we can focus only on the former. Denoting this row vector (i.e.\ $1 \times t$ matrix) $\vP_{\ell,r} = \sum_{i \in [\ell:r]} \vB_i \vT_{[\ell:i]}$ (where $\vB_i = \vb[i]\ve_{t-1}^T = \begin{bmatrix} 0 & \cdots & \vb[i] \end{bmatrix}$), it satisfies a simple relation
\begin{equation}
  \label{eq:row-recursion}
  \vP_{\ell,r} = \vP_{\ell,m} + \vP_{m,r}\vT_{[\ell:m]}
\end{equation}
for any $\ell \leq m < r$, and thus can be computed with two recursive calls and a matrix-vector multiply over $\F[X]^{t}$.

Also, note that the $\vT_{[\ell:r]}$ are independent of $\vb$, so the relevant products will be pre-computed.

%Finally, observe that when $r-\ell \leq t$, the last row of $\vT_{[\ell:r]}$ is simply an indicator vector with a $1$ in the $r-\ell$-th spot from the end (conceptually, the companion matrices act as a shift before reaching the recurrence width). Thus we stop the recurrence when the problem size gets below $t+1$.

%Finally, when $r - \ell = t$, we can compute $\vP_{[\ell:r]}$ iteratively with $t$ matrix-vector multiplications (over polynomials of degree $O(t)$).

Symmetrically, the $t \times 1$ matrices $\vQ_{\ell,r} = \sum_{j \in [\ell:r]} \vT_{[j:r]}\vC_j$ (where $\vC_j = \vc[j]\ve_{t-1}$) satisfy $\vQ_{\ell,r} = \vT_{[m:r]}\vQ_{\ell,m} + \vQ_{m,r}$ and can be computed in the same way.

We present the full details in Algorithm~\ref{algo:transpose-mult}. The main subroutine $\textsc{P}(\ell,r)$ computes $\vP_{\ell,r}$. Subroutine $\textsc{H}(\ell,r)$ computes $\vb[R]^T \vH[R,R]\vc[R]$ for index range $R=[\ell:r]$. We omit the procedure $\textsc{Q}$ computing $\vQ_{\ell,r}$, which is identical to procedure $\textsc{P}$ up to a transformation of the indexing.

\begin{algorithm}
  \caption{\ATmult}
  \label{algo:transpose-mult}
  \begin{algorithmic}[1]
    \renewcommand{\algorithmicrequire}{\textbf{Input:}}
    \Require{$\vG,\vb,\vc$,$\vT_{[\ell:r]}$ for all dyadic intervals: $[\ell:r] = \left[ \frac{b}{2^d}N: \frac{b+1}{2^d}N \right]$, $d \in [m], b \in \left[2^d\right]$}
    %\Require{$\vb,\ell,r$}
    \renewcommand{\algorithmicensure}{\textbf{Output:}}
    %\Ensure{$\vP_{\ell,r} = $ Last row of $\sum_{i=\ell}^{r}\vb[i]\vT_{[\ell:i]}$}
    %\Ensure{$\sum_{\ell \le j \le i < r} \vb[i] \vT_{[j:i]}\vc[j]$}
    \Ensure{$\vb^T\vH\vc$}
    \Function{P}{$\vb,\ell,r$} \Comment{Computes $\vP_{\ell,r} = \sum_{i=\ell}^{r}\vB_i\vT_{[\ell:i]}$}
    \If{$r-\ell \leq t$}%\Comment{Base case}
    \State $\vP_{\ell,r} \gets \sum_{i \in [\ell:r]} \vB_i\vT_{[\ell:i]}$%\Return{$\begin{bmatrix} 0 & \cdots & \vb[k+2^a-1] & \cdots & \vb[k] \end{bmatrix}$}
    \label{step:P:base}
    \Else%\Comment{Recursive Step}
    \State $m \gets (\ell+r)/2$
    \State $\vP_{\ell,r} \gets \Call{P}{\vb,\ell,m} + \Call{P}{\vb,m,r} \vT_{[\ell:m]}$
    \label{step:P:recurse}
    \State \Return{$\vP_{\ell,r}$}
    \EndIf
    \EndFunction

    \Function{Q}{$\vc,\ell,r$} \Comment{Computes $\vQ_{\ell,r} = \sum_{j \in [\ell:r]} \vT_{[j:r]}\vC_j$}
    \State (Identical to \textsc{P} up to indexing)
    \EndFunction

    \Function{H}{$\ell,r$} \Comment{Computes $\sum_{\ell \le j \le i < r} \vB[i] \vT_{[j:i]}\vC[j]$}
    %\Require{$\vP_{\ell,r}, \vQ_{\ell,r}$ for all dyadic intervals}
    \If{$r-\ell \leq t$}
    \State \Return{$\sum_{\ell \le j \le i < r} \vB[i] \vT_{[j:i]}\vC[j]$}
    \label{step:H:base}
    \Else
    \State $m \gets (\ell+r)/2$
    \State \Return{\Call{H}{$\ell,m$} + $\vP_{\ell,m}\vQ_{m,r}$ + \Call{H}{$m,r$}}
    \label{step:H:recurse}
    \EndIf
    \EndFunction

    \Function{TransposeMult}{$\vb,\vc$}
    \State \Call{P}{$\vb,0,N$}
    \State \Call{Q}{$\vc,0,N$}
    \State \Return{\Call{H}{$0,N$}}
    \EndFunction
  \end{algorithmic}
\end{algorithm}

%The above discussion implies that Algorithm~\ref{algo:transpose-mult} is correct:
%\begin{thm}
%\label{thm:ATmult-correct}
%Assuming all the required $\vT_{[\cdot,\cdot]}$ are correctly computed, $\ATmult(\vb,\vG,m,0)$ returns $\vP_{0,N}$.
%\end{thm}

\subsection{Pre-processing time}
\label{sec:pre-process}

The pre-processing step is computing a compact representation of $\vH$, which we represent through the matrices $\vT_{\left[bN/2^d:bN/2^d+N/2^{d+1}\right]} \in \F[X]^{t \times t}$ for $0\le d<m$, $0\le b<2^d$. Since we have assumed that $N$ is a power of 2, these ranges can be expressed in terms of dyadic strings; we need to pre-compute $\vT_{[s]}$ for all strings $s \in \{0,1\}^*, |s| \le \lg N$ (where we interpret $[s]$ as the corresponding dyadic interval in $[0,N-1]$). All the required matrices can be computed in a natural bottom-up fashion: 

\begin{lmm}
\label{lmm:tree-compute}
We can pre-compute $\vT_{[s]}$ for all strings $s \in \{0,1\}^*, |s| \le \lg N$ with $O(t^{\omega} \cM(N)\log N)$ operations.
\end{lmm}
\begin{proof}
Fix an arbitrary $s^*$ of length $\ell < \lg N$. We can compute $\vT_{[s^*]} = \vT_{[s^*0]}\cdot \vT_{[s^*1]}$. Using the matrix multiplication algorithm, we have $O(t^{\omega})$ polynomial multiplications to compute, where the polynomials are of degree at most $\frac{N}{2^\ell}+t$ by Lemma~\ref{lmm:T-range-sizes}. So computing $\vT_{[s^*]}$ takes $O(t^\omega \cM(N/2^{\ell}+t) )$ operations. Thus computing $\vT_{[s]}$ for all $|s| = \ell$ is $O(t^\omega\cM(N+t2^{\ell}))$, and computing all $\vT_{[s]}$ is $O(t^{\omega} \cM(N)\log N)$, as desired.\footnote{In the last estimate we note that $|s^*|\le \log_2(t+1)$, so the $t^{\omega+1} 2^{\ell}$ term only gets added up to $\ell=\log(N/t)$.}
\end{proof}

\begin{cor}
\label{cor:transpose-mult-pre}
Pre-processing for Algorithm~\ref{algo:transpose-mult} requires $O(t^{\omega} \cM(N)\log N)$ operations over \,$\F$.
\end{cor}

\subsection{Runtime analysis}

Assuming that the pre-processing step is already done, the runtime of Algorithm~\ref{algo:transpose-mult} can be bounded with a straightforward recursive analysis.

\begin{lmm}
\label{lmm:transpose-mult-comp}
After pre-processing, Algorithm~\ref{algo:transpose-mult} needs $O(t^2\cM(N)\log{N})$ operations over $\F$.
\end{lmm}
\begin{proof}
  We analyze the complexity of each of the main steps of Algorithm~\ref{algo:transpose-mult}.

  \begin{description}
    \item[Step~\ref{step:P:base}]
      This can be computed as
      \[
        \sum_{i=\ell}^{r-1} \vB_i\vT_{[\ell:i]} = \vB_\ell + \left( \vB_{\ell+1} + \dots + (\vB_{r-2} + \vB_{r-1}\vT_{r-2})\vT_{r-3} \dots \right)\vT_\ell,
      \]
      from the inside out, which consists of $r-\ell=t$ multiplications over polynomials of degree $O(t)$ and has $O(t^2\cM(t))$ complexity.\footnote{Multiplication by a companion matrix $\vT_i$ only requires $t$, not $t^2$, multiplications.} Over all $N/t$ such base case ranges, the work is $O(Nt\cM(t)) \le O(t^2\cM(N))$.
    \item[Step~\ref{step:P:recurse}]
      We analyze the number of operations ignoring recursive calls. Note that by Lemma~\ref{lmm:T-range-sizes}, each element of vector $\vP_{\ell,r}$ is a polynomial of size at most $r-\ell+t$. Thus the matrix-vector multiplication $\vP_{m,r}\vT_{[\ell:m]}$ performs $O(t^2)$ multiplications of polynomials of degree at most $r-\ell+t$. This simplifies to $O(t^2\cM(r-\ell+t)) = O(t^2\cM(r-\ell))$, since the recursion stops at size $t$. 
    \item[Step~\ref{step:H:base}]
      If the $\vT_{[j:i]}$ are already known, this can be trivially computed in $O(t^3)$ time by summing up the polynomials. As with Step~\ref{step:P:base}, all the base cases perform at most $O(t^2\cM(N))$ total operations.

      Computing the $\vT_{[j:i]}$ can be done as part of pre-processing by explicitly inverting the submatrix $\vG[\ell:r,\ell:r]$, which is done in $O(t^\omega \cM(t))$ time. Over all base cases, this contributes $O(t^\omega \cM(N))$ operations, less than the other pre-processing steps.
    \item[Step~\ref{step:H:recurse}]
      Multiplication of $\vP_{\ell,m} \in \F[X]^{1 \times t}$ and $\vQ_{m,r} \in \F[X]^{t \times 1}$ requires $t$ multiplications of polynomials of degree $r-\ell+t$, requiring $O(t\cM(r-\ell))$ operations.
  \end{description}

  The total complexity follows from a standard divide-and-conquer analysis. If $T(N)$ is the runtime of the call $\textsc{P}(\vb,0,N)$, then the above shows $T(N) = 2T(N/2) + O(t^2\cM(N))$. This solves to $T(N) = O(t^2 \cM(N)\log{N})$ assuming the base cases are efficient enough, which they are. The call $\textsc{H}(0,N)$ is similar.
\end{proof}

We remark that the input vectors $\vb$ and the error coefficients $\vc$ are duals in a sense: they affect Algorithm~\ref{algo:transpose-mult} symmetrically and independently (through the computations of $\vP_{\ell,r}$ and $\vQ_{\ell,r}$ respectively). The main difference is that $\vc$ affects the pre-processing steps and $\vb$ affects the main runtime. We will make use of this observation later in Sections~\ref{subsec:recur-error} and~\ref{subsec:matrix-matrix}, which respectively generalize $\vc$ and $\vb$ from vectors to matrices.

%\subsection{Post-processing}
%We are not quite done yet with regard to computing $\vA^T\vb$ since after Algorithm~\ref{algo:transpose-mult}, we still have to perform the step $\vP_{0,N} \otimes \vF$. Note that in the setting of Definition~\ref{defn:otimes-poly} this is just $t$ multiplications over $\F[X]$ of polynomials of size at most $N$: this can be done with $O(tN\log{N})$ operations.
%Further for the case of an $\vR$-matrix recurrence when $\vR$ is $(\alpha,\beta)$-Krylov efficient, implies that each of the $t$ $\otimes$ operations can be done with $\tO(\beta N)$ operations (and $\tO(\alpha N)$ pre-processing operations).\footnote{Indeed if $\vP_{0,N}[j]=\sum_{i=0}^{N-1} p_i X^i$, then $\vP_{0,N}[j](\vR)\otimes \vf_j=\kry(\vR,\vf_j)\cdot \vp$, where $\vp=(p_0,\dots,p_{N-1})$.}
Corollary~\ref{cor:transpose-mult-pre} and Lemma~\ref{lmm:transpose-mult-comp} imply the following result:
\begin{thm}
\label{thm:mainATb}
For any matrix $\vA$ satisfying a modular recurrence~\eqref{eq:recur-mod} of width $(t,1)$ and degree $(1,0)$,
with $O(t^{\omega} \cM(N) \log N)$ pre-processing operations, the product $\vA^T\vb$ can be computed for any $\vb$ with $O(t^2\cM(N)\log{N})$ operations over $\F$.
%For an $\vR$-matrix recurrence~\eqref{eq:intro-recur-matrix} \ar{Maybe link the appropriate definition here?} that is nice and $\vR$ is $(\alpha,\beta)$-Krylov efficient, pre-processing and computation need additional $\tO(t\alpha N)$ and $\tO(t\beta N)$ operations respectively.
\end{thm}

\subsection{Space Complexity}
\label{subsec:space}
The space used by Algorithm~\ref{algo:transpose-mult} is dominated by the pre-computations. The matrix $\vT_{[\ell:r]}$ contains $t^2$ polynomials of degree $O(r-\ell)$ which has space requirement $O(t^2(r-\ell))$ total. For a fixed size $2^a$, all the matrices of the form $\vT_{[\ell:\ell+2^a]}$ thus require $O(t^2N)$ space. The total space complexity is $O(t^2N \log{N})$.

Note that we can always run the algorithm without the explicit pre-processing and the coefficient of runtime complexity becomes $t^\omega$ instead of $t^2$. This is asymptotically the same if $t$ is held constant (which is true for most applications, e.g.\ $t\le 2$ for the previous results of Driscoll~\cite{driscoll} and Olshevsky and Shokrollahi~\cite{OS00}). Furthermore, in this case, a factor of $\log{N}$ can be saved in the space complexity: In Algorithm~\ref{algo:transpose-mult}, the computations can be performed bottom-up layer by layer on the recursion tree, instead of depth-first. The pre-computations were also performed bottom-up, and can be discarded after each layer of the recursion tree is finished.

In the operation $\vA\vb$, the transpose of this algorithm does operations in the opposite order and requires the top level transitions $\vT_{[0:N/2]}$ and $\vT_{[N/2:N]}$ first (see Section~\ref{sec:Ab}). Space can still be saved, but each level of the pre-computation tree will have to be re-computed and discarded every time. This incurs a $\log{N}$ runtime factor, so that the algorithm would have a $O(t^\omega \cM(N)\log^2{N})$ operation complexity to get a $O(t^2N)$ space bound.

% hide original Ab alg

% hide old Ab

\section{The General Recurrence and Related Operations}
\label{sec:general-mult}
In Section~\ref{sec:transpose}, we introduced a multiplication algorithm for matrices $\vA$ satisfying a $(\vG,\vF,\vR)$-recurrence (see Definition~\ref{defn:recurrence-width}) with simplified assumptions on $\vG,\vF,\vR$. In this section, we first finish the details of the matrix-vector multiplication algorithm for matrices satisfying the general recurrence. Then we cover several operations that are simple consequences or extensions of the algorithm, such as the case when $\vA$ is rectangular and multiplication by a matrix $\vB$ instead of vector $\vb$.

There are three modifications to recurrence~\eqref{eq:recur-core} needed to recover the general case~\eqref{eq:recur}. First, we generalize from degree-$(1,0)$ to degree-$(d,\bard)$ recurrences. Second, we generalize the error matrix $\vF$ from rank $1$ to rank $r$. Finally, we relax the constraint that $\vR$ is a companion matrix.

\subsection{General $\vG$}
\label{subsec:recur-degree}
First, it is easy to see that if $\vG$ satisfies a degree $(d,0)$ recurrence, the companion matrices $\vT_i$ can still be defined and the degree bounds in Lemma~\ref{lmm:T-range-sizes} are scaled by $d$. Since polynomials involved have degrees scaled by $d$, the cost of Algorithm~\ref{algo:transpose-mult} increases by a factor of $d$.\footnote{The runtime is actually slightly better because the degrees don't grow beyond $N$ (due to the modulus). The $\log{N}$ factor can be replaced by $\log{N}-\log{d}$, but the asymptotics are the same since we usually assume $d$ is fixed.}

Next, consider a matrix $\vG$ for a $(d,\bard)$-recurrence (i.e.\ $\vG$ lower triangular banded, $\deg(\vG[i,j]) \leq d(i-j)+\bar{d}$). By multiplying $\vG$ on the left and right by suitable diagonal matrices (which depend on the $g_{i,0}(X)$'s), it can be converted into a matrix $\vG'$ satisfying $\deg(\vG'[i,j]) \leq (d+\bard)(i-j)$, i.e.\ corresponding to a $(d+\bard, 0)$-recurrence. The details of this transformation are shown in Appendix~\ref{sec:structure}.

Therefore algorithms for a $(d,\bard)$-recurrence have runtimes scaled by a factor of $(d+\bard)$. This generalization is also independent of the other generalizations - the algorithm does not change, only the operation count. We henceforth assume again that $(d,\bard) = (1,0)$ for simplicity.

\subsection{General $\vF$}
\label{subsec:recur-error}
Now consider a modular recurrence of width $(t,r)$. As usual convert the vectors to polynomials; write $\vF = \vC\vD$ and let $d_i(X) = \sum_{j=0}^{N-1} \vD[i,j]X^j$. By a small modification of the derivation in Section~\ref{sec:transpose}, we deduce that the desired quantity $\vA^T\vb$ is the coefficient vector of the polynomial
\[
  \vb^T \vH \vC \cdot \begin{bmatrix}d_0(X) \\ \vdots \\ d_{r-1}(X) \end{bmatrix} \pmod{c_\vR(X)}
\]
(compare to equation~\eqref{eq:core-answer} in the case $r=1$). Again the multiplications by $d_0(X), \dots, d_{r-1}(X)$ can be postponed to the end and incur a cost of $r\cM(N)$ operations, so we focus on computing $\vb^T\vH\vC \pmod{c_\vR(X)}$. Recall that $\vb \in \F^{N}, \vH \in \F[X]^{N \times N}, \vC \in \F^{N \times r}$.
We can still apply Algorithm~\ref{algo:transpose-mult} directly. Procedure $\textsc{P}$ still computes $\vP_{\ell,s} = \sum_{\ell}^{s-1} \vB_i \vT_{[\ell:i]}$, but procedure \textsc{Q} now computes
\[
  \quad \vQ_{\ell,s} = \sum_{\ell}^{s-1} \vT_{[i:s]}\vC_i \qquad \text{where} \qquad 
  \vC_i = \begin{bmatrix} 0 & \cdots & 0 \\ \vdots & \ddots & \vdots \\ \vC[i,0] & \cdots & \vC[i,s-1] \end{bmatrix},
\]
with the only difference being that it now has dimensions $t \times r$ instead of $t \times 1$.

We analyze the change in runtime compared to Lemma~\ref{lmm:transpose-mult-comp}. The call to \textsc{P} does not change; its complexity is still $O(t^2 \cM(N) \log{N})$.

In procedure \textsc{Q}, the change is that the recursive computation $\vQ_{\ell,r} = \vT_{[m:r]}\vQ_{\ell,m} + \vQ_{m,r}$ 
is now a multiplication of a $t\times t$ and $t\times r$ matrix instead of $t \times t$ by $t \times 1$. The runtime coefficient changes from $t^2$ to $\alpha_{t,r}$, where $\alpha_{t,r}$ is the cost of performing a $t\times t$ by $t\times r$ matrix multiplication. Also note that since this does not depend on $\vb$, this can be counted as part of the pre-processing step. The total pre-processing step is now $O\left( (t^\omega + \alpha_{t,r}) \cM(N) \log{N} \right)$. 

In procedure \textsc{H}, we are now performing a $1 \times t$ by $t \times r$ multiplication. The runtime of $\textsc{H}(0,N)$ is now $O(tr\cM(N)\log{N})$. The total runtime (the calls to $\textsc{P}$ and $\textsc{H}$) is $O( t(t+r)\cM(N)\log{N})$.

Finally, we note that $\alpha_{t,t} = t^\omega$, and in general $\alpha_{t,r} \leq \min(rt^2, (1+r/t)t^\omega)$. For large enough $r$ we use the latter bound, whence the pre-processing coefficient $O(t^\omega + \alpha_{t,r})$ above becomes $O((t+r)t^{\omega-1})$.

In particular, the pre-processing time is still $O(t^\omega \cM(N) \log{N})$ and the runtime is still $O(t^2 \cM(N)\log{N})$ when $r = O(t)$. This fully captures the original polynomial recurrence~\eqref{eq:intro-recur-poly}, and in particular the orthogonal polynomial transforms~\cite{driscoll}.

\subsection{General $\vR$}
\label{subsec:recur-matrix}
%\begin{lmm}
%  If $f(X) = a(X)/b(X)$ is a rational function 
%  For rational functions $f,g \in \F(X)$, if $f(X) \equiv g(X) \pmod{M(X)}$, then $f(\vR) = g(\vR)$. (Here we define $f(\vR)$ naturally: if $f(X) = a(X)/b(X)$ where $a,b\in\F[X]$, then $f(\vR) = a(\vR)b(\vR)^{-1}$.)
%\end{lmm}
%Thus if we define $\vH = (h_{ij}(X)) = \vG^{-1} \pmod{c_\vR(X)}$, then by the work-horse lemma

Aside from the reason that it is necessary to handle the displacement rank recurrence~\eqref{eq:intro-recur-dp}, we provide an intrinsic reason for considering the general matrix-recurrence~\eqref{eq:recur}, as a natural continuation of polynomial recurrences~\eqref{eq:intro-recur-poly}. Equation~\eqref{eq:intro-recur-poly} is written in terms of polynomials, but it is defining vectors. Another way of writing it is as follows: we are actually defining a recurrence on vectors $\va_i$ (the rows of $\vA$) satisfying $\va_i = \sum g_{ij}(X)\cdot\va_{i-j}$, where the bilinear operator $(\cdot: \F[X]\times\F^N \to \F^N)$ is defined for $g(X)\cdot\va$ by converting $\va$ to a polynomial, multiplying by $g(X) \pmod{c_\vR(X)}$, and converting back to a vector. This is just an instance of equipping the vector space $\F^N$ with a $\F[X]$-module structure. Thus it is natural to consider what happens in general when we define $\va_i = \sum g_{ij}(X)\cdot\va_{i-j}$ for any $\F[X]$-module structure on $V=\F^N$. In general, this is uniquely defined by the action of $X$; this is a linear map on $V$, hence equivalent to multiplication by a matrix $\vR$. This leads to the matrix recurrence~\eqref{eq:recur}.

By Lemma~\ref{lmm:structure1},
\[
  \va_i = \sum_{j=0}^{N-1} h_{i,j}(\vR) \vf_j.
\]
We can simplify this expression:
\begin{align*}
  \va_i &= \sum_{j=0}^{N-1} h_{i,j}(\vR) \vf_j \\
        & = \sum_{j=0}^{N-1} h_{i,j}(\vR) \left( \sum_{k=0}^{r-1} c_{jk}\vd_k \right) \\
        & = \sum_{j=0}^{N-1} \sum_{k=0}^{r-1} c_{jk}h_{i,j}(\vR)\vd_k \\
        &= \sum_{j=0}^{N-1} \sum_{k=0}^{r-1} c_{jk}\left( \sum_{\ell=0}^{N-1} h_{i,j}[\ell]\vR^\ell \right) \vd_k \\
        & = \sum_{j=0}^{N-1} \sum_{k=0}^{r-1} c_{jk}\left( \sum_{\ell=0}^{N-1} h_{i,j}[\ell]\vR^\ell\vd_k \right) \\
        &= \sum_{j=0}^{N-1} \sum_{k=0}^{r-1} c_{jk} (\kry(\vR,\vd_k)\vh_{ij}),
\end{align*}
where $\vh_{ij}$ is the coefficient vector of $h_{i,j}(X)$.

Thus the desired answer to $\vA^T\vb$ is
\begin{align}
  \sum_{i=0}^{N-1} \vb[i] \va_i &= \sum_{i=0}^{N-1} \sum_{j=0}^{N-1} \sum_{k=0}^{r-1} \vb[i] c_{jk} \kry(\vR,\vd_k) \vh_{ij} \nonumber \\
  &= \sum_{k=0}^{r-1} \kry(\vR,\vd_k) \sum_{i=0}^{N-1} \sum_{j=0}^{N-1} \vb[i] c_{jk} \vh_{ij} \label{eq:recur-matrix-reduction}
\end{align}
Finally, recall that $\vb^T\vH\vC$ is a $1 \times r$ vector of polynomials, and $\sum_{i=0}^{N-1} \sum_{j=0}^{N-1} \vb[i] c_{jk} \vh_{ij}$ is the coefficient array of its $k$th entry. Thus we compute $\vb^T\vH\vC$ as before, and then perform $r$ matrix-vector multiplications by the Krylov matrices $\kry(\vR,\vd_k)$.

With fully general $\vG,\vF,\vR$, we get the following matrix-vector multiplication runtime.
\begin{thm}
  \label{thm:recur-general}
  If $\vA \in \F^{N \times N}$ is a matrix satisfying a $\vR$-matrix recurrence~\eqref{eq:recur} (with known characteristic polynomial $c_\vR(X)$) of width $(t,r)$ and degree $(d,\bard)$, and $\vR$ is $(\alpha,\beta)$-Krylov efficient, then with $O( (d+\bar{d})t^{\omega-1}(t+r) \cM(N)\log{N} + \alpha)$ pre-processing, the products $\vA^T\vb$ and $\vA\vb$ for any vector $\vb$ can be computed with $O( (d+\bar{d})t(t+r) \cM(N)\log{N} + r\beta )$ operations.
\end{thm}

%We remark that in these reductions, the $\otimes$ operator was first isolated and the specific properties of each instantiation is handled somewhat separately from the core algorithms. In particular, the factorization in Section~\ref{subsec:recur-matrix} can be performed as long as the $h_{i,j}^{(k)}(X)$ defined in Lemma~\ref{lmm:work-horse} are polynomials. This is always true when the ring $\cR$ in Definition~\ref{defn:recurrence-width} is $\F[X]$ or a quotient ring, which holds for all of our specific instantiations in Definitions~\ref{defn:otimes-poly}-\ref{defn:otimes-mod-matrix}. We hypothesize that there are potentially other instantiations of $\otimes$ with useful applications that can also be solved in a similar way using our generic framework.

\subsection{Rectangular Matrices}
\label{subsec:rectangular}
When $\vA \in \F^{M \times N}$, Algorithm~\ref{algo:transpose-mult} and the modifications in this section still apply. The runtime is easy to analyze; the sizes of the pre-processing and recursion parameters (e.g.\ degree of polynomials in $\vT_{[\ell:r]}$) depend on $M$, and the sizes of the post-recursion steps (e.g.\ multiplication by Krylov matrices) depend on $N$.
\begin{cor}
  \label{cor:recur-general}
  If $\vA \in \F^{M \times N}$ is a matrix satisfying a $\vR$-matrix recurrence~\eqref{eq:recur} (with known characteristic polynomial $c_\vR(X)$) of width $(t,r)$ and degree $(d,\bard)$, and $\vR$ is $(\alpha,\beta)$-Krylov efficient, then with $O( (d+\bar{d})t^{\omega-1}(t+r) \cM(M)\log{M} + \alpha)$ pre-processing, the products $\vA^T\vb$ and $\vA\vb$ for any vector $\vb$ can be computed with $O( (d+\bar{d})t(t+r) \cM(M)\log{M} + r\beta )$ operations.
\end{cor}
Note in particular that the dependence on $N$ is folded into the Krylov multiplication step.
\ar{Maybe to emphasize this point just for the theorem above use $\alpha(N)$ and $\beta(N)$ instead of just $\alpha$ and $\beta$?}
\agu{Good idea but I feel it would be confusing to change notation just in one location.}
For our applications, the Krylov efficiency constants $(\alpha,\beta)$ of an $N \times N$ matrix will be of order $\tO(N)$ (see Section~\ref{sec:Krylov}).

\subsection{Krylov Efficiency}
\label{sec:Krylov}
%\ar{This is the section where the Krylov stuff goes in. Currently, I just put in the stuff on what I wrote showing the Jordan normal form is $(1,1)$-Krylov efficient (you guys might have a simpler/shorter way of doing this-- feel free to change the Jordan normal form stuff). You guys have the tokens on this section. I would recommend that you guys put each new piece in subsequent subsections in here.}

In Section~\ref{subsec:recur-matrix}, we showed how to factor matrix recurrences into the product of a polynomial/modular recurrence and a Krylov matrix, thus reducing the runtime of a $\vR$-matrix recurrence to the Krylov efficiency (Definition~\ref{defn:krylov}) of $\vR$.
In this section we show Krylov efficiency for a particular class of matrices that is necessary to prove Theorem~\ref{thm:intro-dr}.

We note that a natural approach to Krylov efficiency is utilizing the Jordan normal form. In Appendix~\ref{subsec:jordan}, we show how knowing the Jordan form $\vM = \vA\vJ\vA^{-1}$ implies a particularly simple algorithm for Krylov efficiency, and provide cases for which we can compute this Jordan decomposition efficiently.

However, this reduction is clearly one way-- finding a Jordan decomposition is stronger than Krylov efficiency, but the latter problem has more structure that we can take advantage of. Now we will show that the class of banded triangular matrices are Krylov efficient by showing that the Krylov matrix itself has low recurrence width (equal to the bandwidth).

We remark that the Krylov efficiency concept does not apply only to our problem. If $\vK$ is the Krylov matrix on $\vA$ and $\vb$, then $\vK\vb = \sum \vb[i] \vA^i \vx$ is naturally related to contexts involving Krylov subspaces, matrix polynomials, and so on. The product $\vK^T\vb = [\vb \cdot \vx, \vb \cdot \vA\vx, \vb \cdot \vA^2\vx, \dots]$ is also useful; it is the first step in the Wiedemann algorithm for computing the minimal polynomial or kernel vectors of a matrix $\vA$~\cite{wiedemann}.

\subsubsection{Krylov Efficiency of triangular banded matrices}
Let $\vM$ be a lower triangular $(\Delta+1)$-band matrix, i.e.\ all values other than $\vM[i,\ell]$ for $i-\Delta \le \ell \le i$ are zero. Let $\vy$ be an arbitrary vector and let $\vK$ denote the Krylov matrix of $\vM$ with respect to $\vy$.

We will show that $\vK$ satisfies a modular recurrence of width $(\Delta,1)$. The same results also hold for upper triangular matrices.

Define polynomials $f_i(X) = \sum_{j=0}^{N-1} \vK[i,j] \cdot X^j$ and let $\vF = \begin{bmatrix} f_0(X) \\ \vdots \\ f_{N-1}(X) \end{bmatrix}$. We can alternatively express it as
\[ \vF = \sum_{j=0}^{N-1} \vK[:,j] X^j = \sum_{j=0}^{N-1} (\vM^j \vy) X^j = \left( \sum_{j=0}^{N-1} (\vM X)^j \right)\vy. \]
Multiplying by $\vI-\vM X$, we get the equation
\begin{equation}
\label{eq:krylov-system1}
(\vI-\vM X)\vF = \vy - (\vM X)^N \vy
\end{equation}
Therefore it is true that
\begin{equation}
  \label{eq:krylov-mod}
  (\vI-\vM X)\vF \equiv \vy \pmod{X^N}
\end{equation}
and furthermore, $\vF$ can be defined as the unique solution of equation~(\ref{eq:krylov-mod}) because $\vI-\vM X$ is invertible in $\mathbb F[X]/(X^N)$ (since it is triangular and its diagonal is comprised of invertible elements $(1-\vM[i,i]X)$).

%But equation~\eqref{eq:krylov-mod} can be interpreted as a rational recurrence with error. Explicitly expanding \eqref{eq:krylov-mod}, the $f_i(X)$ satisfy
%\[ (1-\vM[i+1,i+1]X)f_{i+1}(X) \equiv \sum_{j=0}^\Delta \vM[i+1,i-j]X f_{i-j}(X) + \vy[i] \pmod{X^N} \]
%Therefore the multiplications $\vK\vx$ and $\vK^T\vx$ can be computed using Theorem~\ref{thm:recur-mod}.
%
But equation~\eqref{eq:krylov-mod} exactly defines a modular recurrence~\eqref{eq:recur-mod} of degree $(0,1)$ and width $(\Delta,1)$. Theorem~\ref{thm:recur-general} implies
\begin{thm}
  \label{thm:ut-band-ke}
  Any triangular $\Delta$-band matrix is $(\Delta^\omega \cM(N)\log{N}, \Delta^2 \cM(N)\log{N})$-Krylov efficient.
\end{thm}

\subsection{Matrix-matrix multiplication}
\label{subsec:matrix-matrix}
Consider the multiplication $\vA^T\vB$ for a matrix $\vB \in \F^{N \times P}$. An obvious way to compute this is by computing the product for each of the $P$ column vectors of $\vB$ individually. However, it can be more efficient to compute the matrix-matrix product as one (similar observations have been made in the context of finding inverses of matrices with small displacement rank~\cite{bostan2007}). Algorithm~\ref{algo:transpose-mult} can still be applied; the only change is that the sums
\[
  \vP_{\ell,r} = \sum_{\ell}^{r-1}\vB_i\vT_{[\ell:i]}
\]
(see equation~\eqref{eq:row-recursion})
have dimensions $P \times t$ instead of $1 \times t$. Thus the cost of computing $\vA^T\vB$ is the same as for $\vA\vb$ but with the $t^2$ constant replaced with $\alpha_{t,P}$ (recall from Section~\ref{subsec:recur-error} that $\alpha_{t,r}$ is the cost of performing a $t\times t$ by $t\times r$ matrix multiplication).

Note this is exactly the same as the change when going from rank $1$ to rank $r$ error in Section~\ref{subsec:recur-error}, but in this case it affects the main runtime instead of pre-processing step.

For $P$ large (for example, suppose $\vB$ has the same dimension $N \times N$ as $\vA$), an interesting way to view this is that the matrix-vector multiplication algorithm has \emph{amortized complexity} $1/P \cdot \tO(\alpha_{t,P}N) \le 1/P \cdot \tO( (1+P/t)t^\omega N) = \tO( (1/P + 1/t) t^\omega N) \approx \tO(t^{\omega-1} N)$ per vector.

\subsection{$\vA\vb$ multiplication}
The transposition principle gives a matrix-vector multiplication algorithm for $\vA\vb$ that can be derived from our algorithm for $\vA^T\vb$ (which only uses linear operations). This algorithm has the same pre-processing step and has the same time and space complexity as Algorithm~\ref{algo:transpose-mult} (see discussion in Section~\ref{subsec:space}).

In Appendix~\ref{sec:Ab}, we provide details of the general $\vA\vb$ algorithm. We provide an intuitive algorithm under the simplified assumptions of equation~\eqref{eq:intro-recur-poly}, and then generalize to the full case of equation~\eqref{eq:recur}. Interestingly, the restriction of this algorithm to the setting of Driscoll et al.~\cite{driscoll} turns out to be the same algorithm.

The approach of using the transposition principle to convert an algorithm for $\vA^T\vb$ to one for $\vA\vb$, through finding a sparse factorization of $\vA$ (see Section~\ref{sec:spw}), was previously used in the case of orthogonal polynomial transforms~\cite{bostan2010op}.

\section{Displacement Rank}
\label{sec:disp-rank}
Recall that the \emph{displacement rank} of a matrix $A$ with respect to matrices $\vL, \vR$ is defined as the rank of the {\em error matrix} 
\[\vE=\vL\vA - \vA\vR.\]
The concept of displacement rank has been used to generalize and unify common structured matrices such as Hankel, Toeplitz, Vandermonde, and Cauchy matrices; these matrices all have low displacement ranks with respect to diagonal or shift matrices being $\vL$ and $\vR$. Olshevsky and Shokrollahi~\cite{OS00} defined the confluent Cauchy-like matrices to be the class of matrices with low displacement rank with respect to Jordan form matrices; this class of matrices generalized and unified the previously mentioned common structured matrices. Our class of structured matrices extends the results of~\cite{OS00} to a more general form for $\vL$ and $\vR$ while matching the complexity bound in their setting.

As usual in the displacement rank approach, we wish to work with a matrix $\vA$ defined by a compressed displacement representation. We consider square matrices for simplicity, although as noted in Section~\ref{subsec:rectangular} the techniques work for rectangular matrices as well. Definition~\ref{defn:displacement} formally lays out the representation of $\vA$.
\begin{defn}
  \label{defn:displacement}
  Suppose we are given the following:
  \begin{itemize}
    \item $\vR \in \F^{N \times N}$ that is $(\alpha,\beta)$-Krylov efficient and such that we know its characteristic polynomial $c_\vR(X)$,
    \item $\vL \in \F^{N \times N}$ that is triangular (throughout this section, we will assume it is lower triangular) and $(\Delta+1)$-band,
    \item $\vC \in \F^{N \times r}$ and $\vD \in \F^{r \times N}$, generators for a low rank matrix
    \item a displacement operator $D_{\vL,\vR} \in \{\sylv{\vL}{\vR}, \stein{\vL}{\vR}\}$. The Sylvester operator is defined as $\sylv{\vL}{\vR} : \vA \mapsto \vL\vA-\vA\vR$, and the Stein operator is $\stein{\vL}{\vR} : \vA \mapsto \vA-\vL\vA\vR$.
  \end{itemize}
  Assume that $\vA \in \F^{N \times N}$ is implicitly and uniquely defined by the equation $D_{\vL,\vR}(\vA) = \vC\vD$.
  (For the Sylvester displacement operator, the last condition is equivalent to $\vL$ and $\vR$ not sharing eigenvalues. For the Stein displacement operator, the last condition is equivalent to the set of $\vR$'s eigenvalues being disjoint from the reciprocals of $\vL$'s eigenvalues~\cite{simoncini2016computational}.)
\end{defn}

\paragraph{Sylvester displacement.}
First suppose that $\vA$ is defined according to a Sylvester displacement equation.
We will show that the rows of $\vA$ satisfy a standard $\vR$-matrix recurrence as in~\eqref{eq:recur}. Let $\vd_0, \dots, \vd_{r-1} \in \F^N$ be the rows of $\vD$ (vectorized as a column vector), so that every row of $D_{\vL,\vR}(\vA)$ (vectorized) can be written as a linear combination of the basis $\vd_0, \dots, \vd_{r-1}$.

Expanding and rearranging the $i$th row of $\vL\vA - \vA\vR = \vE$ yields
\begin{equation}
\label{eq:recur-dp}
\begin{aligned}
  &\sum_{j=0}^t \vL[i,i-j]\vA[i-j,:] - \vA[i,:]\vR = \vE[i,:] \\% \nonumber \\
  &\vA[i,:](\vL[i,i]\vI-\vR) = \sum_{j=1}^t -\vL[i,i-j]\vA[i-j,:] + \sum_{k=0}^{r-1} c_{i,k} \vd_k
\end{aligned}
\end{equation}

By Definition~\ref{defn:recurrence-width}, this exactly defines a $\vR^T$-matrix recurrence of width $(t,r)$ and degree $(0,1)$. Note that the disjoint eigenvalue assumption means $\vL[i,i]\vI-\vR$ is invertible for all $i$. In this case $g_{i,0}(X) = \vL[i,i]-X$ and $g_{i,j}(X) = -\vL[i,i-j]$.

\paragraph{Stein displacement.}
A similar reduction can be applied for the Stein-type displacement operator. If $\vA-\vL\vA\vR = \vC\vD$, then
\[
  \vA[i,:](\vI-\vL[i,i]\vR) = \sum_{j=1}^t -\vL[i,i-j]\vA[i-j,:]\vR + \sum_{k=0}^{r-1} c_{i,k} \vd_k
\]
with the only difference being $g_{i,0}(X)$ becomes $1-\vL[i,i]X$ and $g_{i,j}(X)$ becomes $\vL[i,i-j]X$. This again defines a $\vR^T$-matrix recurrence of degree $(0,1)$ and width $(t,r)$.

Theorem~\ref{thm:recur-general} gives the complexity of superfast matrix-vector multiplication by $\vA$ and $\vA^T$ in terms of the Krylov efficiency of $\vR$. Furthermore, when $\vR$ is also triangular and $\Delta$-band, its characteristic polynomial can be computed in $O(\cM(N)\log{N})$ time and it is $(\Delta^{\omega}\cM(N)\log{N}, \Delta^2\cM(N)\log{N})$-Krylov efficient by Theorem~\ref{thm:ut-band-ke}.

%We also note that the conditions on $\vL,\vR$ are symmetric. Let $\phi$ and $\psi$ be matrix properties that apply through transposition (i.e.\ if $\vA$ satisfies $\phi$ then so does $\vA^T$). Now suppose that $\vA\vb$ and $\vA^T\vb$ admit fast multiplication algorithms when $\vL$ satisfies $\phi$ and $\vR$ satisfies $\psi$. Then they are also efficient when $\vR$ satisfies $\phi$ and $\vL$ satisfies $\psi$. This is because $\vR^T\vA^T - \vA^T\vL^T = -\vE^T$ is also rank $r$, let $\vL' = \vR^T$ and $\vR' = \vL^T$ which satisfy the appropriate properties so $\vA^T$ and $(\vA^T)^T$ admit fast multiplication.
%\ar{I don't quite follow what the above para is trying to do. I mean I understand what it is trying to say but I do not quite follow why the para proves the claim on symmetry. $\vA$ being Krylov efficient is not a property that applies through transposition. In particular, in general $\kry(\vA,\vb)^T\neq \kry(\vA^T,\vb)$, right?}
Again, we stated the reduction for lower triangular matrices, but upper triangular $\vL$ defines a similar recurrence. Finally, it is known that $\vA^T$ also has low displacement rank, with respect to $\vR^T$ and $\vL^T$, so the same reduction and algorithm works for $\vA^T$.

The above discussion and applying Theorem~\ref{thm:recur-general} implies
\begin{thm}
  \label{thm:disp-rank}
  Suppose we are given $\vL,\vR,\vC,\vD$ that define a matrix $\vA$ according to Definition~\ref{defn:displacement}. Then we can compute $\vA \vb$ and $\vA^T \vb$ for any vector $\vb$ in $O((\Delta + r)\Delta \cM(N)\log{N} + r\beta)$ operations with $O(\alpha_{\Delta,r}\cM(N)\log{N} + \alpha)$ preprocessing.
\ar{$\alpha$ is being used for two different things in the last bound above.}
\agu{I will go through and change the $\alpha_{t,r}$ notation at some point. Suggestions are welcome.}
\end{thm}
%\begin{proof}
%\end{proof}
\begin{cor}
  \label{cor:disp-rank}
  Suppose we are given $\vL,\vR,\vC,\vD$ that define a matrix $\vA$ according to Definition~\ref{defn:displacement}, and additionally suppose that $\vL$ and $\vR$ are both $\Delta$-band. Then we can compute $\vA \vb$ and $\vA^T \vb$ for any vector $\vb$ in $O(\Delta^2r\cM(N)\log{N})$ operations with $O( (\Delta^{\omega-1}(\Delta+r))\cM(N)\log{N})$ preprocessing.
\end{cor}
%\begin{proof}
%  By Theorem~\ref{thm:ut-band-ke}, $\vL$ is $(\Delta^\omega, \Delta^2)$-Krylov efficient and furthermore its characteristic polynomial can be computed in $\tO(N)$ time since it is triangular.
%\end{proof}

This finishes the proof of the first part of Theorem~\ref{thm:intro-dr}.
%\ar{use parameter instantiations, mention characteristic polynomial, etc. also metnion that this is what proved Thm 1.4. Also we use Theorem 7.1}

%A Vandermonde-like matrix $\vV$ is a matrix that has low displacement rank with respect to $\vL = \vD$, a diagonal matrix, and $\vR = \vS^T$. So the columns of $\vV$ must follow the recurrence
%$$\vf_i = \frac{1}{X}\vf_{i-1} +  \frac{c_i}{X} \otimes  \vd$$ where $p(X) \otimes \vf = p(\vD) \otimes \vf$.
%As in the example of Vandermonde-like matrices
We remark that this captures the previous displacement results in the literature before the very recent results of Bostan et al.~\cite{bostan2017}. For the four classic types, Toeplitz- and Hankel-like matrices are defined with the Stein operator and $\vL=\vS$, $\vR=\vS^T$; Vandermonde-like matrices are defined with the Sylvester operator and $\vL$ diagonal, $\vR=\vS^T$; and Cauchy-like matrices are defined with the Sylvester operator and $\vL,\vR$ diagonal. Until recently, the most general previous displacement rank results in literature had $\vL$ and $\vR$ in Jordan normal form, which were handled by Olshovsky and Shokrollahi~\cite{OS00}. The results in this section cover all $\vL$ and $\vR$ in Jordan normal form that have distinct eigenvalues. We note that it is not possible to have a single efficient algorithm for matrices with low displacement rank with respect to arbitrary $\vL, \vR$ in Jordan normal form. In particular, every matrix has low displacement rank with respect to $\vL = \vR = \vI$. In general, when $\vL$ and $\vR$ share eigenvalues, the equation $\vL\vA-\vA\vR=\vE$ does not uniquely specify $\vA$, and we hypothesize that a general algorithm will incur an extra factor roughly corresponding to the complexity of fully specifying $\vA$.

\subsection{Quasiseparable $\vL$ and $\vR$}
\label{subsec:disp-rank-quasi}
We now show how to adapt the above to more general $\vL$ and $\vR$, which in particular includes both the triangular band matrices of Corollary~\ref{cor:disp-rank} and the block companion matrices of Bostan et al.~\cite{bostan2017}. For concreteness, we focus on the Sylvester displacement $\vL\vA-\vA\vR = \vC\vD^T$, but the Stein displacement case is similar. Let us trace through the execution of the full algorithm, paying special attention to equation~\eqref{eq:recur-matrix-reduction} which we re-write here for convenience.
\[
  \vA^T\vb = \sum_{k=0}^{r-1} \kry(\vR,\vd_k) \sum_{i=0}^{N-1} \sum_{j=0}^{N-1} \vb[i] c_{jk} \vh_{ij}
\]
(where $\vH$ is from the Structure Lemma~\ref{lmm:structure1} and $\vh_{ij}$ is the coefficient vector of $\vH_{ij} \in \F[X]$).

The full algorithm for computing $\vA^T\vb$ can be summarized as follows.
Let $\vH = (\vL-X\vI)^{-1} \pmod{c_\vR(X)} \in \F[X]^{N \times N}$.
Compute $\vF = \vb^T\vH\vC \in \F[X]^{1 \times r}$.
Let $\vf_k \in \F^N$ be the coefficient vector of the $k$th element of $\vF$. The answer is $\sum_{k=0}^{r-1} \kry(\vR,\vd_k) \vf_k$.
Finally, to perform the Krylov multiplications, recall that in Section~\ref{sec:Krylov} we showed that $\kry(\vR,\vd)^T \vf$ is the coefficient vector of the polynomial $\vf^T (\vI-\vR X)^{-1} \vd \pmod{X^N}$ (analogous to the $\vb^T\vH\vc$ step of Algorithm~\ref{algo:transpose-mult}). The multiplication $\kry(\vR,\vd)\vf$ has the same complexity by the transposition principle, and an explicit algorithm can be found by using the same techniques to convert the recurrence width transpose multiplication algorithm $\vA^T\vb$ to the algorithm for $\vA\vb$ (Appendix~\ref{sec:Ab}).

Thus the multiplication $\vA^T\vb$ can be reduced to performing $O(r)$ computations of the form $\vb^T (X\vI-\vR)^{-1} \vc \pmod{M(X)}$ (or $\vb^T (\vI-\vR X)^{-1} \vc \pmod{M(X)}$, but these are algorithmically equivalent, so we focus on the former) for some $M(X)$ of degree $N$ (note that $M(X)$ will be equal to either $c_\vR(X)$ or $X^N$).\footnote{The above reduction from the Sylvester equation to resolvents is similar to various known formulae for $\vA$ based on the Sylvester equation~\cite{simoncini2016computational,lancaster1984factored}.}
It is enough to find $\vb^T (X\vI-\vR)^{-1} \vc$, which is a rational function of the form $f(X)/g(X)$ - then reducing it $\pmod{M(X)}$ requires $\cM(N)\log{N}$ steps for inverting $g(X) \pmod{M(X)}$~\cite{yap2000fundamental} and $\cM(N)$ for multiplying by $f(X) \pmod{M(X)}$. We call computing $\vb^T (X\vI-\vR)^{-1} \vc$ the resolvent problem on $\vR$. Henceforth we also let $X-\vR$ denote $X\vI-\vR$.

\subsubsection{Resolvent computation}
The most general useful class of matrices for which we know how to solve the resolvent problem in soft-linear time are the \emph{quasiseparable} matrices, introduced by Eidelman and Gohberg~\cite{eidelman1999}.
\begin{defn}
  \label{defn:quasi}
  A matrix $\vR \in \F^{N \times N}$ is $(p,q)$-quasiseparable if
  \begin{itemize}
    \item Every submatrix contained strictly below the diagonal has rank at most $p$.
    \item Every submatrix contained strictly above the diagonal has rank at most $q$.
  \end{itemize}
  A $(q,q)$-quasiseparable matrix is also called $q$-quasiseparable.\footnote{Given a $q$-quasiseparable matrix $\vR$ satisfying Definition~\ref{defn:quasi}, we will assume that we have access to a factorization of any rank-$q$ sub-matrix (or can compute one in time equal to the size of this factorization, i.e.\ $q(k+\ell)$ for a $k \times \ell$ sub-matrix). There are many efficient representations of quasiseparable matrices that allow this~\cite{eidelman1999, delvaux2007givens}, and even without one, simple randomized approaches still allow efficient computation of the generators (see Lemma~\ref{lmm:dr-compress}).}
\end{defn}

The problem we now address is given $t$-quasiseparable $\vR$, to compute the rational function $\vb^T (X-\vR)^{-1} \vc$ for any vectors $\vb,\vc$.

The idea here is that quasiseparable matrices are recursively ``self-similar'', in that the leading and trailing principal submatrices are also quasiseparable, which leads to a divide-and-conquer algorithm. Consider a quasiseparable matrix $\vR$ for which we want to compute the resolvent $(X-\vR)^{-1}$. The top left and bottom right blocks of $X-\vR$ are self-similar to $X-\vR$ itself by definition of quasiseparability. Suppose through recursion we can invert each of them, in other words compute $\diag\{X-\vR_{11}, X-\vR_{22}\}$. But by quasiseparability, $X-\vR$ is simply a low-rank perturbation of $\diag\{X-\vR_{11}, X-\vR_{22}\}$, and so by standard techniques we can compute $(X-\vR)^{-1}$:

%\ar{Put in a filler line here}
\begin{prop}[Binomial Inverse Theorem / Woodbury matrix identity~\cite{woodbury1950}]
\label{prop:woodbury}
  Over a commutative ring $\cR$, let $\vA \in \cR^{N \times N}$ and $\vU,\vV \in \cR^{N \times p}$. Suppose $\vA$ and $\vA+\vU\vV^T$ are invertible. Then $\vI_p + \vV^T\vA^{-1}\vU$ is invertible and
  \[
    (\vA + \vU\vV^T)^{-1} = \vA^{-1} - \vA^{-1}\vU(\vI_p + \vV^T\vA^{-1}\vU)^{-1}\vV^T\vA^{-1}
  \]
\end{prop}

For our purposes, $\cR$ will be the ring of rational functions over $\F$. Now we can prove the following.
\begin{lmm}
\label{lem:quasi-transpose} 
  Let $\vR$ be a $t$-quasiseparable matrix. Then $\vb^T (X-\vR)^{-1}\vc$ for any scalar vectors $\vb,\vc$ can be computed in $O(t^\omega \cM(N) \log^2 N + t^2\cM(N) \log^3{N})$ operations.
\end{lmm}

\begin{proof}
  More generally, we will consider computing $\vB^T(X-\vR)^{-1}\vC$ for matrices $\vB \in \F^{N \times k}$ and $\vC \in \F^{N \times k}$. Note that the result is a $k\times k$ matrix of rational functions of degree at most $N$ on top and bottom.

  Let $\vR$ be partitioned into submatrices $\vR_{11},\vR_{12},\vR_{21},\vR_{22} \in \left(\F(X)\right)^{N/2 \times N/2}$ in the usual way. Since $\vR$ is $t$-quasiseparable, we can write $\vR_{21} = \vU_L \vV_L^T$ and $\vR_{12} = \vU_U \vV_U^T$ where $\vU_\cdot, \vV_\cdot \in \F^{N \times t}$. Notice that we can write $X-\vR$ as
  \[
    X-\vR =
    \begin{bmatrix} X-\vR_{11} & \vzero \\ \vzero & X-\vR_{22} \end{bmatrix}
    +
    \begin{bmatrix} \vzero & \vU_U \\ \vU_L & \vzero \end{bmatrix}
    \begin{bmatrix} \vV_L & \vzero \\ \vzero & \vV_U \end{bmatrix}^T.
  \]

  Suppose we know the expansions of each of
  \begin{align}
    & \vM_1 = \vB^T \begin{bmatrix} X-\vR_{11} & \vzero \\ \vzero & X-\vR_{22} \end{bmatrix}^{-1} \vC \\
    & \vM_2 = \vB^T
    \begin{bmatrix} X-\vR_{11} & \vzero \\ \vzero & X-\vR_{22} \end{bmatrix}^{-1}
    \begin{bmatrix} \vzero & \vU_U \\ \vU_L & \vzero \end{bmatrix} \\
    & \vM_3 = \begin{bmatrix} \vV_L & \vzero \\ \vzero & \vV_U \end{bmatrix}^T
    \begin{bmatrix} X-\vR_{11} & \vzero \\ \vzero & X-\vR_{22} \end{bmatrix}^{-1}
    \begin{bmatrix} \vzero & \vU_U \\ \vU_L & \vzero \end{bmatrix} \\
    & \vM_4 = \begin{bmatrix} \vV_L & \vzero \\ \vzero & \vV_U \end{bmatrix}^T
    \begin{bmatrix} X-\vR_{11} & \vzero \\ \vzero & X-\vR_{22} \end{bmatrix}^{-1}
    \vC.
  \end{align}
  These have dimensions $k \times k, k \times 2t, 2t \times 2t, 2t \times k$ respectively (and entries bounded by degree $N/2$ on the top and bottom). Also note that all the above inverses exist because a matrix of the form $\vX-\vR$ for $\vR \in \F$ has non-zero determinant.

  By Proposition~\ref{prop:woodbury}, the desired answer is
  \[
    \vB^T (X-\vR)^{-1} \vC = \vM_1 - \vM_2(\vI_{2t}+\vM_3)^{-1}\vM_4.
  \]
  %\begin{align*}
  %  \vB^T \begin{bmatrix} X-\vR_{11} & \vzero \\ \vzero & X-\vR_{22} \end{bmatrix}^{-1} \vC
  %  - \\
  %  &
  %  \vB^T
  %  \begin{bmatrix} X-\vR_{11} & \vzero \\ \vzero & X-\vR_{22} \end{bmatrix}^{-1}
  %  \begin{bmatrix} \vzero & \vU_U \\ \vU_L & \vzero \end{bmatrix} \times \\
  %  & \left( \vI_{2t} + 
  %  \begin{bmatrix} \vV_L & \vzero \\ \vzero & \vV_U \end{bmatrix}^T
  %  \begin{bmatrix} X-\vR_{11} & \vzero \\ \vzero & X-\vR_{22} \end{bmatrix}^{-1}
  %  \begin{bmatrix} \vzero & \vU_U \\ \vU_L & \vzero \end{bmatrix}
  %  \right)^{-1} \times \\
  %  & \begin{bmatrix} \vV_L & \vzero \\ \vzero & \vV_U \end{bmatrix}^T
  %  \begin{bmatrix} X-\vR_{11} & \vzero \\ \vzero & X-\vR_{22} \end{bmatrix}^{-1}
  %  \vC.
  %\end{align*}
  
Then the final result can be computed by inverting $\vI_{2t}+\vM_3$ ($O(t^\omega \cM(N))$ operations), multiplying by $\vM_2,\vM_4$ ($O(k/t \cdot t^\omega \cM(N))$ operations each\footnote{When $k < t$ this term is subsumed by the previous one anyways}), and subtracting from $\vM_1$ ($O(k^2 \cM(N))$ operations). This is a total of $O( (t^\omega + kt^{\omega-1} + k^2) \cM(N) )$ operations. Note that when $k = O(t\log N)$, this becomes $O(t^\omega \cM(N)\log{N} + t^2 \cM(N) \log^2{N})$; we will use this in the analysis shortly.

  To compute $\vM_1,\vM_2,\vM_3,\vM_4$, it suffices to compute the following:
    \begin{align*}
      \vB_1^T (X-\vR_{11})^{-1}\vC_1 \qquad& \vB_2^T (X-\vR_{22})^{-1}\vC_2 \\
      \vB_1^T (X-\vR_{11})^{-1}\vU_U  \qquad&\vB_2^T (X-\vR_{22})^{-1}\vU_L \\
      \vV_L^T (X-\vR_{11})^{-1}\vU_U \qquad& \vV_U^T (X-\vR_{22})^{-1}\vU_L \\
      \vV_L^T (X-\vR_{11})^{-1}\vC_1 \qquad& \vV_U^T (X-\vR_{22})^{-1}\vC_2.
    \end{align*}
  But to compute those, it suffices to compute the following $(k+t) \times (k+t)$ matrices:
  \begin{align*}
    & \begin{bmatrix} \vB_1 & \vV_L \end{bmatrix}^T
    (X-\vR_{11})^{-1} 
    \begin{bmatrix} \vC_1 & \vU_U \end{bmatrix} \\
    & \begin{bmatrix} \vB_2 & \vV_U \end{bmatrix}^T
    (X-\vR_{22})^{-1} 
    \begin{bmatrix} \vC_2 & \vU_L \end{bmatrix} \\
  \end{align*}
  % [idea: the answer is some linear combination of the above. a linear combination of the left side can be written directly as some sum of $\vB_1$ and $\vV_L$ on the left, and some sum of $\vC_1$ and $\vU_U$ on the right. Perhaps we can pass the combine step into the recursion instead]

  Since $\vR_{11}$ and $\vR_{22}$ have the same form as $\vR$, this is two recursive calls of half the size. Notice that the size of the other input (dimensions of $\vB,\vC$) is growing, but when the initial input is $k=1$, it never exceeds $1+t\log{N}$ (since they increase by $t$ every time we go down a level). Earlier, we noticed that when $k = O(t\log N)$, the reduction step has complexity $O( t^\omega \cM(N)\log{N} + t^2 \cM(N) \log^2{N})$ for any recursive call. As usual, the complete runtime is a $\log{N}$ multiplicative factor on top of this.
\end{proof}

Combining the aforementioned reduction from the Sylvester equation to resolvents with this algorithm proves the multiplication part of Theorem~\ref{thm:intro-dr-quasi}. The full algorithm is detailed in Algorithm~\ref{algo:resolvent}.

\begin{algorithm}
  \caption{\textsc{Resolvent}}
  \label{algo:resolvent}
  \begin{algorithmic}[1]
    \renewcommand{\algorithmicrequire}{\textbf{Input:}}
    \renewcommand{\algorithmicensure}{\textbf{Output:}}
    \Require{ $\vB = \begin{bmatrix} \vB_1 & \vB_2 \end{bmatrix}$}
    \Require{$\vR$ quasiseparable such that $\vR_{12} = \vU_U\vV_U^T$, $\vR_{21} = \vU_L\vV_L^T$}
    \Require{ $\vC = \begin{bmatrix} \vC_1 & \vC_2 \end{bmatrix}$}
    \Ensure{$\vB^T (X-\vR)^{-1} \vC$}
    \If{$\mathsf{dim}(\vR) = 1$}
    \State \Return{$\vB^T\vC / (X-\vR[0,0])$}
    \Else
    %\State $\vB = \begin{bmatrix} \vB_1 & \vB_2 \end{bmatrix}$
    %\State $\vC = \begin{bmatrix} \vC_1 & \vC_2 \end{bmatrix}$
    %\State $\vR_{12} = \vU_U\vV_U^T$
    %\State $\vR_{21} = \vU_L\vV_L^T$
    %\State $\vB_L \gets \begin{bmatrix} \vB_1 & \vV_L \end{bmatrix}$
    %\State $\vC_L \gets \begin{bmatrix} \vC_1 & \vU_U \end{bmatrix}$
    %\State $\vB_R \gets \begin{bmatrix} \vB_2 & \vV_U \end{bmatrix}$
    %\State $\vC_R \gets \begin{bmatrix} \vC_C & \vU_L \end{bmatrix}$
    \State $\begin{bmatrix} \vM_{11} & \vM_{22} \\ \vM_{41} & \vM_{31} \end{bmatrix} \gets \Call{Resolvent}{\left[ \vB_1 \quad \vV_L \right], \vR_{11}, \left[ \vC_1 \quad \vU_U \right]}$
    \State $\begin{bmatrix} \vM_{12} & \vM_{21} \\ \vM_{42} & \vM_{32} \end{bmatrix} \gets \Call{Resolvent}{\left[ \vB_2 \quad \vV_U \right], \vR_{11}, \left[ \vC_C \quad \vU_L \right]}$
    \State $\vM_1 = \vM_{11} + \vM_{12}$
    \State $\vM_2 = \begin{bmatrix} \vM_{21} & \vM_{22} \end{bmatrix}$
    \State $\vM_3 = \begin{bmatrix} \vzero & \vM_{31} \\ \vM_{32} & \vzero \end{bmatrix}$
    \State $\vM_4 = \begin{bmatrix} \vM_{41} \\ \vM_{42} \end{bmatrix}$
    \State \Return{$\vM_1 - \vM_2(\vI_{2t}+\vM_3)^{-1}\vM_4$}
    \EndIf
  \end{algorithmic}
\end{algorithm}

We note that the bounds in Lemma~\ref{lem:quasi-transpose} are slightly worse in the exponent of $t$ and the number of $\log{N}$ factors, compared to the bounds derived from the recurrence width algorithm as in Corollary~\ref{cor:disp-rank}. A more detailed analysis that isolates operations independent of $\vb$ as pre-computations should be possible to bridge the gap between these bounds, and is left for future work.

\section{Inverses and Solvers}
\label{sec:inverse}

In this section, we address the inverses of the structured matrices we consider. The two important cases are for matrices of low recurrence width from Definition~\ref{defn:intro-width}, and all types of matrices of low displacement rank. These will be addressed slightly differently. For the standard matrices of low recurrence width, we find a \emph{solver}. That is given invertible $\vA$ and input $\vy$, we compute $\vx$ such that $\vA\vx = \vy$.

For matrices $\vA$ of low displacement rank, we find their inverse. More specifically, it is known that the inverse of a matrix with low displacement rank also has low displacement rank. Thus the natural problem of inverting displacement structured matrices is: given the compressed displacement parameterization of $\vA$, find the compressed displacement rank parameterization of $\vA^{-1}$. We show that the classic techniques~\cite{kaltofen1994} can be extended to our class of more generalized displacement structure.

\subsection{Solvers for Matrices of low Recurrence Width}
\label{subsec:further:inverse}

Consider a matrix $\vA$ in the setting of Definition~\ref{defn:intro-width}. In this section we show how to compute $\vA^{-T}\vb$ for any vector $\vb$. In other words, we find $\vx$ satisfying $\vA^T\vx=\vy$ given $\vy$. The main idea we use is that sub-matrices of $\vA$ have the same structure as $\vA$, and to solve the equation $\vA^T\vx = \vy$ it suffices to perform forward multiplication on sub-matrices of $\vA^T$ using Algorithm~\ref{algo:transpose-mult}. This algorithm provides a solver for $\vA$ as well, where the invocations of Algorithm~\ref{algo:transpose-mult} are replaced with a $\vA\vb$ multiplication algorithm (Appendix~\ref{sec:Ab}) or directly via the transposition principle.

Consider a matrix $\vA$ satisfying Definition~\ref{defn:intro-width}, and suppose that $\vA$ is invertible, i.e.\ $\deg(f_i) = i$ for $i = 0, \dots, N-1$. Note in particular that $\vA^T$ is upper triangular. As usual, let $\vA$ be blocked as $\vA_{11},\vA_{12},\vA_{21},\vA_{22}$ in the usual way.

\begin{thm}
  \label{thm:recur-inverse}
  For any invertible matrix $\vA$ satisfying Definition~\ref{defn:intro-width},
  with $O(t^{\omega} \cM(N)\log{N})$ pre-processing operations, a solution $\vx$ to $\vA^T\vx = \vy$ can be found for any $\vy$ with $O(t^2\cM(N)\log^2{N})$ operations over $\F$.
%For an $\vR$-matrix recurrence~\eqref{eq:intro-recur-matrix} \ar{Maybe link the appropriate definition here?} that is nice and $\vR$ is $(\alpha,\beta)$-Krylov efficient, pre-processing and computation need additional $\tO(t\alpha N)$ and $\tO(t\beta N)$ operations respectively.
\end{thm}

The basic idea is to use a divide-and-conquer algorithm to reduce the problem of multiplying by the inverse of $\vA$ into multiplying by the inverses of $\vA_{11}$ and $\vA_{22}$. One issue is that $\vA_{22}$ does not have the exact same structure as $\vA$ so we cannot directly use a recursive call. To get around this, we will consider a slightly more general setting that includes the structure of both $\vA$ and $\vA_{22}$.

Instead of running a recurrence and extracting the low order $X^0 \dots X^{n-1}$ coefficients into a matrix, we extract the high order $X^n \dots X^{2n-1}$ coefficients. Formally, let $f_0(X), \dots, f_{t-1}(X)$ and $g_{i,j}(X), i \in [N], j \in [t]$ be given such that
\begin{equation}
  \label{eq:fg}
  \begin{aligned}
    &\deg(f_i) \le N+i \\
    &\deg(g_{i,j}) \le j
  \end{aligned}
\end{equation}
and define $f_t(X), \dots, f_{N-1}(X)$ according to recurrence~\eqref{eq:intro-recur-poly}. Note that now $\deg(f_i) \leq N+i$ for all $i$. Define matrix $\vA \in \F^{N \times N}$ such that $\vA[j,i]$ is the coefficient of degree $N+j$ in $f_i(X)$. For shorthand, let this construction be denoted $\vA = \bar\vA(\{f_i\},\{g_{i,j}\})$.

We shall prove
\begin{lmm}
  \label{lmm:inv}
  Given $f_0(X), \dots, f_{t-1}(X)$ and $g_{i,j}(X)$ satisfying~\eqref{eq:fg} and let $\vA = \bar\vA(\{f_i\},\{g_{i,j}\})$ be invertible. With $O(t^{\omega} \cM(N)\log{N})$ pre-processing operations, a solution $\vx$ to $\vA^T\vx = \vy$ can be found for any $\vy$ with $O(t^2\cM(N)\log^2{N})$ operations over $\F$.
\end{lmm}

This is stronger than Theorem~\ref{thm:recur-inverse}.
\begin{proof}[Proof of Theorem~\ref{thm:recur-inverse}]
  Given $\vA$ defined by $f_0(X), \dots, f_{t-1}(X)$ and $g_{i,j}(X)$ satisfying Definition~\ref{defn:intro-width}, note that $X^Nf_0(X), \dots, X^Nf_{t-1}(X)$ and $g_{i,j}(X)$ satisfy~\eqref{eq:fg}.
  Furthermore, $\vA = \bar\vA(\{X^Nf_i\},\{g_{i,j}\})$. Applying Lemma~\ref{lmm:inv} gives the result.
\end{proof}

We show Lemma~\ref{lmm:inv} using a standard divide-and-conquer algorithm by partitioning $\vA^T$ into blocks. In order to solve
\[
    \vA^T\vx = \begin{bmatrix}\vA^T_{11} & \vA^T_{12} \\ \vzero & \vA^T_{22} \end{bmatrix}\vx = \vy,
\]
three steps are required. First, we see $(\vA_{22})^T\vx_2 = \vy_2$, so solving for $\vx_2$ is a recursive call. Then we have $\vA_{11}^T \vx_1 + (\vA^T)_{12}\vx_2 = \vy_1$. Note that multiplication by $\vA^T_{12}$ can be done with a call to Algorithm~\ref{algo:transpose-mult} (padding the input with $0$s), so $\vz = \vy_1 - (\vA^T)_{12}\vx_2$ can be computed efficiently. Finally, $\vA_{11}^T\vx_1 = \vz$, so solving for $\vx_1$ is another recursive call.

In order to run the above, it suffices to show that the standard block decomposition of $\vA$ yields matrices of the same structure.

\begin{lmm}
  \label{lmm:inv-recursive}
  Given $\{f_i\}, \{g_{i,j}\}$ satisfying~\eqref{eq:fg} and $\vA = \bar\vA(\{f_i\},\{g_{i,j}\})$, let $\vA' = \vA[0:N/2, 0:N/2]$ or $\vA[N/2:N, N/2:N]$. There exists $\{\hat{f}_i\}, \{\hat{g}_{i,j}\}$ satisfying~\eqref{eq:fg} (for dimensions $N/2$ instead of $N$) such that $\vA' = \bar\vA(\{\hat{f}_i\}, \{\hat{g}_{i,j}\})$. That is, the two triangular subblocks of $\vA$ have the same structure as itself.
\end{lmm}
\begin{proof}
  First consider $\vA[0:N/2,0:N/2]$. The idea is that knowing the higher-order coefficients of $f_0(X), \dots, f_{t-1}(X)$ is enough to recover $\vA[0:N/2,0:N/2]$ because the degree condition on $g_{i,j}(X)$ means the low-order coefficients of $f_{0:t}(x)$ cannot influence this submatrix. Define
  \[
    \floor{p(X)}_i = {p(X) - (p(X)\pmod{X^i}) \over X^i},
  \]
  i.e.\ the polynomial formed by coefficients of $p(X)$ of order at least $X^i$.

  It is easy to show
  \begin{align*}
    \vA[0:N/2,0:N/2] &= \bar\vA\left(\left\{\floor{f_i}_{N/2} : i<t\}, \{g_{i,j}:i<N/2\right\} \right) \\
    \vA[N/2:N,N/2:N] &= \bar\vA\left(\left\{\floor{f_i}_{N} : i \in [N/2:N/2+t] \}, \{g_{i,j}:i\ge N/2\right\} \right)
  \end{align*}

  Finally, note that $\{f_i: i \in [N/2 : N/2+t]\}$ are easy to compute because
  \[
    \begin{bmatrix} f_{N/2} \\ \vdots \\ f_{N/2+t-1} \end{bmatrix}
    =
    \vT_{[0:N/2]}
    \begin{bmatrix} f_{0} \\ \vdots \\ f_{t-1} \end{bmatrix}
  \]
  (see Lemma~\ref{lmm:structure2}). This matrix-vector multiplication by $\vT_{[0:N/2]}$ is $t^2$ polynomial multiplications.
\end{proof}

\begin{algorithm}
  \begin{algorithmic}[1]
    \renewcommand{\algorithmicrequire}{\textbf{Input:}}
    \renewcommand{\algorithmicensure}{\textbf{Output:}}
    \Require{$\vT_{[\ell:r]}$ for all dyadic intervals}
    \Require{$\vA$ parameterized by $\{f_i : i \in [t]\}, \{g_{i,j}: i \in [N], j \in [t]\}$}
    \Require{$\vy \in \F^N$}
    \Ensure{$\vx = \bar\vA\left( \{f_i\},\{g_{i,j}\} \right)^{-1}\vy$}
    %\Function{Solve}{$\{f_i : i \in [t]\}, \{g_{ij}: i \in [N], j \in [t]\}$, $\vy$}
    \State $\vx_1 \gets \Call{InverseMult}{\left\{\floor{f_i}_{N/2} : i<t\}, \{g_{i,j}:i<N/2\right\}, \vy_2}$
    \State $\vz \gets \Call{TransposeMult}{\vA_{12}^T, \vx_1}$
    \label{step:inv:mult}
    \State $\vx_2 \gets  \Call{InverseMult}{\left\{\floor{f_i}_{N} : i \in [N/2:N/2+t] \}, \{g_{i,j}:N/2\le i\right\}, \vz}$
    \label{step:inv:recurse2}
    \State \Return{$(\vx_1, \vx_2)$}
    %\EndFunction
  \end{algorithmic}
  \caption{\textsc{InverseMult}}
  \label{algo:inv}
\end{algorithm}

We present the entire algorithm in Algorithm~\ref{algo:inv}.
This entire algorithm only requires the same pre-processing step as Algorithm~\ref{algo:transpose-mult}, which takes $O(t^\omega \cM(N)\log{N})$ operations by Lemma~\ref{cor:transpose-mult-pre}. Given this pre-processing, Step~\ref{step:inv:mult} requires $O(t^2\cM(N)\log{N})$ operations by Theorem~\ref{thm:mainATb}, and Step~\ref{step:inv:recurse2} requires $O(t^2\cM(N))$ operations as mentioned in the proof of Lemma~\ref{lmm:inv-recursive}. Therefore the entire algorithm requires $O(t^2\cM(N)\log^2{N})$ operations by the standard recursive analysis. This proves Lemma~\ref{lmm:inv}.

\subsection{Inverses of Matrices of Low Displacement Rank}
Along with Definition~\ref{defn:intro-width}, the other main class of matrices that recurrence width generalizes is those of low displacement rank. In this section we consider the problem of matrix-vector multiplication by the inverses of matrices of low displacement rank, or ``inverse multiplication'' for short.

The problem of inverting a matrix $\vA$ with low displacement rank is well-understood when $\vA$ is Toeplitz-like, Hankel-like, Vandermonde-like, or Cauchy-like~\cite{pan2000nearly}. These algorithms are generally variants of the Morf/Bitmead-Anderson algorithm~\cite{bitmead1980}. Given a matrix of low displacement rank, its inverse also has low displacement rank\footnote{e.g.\ for Sylvester-type displacement: If $\rank(\vL\vA-\vA\vR) = r$, then $\rank(\vR\vA^{-1}-\vA^{-1}\vL) = \rank(\vA^{-1}(\vL\vA-\vA\vR) \vA^{-1}) = r$.}, so it suffices to compute a representation of the displacement structure.

\subsubsection{Triangular $\vL,\vR$}
A general methodology for computing inverses of these matrices is due to Jeannerod and Mouilleron~\cite{jeannerod2010}, who provide an algorithm for computing specified generators of the inverse of a matrix $\vA$ with low displacement rank with respect to triangular $\vL$ and $\vR$. Their framework is general enough to be applicable to the triangular banded matrices of Theorem~\ref{thm:intro-dr}. Analogously to the result of Section~\ref{subsec:further:inverse}, their algorithm reduces to computing matrix-vector multiplication by $\vA$.

\begin{thm}[\cite{jeannerod2010}]
  Consider a class $\calC$ of square matrices such that if $\vL \in \calC$, then $\vL$ is triangular and every principal and trailing square sub-matrix (e.g.\ $\vL[0:i, 0:i]$) also lies in $\calC$.
  
  Consider $\vL,\vR \in \F^{N \times N}$ that belong to $\calC$ and $\vG,\vH \in \F^{N \times r}$ such that the equation $\vL\vA-\vA\vR = \vG\vH^T$ defines a strongly regular matrix $\vA$.\footnote{A matrix is strongly regular if every principal minor is non-zero. This is typically assumed in the displacement rank literature because a reduction can be made with an efficient probabilistic preconditioning step~\cite{kaltofen1994} (see also Lemma~\ref{lmm:dr-compress}).} Note that $\vR\vA^{-1}-\vA^{-1}\vL = \vA^{-1}(\vL\vA-\vA\vR) \vA^{-1} = (\vA^{-1}\vG)(\vA^{-T}\vH)^{-1}$, i.e.\ $\vA^{-1}$ has low displacement rank with generators $\vA^{-1}\vG, \vA^{-T}\vH$. These generators can be computed in $O(f_{\calC,r}(N) \log{N})$ operations, where $f_{\calC,r}(N)$ is the cost of multiplying $\vA\vB$ where $\vA \in \F^{N \times N}$ has displacement rank $r$ w.r.t. $\vL,\vR \in \calC$ and $\vB \in \F^{N \times r}$.
\end{thm}

For our purposes, we consider $\calC$ to be the triangular $(\Delta+1)$-band matrices. Corollary~\ref{cor:disp-rank} gives the complexity of multiplication for matrices of low displacement rank with respect to this class. Finally, we note that although stated for Sylvester-type displacement, their technique also works for the Stein-type displacement operator. Therefore all the matrices of low displacement rank that we cover in Section~\ref{sec:disp-rank} also admit fast inverses. Note that the following runtime includes the cost of pre-processing $\vA$ to be able to run the matrix-vector multiplication algorithm needed above.
\begin{cor}
  \label{cor:disp-inverse}
  Let $\vL,\vR$ be triangular $(\Delta+1)$-band matrices and $\vA$ be a matrix of displacement rank $r$ with respect to $\vL,\vR$. Generators for $\vA^{-1}$ can be found in $O(\Delta^{\omega-1}(\Delta+r)\cM(N)\log{N} + \Delta^2 r^2 \cM(N) \log^2{N})$ operations.
\end{cor}

Given a displacement representation of $\vA^{-1}$, the product $\vA^{-1}\vb$ can be computed in $O(r\Delta^2 \cM(N)\log{N})$ by Corollary~\ref{cor:disp-rank}.

For constant $\Delta$, Corollary~\ref{cor:disp-inverse} matches the classic bounds of $\tO(r^2 N)$ for Toeplitz,Hankel,Cauchy,Vandermonde-like matrices~\cite{pan2000nearly}.
%A better runtime of $\tO(r^{\omega-1}N)$ was derived by Boston et al.~\cite{bostan2007} but not matched by Jeannerod and Mouilleron~\cite{jeannerod2010}.
Recent works have shown that the inverses of these (and generalizations to block companion matrices) can be computed in $\tO(r^{\omega-1}N)$ time~\cite{bostan2007,bostan2017}.
Improving Corollary~\ref{cor:disp-inverse} to match this bound is a question for further exploration.

\subsubsection{Quasiseparable $\vL$, $\vR$}
For our most general class of displacement rank, since the displacements $\vL$ and $\vR$ are no longer triangular, there are no existing results that cover this setting. However, the standard Morf/Bitmead-Anderson technique~\cite{bitmead1980} of computing generators for the inverses of matrices with low displacement rank and its extensions to other matrices with displacement structure~\cite{P01} will still work, with small modifications. The general approach follows Pan's ``compress, operate, decompress'' motto for directly manipulating the compact generators of the displacement structure~\cite{P01}.

The problem is given $\vL\vA-\vA\vR = \vG\vH^T$, to find any $\vG', \vH' \in \F^{n \times r}$ such that $\vR\vA^{-1}-\vA^{-1}\vL = \vG'\vH'^T$. We sketch here a simplified randomized algorithm to illustrate the process and give an overview of the proof here in order to analyze the complexity in our case; for more details of this type of algorithm, see~\cite{kaltofen1994}.
More modern unified algorithms for inverting classic displacement structured matrices exist~\cite{pan2000nearly,P01,jeannerod2010,HLS17} and can still be adapted to this setting with the same asymptotic bounds.

The following Lemma, due to Kaltofen~\cite{kaltofen1994}, provides the randomization step that allows generators to be compressed.
\begin{lmm}[\cite{kaltofen1994}]
  \label{lmm:dr-compress}
  Given $\sylv{\vL}{\vR}(\vA) = \vG\vH^T$ with $\vG,\vH \in \F^{n\times r}$, and knowing $\rank \sylv{\vL}{\vR}(\vA) = p$, then we can compute $\vG',\vH' \in \F^{n \times p}$ s.t. $\sylv{\vL}{\vR}(\vA) = \vG'\vH'^T$ in $O(rpN + p\cM(N))$ operations. The algorithm is randomized and requires $2N-2$ uniformly random elements from a set $S \subset \F$, and is correct with probability at least $1-r(r+1)/|S|$.
\end{lmm}

\begin{algorithm}
  \begin{algorithmic}[1]
    \renewcommand{\algorithmicrequire}{\textbf{Input:}}
    \renewcommand{\algorithmicensure}{\textbf{Output:}}
    \Require{$\vL,\vR$ $t$-quasiseparable, $\vG,\vH\in\F^{N\times r}$ s.t. $\vG\vH^T = \sylv{\vL}{\vR}(\vA)$}
    \Ensure{$\vG',\vH'\in\F^{N\times r}$ s.t. $\vG'\vH'^T = \sylv{\vR}{\vL}(\vA^{-1})$}
    \State Compute $\vG_{ij},\vH_{ij}\in\F^{N \times O(r+t)}$ s.t. $\sylv{\vL_{ii}}{\vR_{jj}}(\vA_{ij}) = \vG_{ij}\vH_{ij}^T$
    \label{step:dr-inverse:1}
    \State $\vG_{11}',\vH_{11}' \gets \Call{DisplacementInverse}{\vL_{11},\vR_{11},\vG_{11},\vH_{11}}$
    \label{step:dr-inverse:2}
    \State Compute $\vG_{S},\vH_{S}\in\F^{N \times O(r+t)}$ s.t. $\sylv{\vL_{22}}{\vR_{11}}(\vS) = \vG_{S}\vH_{S}^T$
    \label{step:dr-inverse:3}
    \State $\vG_{S}',\vH_{S}' \gets \Call{DisplacementInverse}{\vL_{22},\vR_{22},\vG_{S},\vH_{S}}$
    \label{step:dr-inverse:4}
    %\State Compute size $O(t+r)$ generators of $\sylv{\vL_{ii}}{\vR_{jj}}(\vA_{ij})$
    %\State Compute size $O(t+r)$ generators of $\sylv{\vR_{11}}{\vL_{11}}(\vA_{11})$ with $\Call{DisplacementInverse}{\vL_{11},\vR_{11},\vG_{11},\vH_{11}}$
    %\State Compute size $O(t+r)$ generators of $\sylv{\vL_{22}}{\vR_{11}}(\vS)$
    \State Compute size $O(t+r)$ generators of $\sylv{\vL_{11}}{\vR_{11}}(\vB_{11})$ where $\vB_{11} = \vA_{11}^{-1} - \vA_{11}^{-1}\vA_{12}\vS^{-1}\vA_{21}\vA_{11}^{-1}$
    \label{step:dr-inverse:5}
    \State Compute size $O(t+r)$ generators of $\sylv{\vL_{11}}{\vR_{22}}(\vB_{12})$ where $\vB_{12} = \vA_{11}^{-1}\vA_{12}\vS^{-1}$
    \label{step:dr-inverse:6}
    \State Compute size $O(t+r)$ generators of $\sylv{\vL_{22}}{\vR_{11}}(\vB_{21})$ where $\vB_{21} = \vS^{-1}\vA_{21}\vA_{11}^{-1}$
    \label{step:dr-inverse:7}
    \State Compute size $O(t+r)$ generators of $\sylv{\vR}{\vL}(\vA^{-1})$ where $\vA^{-1} = \begin{bmatrix} \vB_{11} & \vB_{12} \\ \vB_{21} & \vS^{-1} \end{bmatrix}$
    \label{step:dr-inverse:8}
    \State Use Lemma~\ref{lmm:dr-compress} to compute size $r$ generators of $\sylv{\vR}{\vL}(\vA^{-1})$
  \end{algorithmic}
  \caption{\textsc{DisplacementInverse}}
  \label{algo:dr-inverse}
\end{algorithm}

The idea behind the inversion algorithm is to use the idea of Schur complements to reduce properties about $\vA$ and $\vA^{-1}$ into properties about $\vA_{11}$ and $\vS = \vA_{22}-\vA_{21}\vA_{11}^{-1}\vA_{12}$. In particular, $\vS$ can be found quickly from $\vA$, and then $\vA^{-1}$ can be expressed completely in terms of $\vA_{12},\vA_{21},\vA_{11}^{-1}$, and $\vS^{-1}$~\cite{MA}. We will also need the well-known multiplication identity for displacement rank~\cite{pan2002structured}
\begin{equation}
  \label{eq:dr-multiplication}
  \sylv{\vL}{\vR}(\vA\vB) = \sylv{\vL}{\vM}(\vA)\vB + \vA\sylv{\vM}{\vR}(\vB).
\end{equation}
In particular, if $\vA$ and $\vB$ have low displacement rank with ``compatible'' types, then $\vA\vB$ also has low displacement rank, and computing generators of it can be done with a small number of matrix-vector multiplications by $\vA$ and $\vB$.
The algorithm is sketched in Algorithm~\ref{algo:dr-inverse}.

\begin{lmm}
  Suppose that $\vL,\vR$ are both $t$-quasiseparable and we have a rank $r$ factorization of $\sylv{\vL}{\vR}(\vA)$. We can compute a rank $r$ factorization of $\sylv{\vR}{\vL}(\vA^{-1})$ in $\tO( (t+r)^2t^\omega N)$ operations.
\end{lmm}
\begin{proof}
  We analyze the runtime of Algorithm~\ref{algo:dr-inverse}. Let the notation $M_t(N,r)$ denote the cost of matrix-vector multiplication for a matrix of displacement rank $r$ with respect to $t$-quasiseparable matrices (see Theorem~\ref{thm:intro-dr-quasi}). We defer analyzing the cost of the recursive calls until the end.
  \begin{description}
    \item[Step~\ref{step:dr-inverse:1}]
      Since $\vL\vA - \vA\vR = \vG\vH^T$, we have
      \[
        \begin{bmatrix} \vL_{11} & \vzero \\ \vzero & \vL_{22} \end{bmatrix} \vA - \vA  \begin{bmatrix} \vR_{11} & \vzero \\ \vzero & \vR_{22} \end{bmatrix} =
        \vG\vH^T - \begin{bmatrix} \vzero & \vL_{12} \\ \vL_{21} & \vzero \end{bmatrix}\vA + \vA\begin{bmatrix} \vzero & \vR_{12} \\ \vR_{21} & \vzero \end{bmatrix}
      \]
      By quasiseparability of $\vL$ and $\vR$, the matrix $\begin{bmatrix} \vzero & \vL_{12} \\ \vL_{21} & \vzero \end{bmatrix}$ has rank at most $2t$. Given a rank $2t$ factorization of it, a rank $2t$ factorization of $\begin{bmatrix} \vzero & \vL_{12} \\ \vL_{21} & \vzero \end{bmatrix}\vA$ can be found with $2t$ matrix-vector multiplications by $\vA$, requiring $2t M_t(N,r)$ operations. Thus a rank $O(t+r)$ factorization of the RHS can be found in $O(t M_t(N,r))$ operations (by stacking the generators of the individual low-rank components of the sum).

      Note that taking sub-blocks of the above equation yield the desired generators described in Step~\ref{step:dr-inverse:1}.

    \item[Step~\ref{step:dr-inverse:3}]
      By definition, $\vS = \vA_{22}-\vA_{21}\vA_{11}^{-1}\vA_{12}$. By~\eqref{eq:dr-multiplication}, its displacement rank with respect to $\vL_{22},\vR_{22}$ is at most $4(r+2t) = O(r+t)$. Furthermore, calculating the generators via~\eqref{eq:dr-multiplication} takes $O( (r+t) M_{t}(N, t+r))$ operations.

    \item[Step~\ref{step:dr-inverse:5}, \ref{step:dr-inverse:6}, \ref{step:dr-inverse:7}]
      These are performed just as in Step~\ref{step:dr-inverse:3}. Note that all of $\vA_{ij}$ and $S$ have displacement rank $O(t+r)$ with respect to $\vL_{ii},\vR_{jj}$ for some $i,j$. Furthermore their displacement types are compatible in the expressions for $\vB_{11},\vB_{12},\vB_{21}$ so that generators for all of them can be computed by~\eqref{eq:dr-multiplication}.

    \item[Step~\ref{step:dr-inverse:8}]
      By Lemma~\ref{lmm:dr-compress}, this step requires $O( (t+r)^2N + (t+r)\cM(N))$ operations. This is subsumed by the other steps.

    \item[Step~\ref{step:dr-inverse:2}, Step~\ref{step:dr-inverse:4}]
      Finally, we analyze the recursive cost. The above shows that outside the recursive calls, the algorithm requires $O( (r+t) M_t(N, r+t))$ operations. Note that the displacement ranks of $\vA_{11}$ and $\vS$ are $O(r+t)$, going up by $t$ from the displacement rank of $\vA$. Therefore throughout the entire algorithm, the displacement ranks are bounded by $O(r+t\log N)$. Thus any recursive call on a submatrix of size $n$ uses $O((r+t\log{N}) M_t(n, r+t\log{N}))$ operations. By standard recursive analysis, the total cost is a $\log{N}$ factor on top of this, for a total complexity of $O((r+t\log{N}) M_t(N, r+t\log{N})\log{N})$.
  \end{description}
\end{proof}
This completes the inverse part of Theorem~\ref{thm:intro-dr-quasi}.

As a final remark, other operations such as computing determinants, triangular factorizations, or vectors in the null space (in the non-square case) have also been considered for matrices of low displacement rank~\cite{kaltofen1994,pan2000nearly}. The algorithms for these operations are all similar to the inverse algorithm, and can also be generalized to quasiseparable $\vL,\vR$.

\section{Other Properties and Applications of Recurrence Width}
\label{sec:other}

\subsection{Recovering Low Recurrence-Width Matrices}
\label{subsec:fitting}

As recurrence width is a measure of parameterized complexity, it is useful to be able to identify instances of low-width matrices and extract information about their structure.
One natural question to ask is how to recover the recurrence width parameterization of a matrix (i.e.\ the recurrence coefficients) from the normal parameterization (i.e.\ the $N^2$ entries).
As opposed to recovering the sparse product width (or an approximation) of an arbitrary matrix which is hard (see Appendix~\ref{subsec:spw:hard}), this problem is feasible for matrices of low recurrence width.

Consider a matrix satisfying the basic recurrence in Definition~\ref{defn:intro-width}, which we rewrite here for convenience:
\[
  a_{i}(X)=\sum_{j=1}^t g_{i,j}(X)a_{i-j}(X),
\]
where $\deg(g_{i,j}) \leq j$.

Knowing the actual width $t$, it is possible to solve for the recurrence coefficients $g_{i,j}(X)$ directly. Suppose that the $a_i(X)$ are known and the $g_{i,j}(X)$ are unknown. Then each coefficient of $a_{i}(X)$ is a linear combination of the coefficients of the $g_{i,j}(X)$. More concretely, let $f[\ell]$ be the coefficient of $X^\ell$ in polynomial $f(X)$. Thus equation~\eqref{eq:intro-recur-poly} is equivalent to
\begin{equation}
  \label{eq:fitting-linear}
  a_i[k] = \sum_{j=1}^t (g_{i,j}a_{i-j})[k] = \sum_{j=1}^t \sum_{\ell=0}^{\min(j,k)} g_{i,j}[\ell] a_{i-j}[k-\ell]
\end{equation}
which is a linear constraint on the unknowns $g_{i,j}[\ell]$ (ranging over $j$ and $\ell$). There are a total of $O(t^2)$ such unknowns, and there are $N$ constraints, so this can be interpreted as a system of linear equations in $O(t^2)$ unknowns (note that this requires $N > O(t^2)$; if not, the multiplication algorithm~\ref{algo:transpose-mult} is slower than the naive method anyways). Thus for any $i$ the recurrence parameters $g_{i,\cdot}$ can be recovered in $O( (t^2)^\omega)$, and recovering the entire parameterization takes $O(t^{2\omega} N)$ operations.

Similarly, a type of \emph{projection} onto the set of width-$t$ matrices can be performed. When $\vA$ does not have recurrence width at most $t$, the system~\eqref{eq:fitting-linear} (by varying $k$) forms an overdetermined linear system on the coefficients of the $g_{i,j}(X)$. Ranging the size of the system from $t^2$ to $N$ provides a tradeoff between operation complexity and closeness between the projected width $t$ matrix and $\vA$.

In fact, consider the general recurrence equation~\eqref{eq:recur}. Note that for fixed $\va_i$ and $\vR$ (i.e.\ $\vR=\vS$ for polynomial recurrences), this equation is linear in the $\vf_i$ and the coefficients of $g_{i,j}$. Thus we can consider the problem of minimizing any convex norm of the error matrix $\vF$, which is a convex program and can be solved efficiently. Although the resulting solution $\vG,\vF$ is not strictly low recurrence width, since $\vF$ is not low rank, by dropping $\vF$ or approximating it with a low rank matrix, we are left with a $(\vR,\vG,\vF)$-recurrence that induces an approximation to $\vA$ of low recurrence width.

\subsection{Hierarchy of Recurrence}
\label{sec:hierarchy}

In this section we show that a $(t+1)$-term recurrence cannot be recovered by a $t$-term recurrence, showing a clear hierarchy among the matrices that satisfy our recurrence. We note that just by a counting argument one can show that there exist a $(t+1)$-term recurrence cannot be recovered by a $t$-term recurrence. However, next we show that the best ``error term" has rank $\Omega(N)$ (at least for constant $t$).

Fix an arbitrary $N$. We will be looking at the simplest form of our recurrence: a polynomial family $a_0, \dots, a_{N-1}$ such that $$a_{i}(X) = \sum_{j=1}^t g_{i,j}(X) a_{i-j}(X)$$ where $\deg(g_{i,j}) \le j$. Define $P(t)$ to be all families of $N$ polynomials that satisfy our recurrence of size $t$. For simplicity, we assume that all polynomials have integer coefficients. For any polynomial family $f \in P(t)$, we define the matrix $\vM(f)$ that contains the coefficients of the polynomials as its elements. To show the hierarchy of matrices, we will show that no family in $P(t)$ can approximate the mapping specified by a particular family in $P(t+1)$. 

\begin{thm}
For every $t\ge 1$, there exists an $f \in P(t+1)$ such that for every $g \in P(t)$
\[\left\|\vM(f)\cdot\vone - \vM(g)\cdot\vone\right\|^2_2 \ge \frac{N}{2(t+1)}.\]
\end{thm}
\begin{proof}
We are going to choose $f$ such that $a_i(X)= X^i$ if $i=k(t+1)$ for some $k \ge 0$ and $a_i(X) = 0$ otherwise. Note that $f\in P(t+1)$. Let $\vc = \vM(f) \cdot \vone$. In particular,
 \[\vc[i] =
\begin{cases}
 1& \text{ if } i=k(t+1) \text{ for some }k \ge 0\\
0&\text{ otherwise}
\end{cases}.
\]

Fix an arbitrary $g \in P(t)$. Let $\vc' = \vM(g) \cdot \vone$; note that $\vc'[i] = g_i(1)$. Note that if any $t$ consecutive $g_i(1), \dots, g_{i+t-1}(1)$ are $0$, the $t$-term recurrence implies that all subsequent polynomials in family uniformly evaluate to $0$ at $1$. So we have two cases: $(1)$ $g_{i}(1) = 0$ for all $i > N/2$ or $(2)$ for each $k$, there exists an $i$ such that $k(t+1) < i < (k+1)(t+1) \le N/2$ and $g_i(1) \neq 0$. In the first case $\vc'[i] = 0$ for all $i > N/2$, and $\| \vc - \vc'\|^2_2 \ge \frac{N}{2(t+1)}$. Similarly, in the second case $\vc'[i] \neq 0$ and $\vc[i] = 0$ for each of the specified $i$, implying once again that $\| \vc - \vc'\|^2_2 \ge \frac{N}{2(t+1)}$.
\end{proof}

\bcor
For every $t \ge 1$, there exists an $f \in P(t+1)$ such that for every $g \in P(t)$, $\rank(\vM(f) - \vM(g)) \ge \frac{N}{2t}$.
\ecor
%\ar{I think the rank bound is also tight. If so, we should state the that argument for tightness. In short, here is the argument I have in mind. Pick $f\in P(t+1)$ and then move every $(t+1)$th polynomials into the error matrix. We now have a matrix with $N-N/(t+1)$ rows-- make the last $N/(t+1)$ polys to be $0$ and the corresponding error rows to be the polynomials that were removed in first step. Note that these polynomials might also have to be repeated in the other rows of the error matrix. So we now have a recurrence with \width\ $t$ and error term with rank $N/(t+1)$, which gives rise to a matrix that is a row permutation of the original matrix. Let me know if I'm missing something here.}
\rpu{I don't think this (or an easy modification of this) can work. In particular, a $t$ recurrence means you can express a polynomial $a_{i+1}$ in terms of $a_{i-t}, \dots a_i$. If you remove one of the previous polynomials, I don't think you can guarantee that $a_{i+1}$ can be written in terms of the other $t-1$ polynomials.}
\begin{proof}
Let $\vH = \vM(f) - \vM(g)$. Recall that this is a lower triangular matrix. Once again, we define $f$ such that $a_i(X) = X^i$ if $i = k(t+1)$ for some $k \ge 0$ and $a_i(X) = 0$ otherwise. Once again, we have two cases. First, $\deg(g_i(X)) < i$ for all $i = k(t+1), i > N/2$. Then $\vH[i,i] = 1$ for all $i = k(t+1), i>N/2$, implying $\rank(\vH) \ge \frac{N}{2t}$. In the second case, we rely on the degree bound of the transition polynomials; in particular, if $g_{i+1}(X) = \sum_{j=0}^t h_{i,j}(X) g_{i-j}(X)$, we bound $\deg(h_{i,j}(X)) \le j$. If $\deg(g_i(X)) = i$ for some $i = k(t+1), i> N/2$, we know that for each value of $k$ such that $(k+1)(t+1)<N/2$, there exists a $j$ such that $k(t+1) < j < (k+1)(t+1)$ and $\deg(g_j(X)) = j$. For each of these $j$, $\vH[j,j] \neq 0$, implying $\rank(\vH) \ge \frac{N}{2t}$.
\end{proof}

\subsection{Succinct Representations and Multivariate Polynomials}
\label{sec:multi}

The goal of this section is two fold. The first goal is to present matrices that have low \width\ in our sense but were not captured by previous notions of widths of structured matrices. The second goal is to show that substantially improving upon the efficiency of our algorithms with respect to sharper notions of input size will lead to improvements in the state-of-the-art algorithms for multipoint evaluation of multivariate polynomials. Our initial interest in these matrices arose from their connections to coding theory, which we will also highlight as we deal with the corresponding matrices.

%\subsection{Multipoint evaluation of multivariate polynomials}

%\ar{\textbf{TODO for me} Need to figure out exactly what form of rational stuff is needed in this section. All recurrences in this section are $\vS$-matrix recurrences so there might be a way to express all the recurrences as rational functions themselves.}

We consider the following problem.
\begin{defn}
Given an $m$-variate polynomial $f(X_1,\dots,X_m)$ such that each variable has degree at most $d-1$ and $N=d^m$ distinct points $\vx(i)=(x(i)_1,\dots,x(i)_m)$ for $1\le i\le N$, output the vector $(f(\vx(i)))_{i=1}^N$.
\end{defn}

The best runtime for an algorithm that solves the above problem (over an arbitrary field) takes time $O(d^{\omega_2(m-1)/2+1})$, where an $n\times n$ matrix can be multiplied with an $n\times n^2$ matrix with $O(n^{\omega_2})$ operations~\cite{bivariate,KU11}. (For the sake of completeness, we state these algorithms in Appendix~\ref{sec:app-multi}.) We remark on three points. First in the multipoint evaluation problem we do not assume any structure on the $N$ points: e.g.\ if the points form an $m$-dimensional grid, then the problem can be solved in $\tO(N)$ many operations using standard FFT techniques. Second, if we are fine with solving the problem over {\em finite} fields, then the breakthrough result of Kedlaya and Umans~\cite{KU11} solves this problem with $N^{1+o(1)}$ operations (but for arbitrary $N$ evaluation points). In other words, the problem is not fully solved only if we do not have any structure in the evaluation points and we want our algorithms to work over arbitrary fields (or even $\R$ or $\C$). Finally, from a coding theory perspective, this problem (over finite fields) corresponds to encoding of arbitrary puncturings of Reed-Muller codes.
%\ar{I have confirmation of the above for $m=2$. Am waiting for confirmation for the general case from Chris U.}

While we do not prove any new upper bounds for these problems, it turns out that the connection to multipoint evaluation of multivariate polynomials has some interesting complexity implications for our result. In particular, recall that the worst-case input size of a matrix with \width\ $t$ is $\Theta(t^2N)$ and our algorithms are optimal with respect to this input measure (assuming $\omega=2$). However, it is natural to wonder if one can have faster algorithms with respect to a more per-instance input size.

Next, we aim to show that if we can improve our algorithms in certain settings then it would imply a fast multipoint evaluation of multivariate polynomials. In particular, we consider the following two more succinct ways of representing the input. For a given polynomial $f(X)\in\F[X]$, let $\|f\|_0$ denote the size of the support of $f$. Finally, consider a matrix $\vA$ defined by a recurrence in~\eqref{eq:recur}. Define
\[\|\vA\|_0=\sum_{i=0}^{N-1}\sum_{j=0}^t \|g_{i,j}\|_0 +rN,\]
i.e.\ the size of sum of the sizes of supports of $g_{i,j}$'s plus the size of the rank $r$-representation of the error matrix in~\eqref{eq:recur}.

The second more succinct representation where we have an extra bound that $\|g_{i,j}\|\le D$ (for potentially $D<t$) for the recurrence in~\eqref{eq:recur}. Then note that the corresponding matrix $\vA$ can be represented with size $\Theta(tDN+rN)$ elements. In this case, we will explore if one can improve upon the dependence on $r$ in Theorem~\ref{thm:recur-general}.

%\ar{As per Chris' suggestion, add the para below on impossibility results in terms of sparsity of $\vE$.}
We would like to point out that in all of the above the way we argue that the error matrix $\vE$ has rank at most $r$ is by showing it has at most $r$ non-zero columns. Thus, for our reductions $rN$ is also an upper bound on $\|\vE\|_0$, so there is no hope of getting improved results in terms of the sparsity of the error matrix instead of its rank without improving upon the state-of-the-art results in multipoint evaluation of multivariate polynomials.

%\ar{I still need to make a pass on the stuff below.}

\subsubsection{Multipoint evaluation of bivariate polynomials}
\label{sec:bivar}

We begin with the bivariate case (i.e.\ $m=2$) since that is enough to connect improvements over our results to improving the state-of-the-art results in multipoint evaluation of bivariate polynomials.

For notational simplicity we assume that the polynomial is $f(X,Y)=\sum_{i=0}^{d-1}\sum_{j=0}^{d-1} f_{i,j} X^iY^j$ and the evaluation points are $(x_1,y_1),\dots,(x_N,y_N)$. Now consider the $N\times N$ matrix
\[\vA^{(2)}=
\begin{pmatrix}
1 & x_1 & \cdots &x_1^{d-1} & y_1 & y_1x_1&\cdots & y_1 x_1^{d-1} &  y_1^2 & y_1^2x_1 &\cdots & y_1^2 x_1^{d-1} & \cdots &  y_1^{d-1} & y_1^{d-1}x_1&\cdots& y_1^{d-1} x_1^{d-1}\\
1 & x_2 & \cdots &x_2^{d-1} & y_2 & y_2x_2&\cdots & y_2 x_2^{d-1} &  y_2^2 & y_2^2x_2 &\cdots & y_2^2 x_2^{d-1} & \cdots &  y_2^{d-1} & y_2^{d-1}x_2&\cdots& y_2^{d-1} x_2^{d-1}\\
 &  & & & & & \vdots &  &   &  & &  & \vdots &   & & & \\
1 & x_N & \cdots &x_N^{d-1} & y_N & y_Nx_N&\cdots & y_N x_N^{d-1} &  y_N^2 & y_N^2x_N &\cdots & y_N^2 x_N^{d-1} & \cdots &  y_N^{d-1} & y_N^{d-1}x_N&\cdots& y_N^{d-1}x_N^{d-1}
\end{pmatrix}.
\]

Note that to solve the multipoint evaluation problem we just need to solve $\vA^{(2)}\cdot \vf$, where $\vf$ contains the coefficients of $f(X,Y)$. Let $\vD_X$ and $\vD_Y$ denote the diagonal matrices with $\vx=(x_1,\dots,x_N)$ and $\vy=(y_1,\dots,y_N)$ on their diagonals respectively. Finally, define $\vZ=\vS^T$. Now consider the matrix
\[\vB^{(2)} = \vD_X^{-1}\vA^{(2)}-\vA^{(2)}\vZ.\]
It can be checked that $\vB^{(2)}$ has rank at most $d$. Indeed note that
\[\vB^{(2)}=
\begin{pmatrix}
\frac{1}{x_1} & 0 & \cdots &0 & \frac{y_1}{x_1}-x_1^{d-1} & 0&\cdots & 0 &  y_1\left(\frac{y_1}{x_1}-x_1^{d-1}\right) & 0 &\cdots & 0 & \cdots &  y_1^{d-2}\left(\frac{y_1}{x_1} -x_1^{d-1}\right) & 0&\cdots& 0\\
\frac{1}{x_2} & 0 & \cdots &0 & \frac{y_2}{x_2}-x_2^{d-1} & 0&\cdots & 0 &  y_2\left(\frac{y_2}{x_2}-x_2^{d-1}\right) & 0 &\cdots & 0 & \cdots &  y_2^{d-2}\left(\frac{y_2}{x_2} -x_2^{d-1}\right) & 0&\cdots& 0\\
 &  & & & & & \vdots &  &   &  & &  & \vdots &   & && & \\
\frac{1}{x_N} & 0 & \cdots &0 & \frac{y_N}{x_N}-x_N^{d-1} & 0&\cdots & 0 &  y_N\left(\frac{y_N}{x_N}-x_N^{d-1}\right) & 0 &\cdots & 0 & \cdots &  y_N^{d-2}\left(\frac{y_N}{x_N} -x_N^{d-1}\right) & 0&\cdots& 0\\
\end{pmatrix}.
\]

The above was already noticed in~\cite{OS99}. The above is not quite enough to argue what we want so we make the following stronger observation. Consider
\begin{equation}
\label{eq:bivariate-recur}
\vC^{(2)}=\vD_Y^{-1} \vB^{(2)}-\vB^{(2)}\vZ^d = \vD_Y^{-1}\vD_X^{-1}\vA^{(2)} -\vD_Y^{-1}\vA^{(2)}\vZ -\vD_X^{-1}\vA^{(2)}\vZ^d +\vA^{(2)}\vZ^{d+1}.
\end{equation}
One can re-write the above recurrence as follows (where $\va_i=\left(\vA^{(2)}[i,:]\right)^T$ and recall $\vZ=\vS^T$) for any $0\le i<N$:
\[\left(\frac{1}{x_iy_i}-\frac{\vS}{y_i}-\frac{\vS^d}{x_i}+\vS^{d+1}\right)\cdot \va_i= \left(\vC^{(2)}[i,:]\right)^T.\]
We now claim that the rank of $\vC^{(2)}$ is at most two. Indeed, note that
\[\vC^{(2)}=
\begin{pmatrix}
\frac{1}{x_1y_1} & 0 & \cdots &0 & \frac{-x_1^{d-1}}{y_1} & 0&\cdots & 0 &  0 & 0 &\cdots & 0 & \cdots &  0& 0&\cdots& 0\\
\frac{1}{x_2y_2} & 0 & \cdots &0 & \frac{-x_2^{d-1}}{y_2} & 0&\cdots & 0 &  0 & 0 &\cdots & 0 & \cdots &  0& 0&\cdots& 0\\
 &  & & & & & \vdots &  &   &  & &  & \vdots &   & && & \\
\frac{1}{x_Ny_N} & 0 & \cdots &0 & \frac{-x_N^{d-1}}{y_N} & 0&\cdots & 0 &  0 & 0 &\cdots & 0 & \cdots &  0& 0&\cdots& 0\\
\end{pmatrix}.
\]

Thus, we have a recurrence with \width\ $(1,2)$ and degree $(D,1)$.  Theorem~\ref{thm:recur-general} implies that we can solve the above problem with $\tO(d^3)$ operations. The algorithm of~\cite{bivariate} uses $\tO(d^{\omega_2/2+1})$ many operations. However, note that 
\[\|\vA^{(2)}\|_0 =\Theta(d^2).\]
Thus, we have the following result:
\begin{thm}
\label{thm:lb-1}
If one can solve $\vA\vb$ for any $\vb$ with $\tO\left(\left(\|\vA\|_0\right)^{\omega_2/4+1/2-\eps}\right)$ operations, then one will have an multipoint evaluation of bivariate polynomials with $\tO(d^{\omega_2/2+1-2\eps})$ operations, which would improve upon the currently best-known algorithm for the latter.
\end{thm}

%A natural question given~\eqref{eq:bivariate-recur} is whether we can hope to run in time linear in a representation where use the sparse polynomial representation for polynomial coefficients. Note that in such a representation, the input size is $O(N)$. However getting near-linear runtime will beat the best-known multipoint evaluation of bivariate polynomials.

%\ar{Note that~\eqref{eq:bivariate-recur} has $t=1$. If we wanted constant $t>1$ then note that we can artificially inflate the number of terms in the recurrence.}

\subsubsection{Multipoint evaluation of multivariate polynomials}

%\ar{Here the idea is to start off with the ``evaluation" matrix $\vA^{(m)}$ and then eliminate $m/2$ variables by doing successive left multiply (with corresponding diagonal matrix) and right multiply (with an appropriate power of $\vZ$). This will lead to the degree of the recurrence being about $d^{m/2}$. Further, the hope is to show that the error matrix has about $d^{m/2}$ non-zero columns. Hence the claim is we should have a $(1,d^{m/2},d^{m/2})$ matrix-dependent holonomic recurrence. Note that here the usual input representation has size about $d^{3m/2}$, which is asymptotically smaller than the known multipoint evaluation algorithm and so e.g.\ we cannot have an additive $O(rN)$ term in the runtime.}

%\ar{TODO: Implement the above idea. I am reasonably confident that the above works (though that does not mean much). I'm stll trying to figure out the correct notation to handle the general $m$ case.}

We now consider the general multivariate polynomial case. Note that we can represent the multipoint evaluation of the $m$-variate polynomial $f(X_1,\dots,X_m)$ as $\vA^{(m)}\vf$, where $\vf$ is the vector of coefficients and $\vA^{(m)}$ is presented as follows.

Each of the $d^m$ columns are indexed by tuples $\vi\in\Z_d^m$ and the columns are sorted in lexicographic increasing order of the indices. The column $\vi=(i_1,\dots,i_m)\in\Z_d^m$ is represented by
\[\vA^{(m)}[:,\vi]=
\begin{pmatrix}
\prod_{j=1}^m x(1)_j^{i_j}\\
\prod_{j=1}^m x(2)_j^{i_j}\\
\vdots\\
\prod_{j=1}^m x(N)_j^{i_j}
\end{pmatrix},
\]
where the evaluation points are given by $\vx(1),\dots,\vx(N)$.

For notational simplicity, we will assume that $m$ is even. (The arguments below can be easily modified for odd $m$.) Define recursively for $0\le j\le m/2$:
\begin{equation}
\label{eq:multivariate-recur}
\vB^{(j)} = \vD_{X_{m-j}}^{-1}\vB^{(j+1)} -\vB^{(j+1)}\vZ^{d^{j}},
\end{equation}
where $\vD_{X_{k}}$ is the diagonal matrix with $(x(1)_{k},\dots,x(N)_{k})$ on its diagonal. Finally, for the base case we have
\[\vB^{(\frac{m}{2}+1)}=\vA^{(m)}.\]

It can be verified (e.g.\ by induction) that the recurrence in~\eqref{eq:multivariate-recur} can be expanded out to
\begin{equation}
\label{eq:multivariate-recur-equiv}
\vB^{(0)}= \sum_{S\subseteq [m/2,m]} (-1)^{m/2+1-|S|}\left(\prod_{j\in S} \vD_{X_j}^{-1}\right)\vA^{(m)}\left(\prod_{j\in [m/2,m]\setminus S} \vZ^{d^{m-j}}\right).
\end{equation}
The above can be re-written as (where $\vf_i=\left(\vA^{(m)}[i,:]\right)^T$):
\[\left(\sum_{S\subseteq [m/2,m]} (-1)^{m/2+1-|S|}\cdot \frac{1}{\prod_{j\in S} x(i)_j} \left(\prod_{j\in [m/2,m]\setminus S} \vS^{d^{m-j}}\right)\right)\cdot \vf_i=\left(\vB^{(0)}[i,:]\right)^T.\]

We will argue in Appendix~\ref{sec:app-multi} that
\begin{lmm}
\label{lmm:disp-rank-multivariate}
$\vB^{(0)}$ has rank at most $2^{m/2}\cdot d^{m/2-1}$.
\end{lmm}
%\ar{The final bound might change a bit.}

Note that the above lemma implies that the recurrence in~\eqref{eq:multivariate-recur-equiv} is a $\vS$-matrix recurrence with \width\  $(1,r=2^{m/2}d^{m/2-1})$ and degree
$(D=\frac{d^{1+m/2}-1}{d-1},1)$. 
Note that in this case we have $tDN+rN=\Theta((2d)^{3m/2-1})$. Thus, we have the following result:
\begin{thm}
\label{thm:lb-2}
 If for an $\vS$-dependent recurrence we could improve the algorithm from Theorem~\ref{thm:recur-general} to run with $\tO(\poly(t)\cdot DN+rN)$ operations for matrix vector multiplication, then we would be able to solve the general multipoint evaluation of multivariate polynomials in time $\tO((2d)^{3m/2-1})$, which would be a polynomial improvement over the current best algorithm (when $d=\omega(1)$), where currently we still have $\omega_2>3$.
\end{thm}

Note that the above shows that improving the dependence in $r$ in Theorem~\ref{thm:recur-general} significantly (even to the extent of having some dependence on $Dr$) will improve upon the current best-known algorithms (unless $\omega_2=3$).

%The proof of Lemma~\ref{lmm:disp-rank-multivariate} is deferred to Appendix~\ref{sec:app-multi}.

Finally, in Appendix~\ref{subsec:multi-derivative}, we present one more example of matrices that have been studied in coding theory that satisfy our general notion of recurrence. These matrices encode multipoint evaluation of multivariate polynomials and their derivatives and correspond to (puncturing of) multivariate multiplicity codes, which have been studied recently~\cite{derivative-code-1,derivative-code-2,derivative-code-3}. However, currently this does not yield any conditional ``lower bounds" along the lines of Theorem~\ref{thm:lb-1} or~\ref{thm:lb-2}.

%\newpage
%\section{Followup work}

%\subsection{Fitting}

%\input{fitting}

%\subsection{Closure Properties}
%\input{closure}

\section*{Acknowledgments}

We thank the anonymous STOC 17, FOCS 17, and SODA 18 reviewers of  earlier drafts of this paper for their helpful comments including pointing us to the transposition principle and other existing literature.
We would like to thank Swastik Kopparty, Chris Umans, and Mary Wootters for helpful discussions.
We gratefully acknowledge the support of the Defense
Advanced Research Projects Agency (DARPA) SIMPLEX program under No. N66001-15-C-4043,
the DARPA D3M program under No. FA8750-17-2-0095,
the National Science Foundation (NSF) CAREER Award under No. IIS-1353606,
the Office of Naval Research (ONR) under awards
No. N000141210041 and No. N000141310129,
a Sloan Research Fellowship, the Moore Foundation, an Okawa Research
Grant, Toshiba, and Intel.
Any opinions, findings, and conclusions
or recommendations expressed in this material are
those of the authors and do not necessarily reflect the views
of DARPA, NSF, ONR, or the U.S. government.
 Atri Rudra's research is supported by NSF grant CCF-1319402 and CCF-1717134.

\bibliographystyle{acm}
\bibliography{matrix-vect}

\appendix

\agu{Note: everything in appendix should have a line in main body referencing it}
\ar{YES!}

\section{Related Work and Known Results}
\label{sec:related}

\agu{related work to end of main body, known results keep in appendix (although do we actually need to keep any of that?)}

\ar{Might not hurt to keep in since this has stuff about quasiseparable. But there should be a forward pointer to this from the main body.}

%Note that because of the special structure of special cases such as Toeplitz-like matrices, faster algorithms than given by Theorem~\ref{thm:intro-dr} (by a $\log{N}$ factor) can be devised for them. We omit such possibilities in our algorithms because our focus is on unification.

Superfast structured matrix-vector multiplication has been a rich area of research. Popular classes of structured matrices such as Toeplitz, Hankel, and Vandermonde matrices and their inverses all have classical superfast multiplication algorithms that correspond to operations on polynomials~\cite{AHU74, BP94}.
%A superfast algorithm was developed for Cauchy matrices (and slight generalizations) was developed by Gerasoulis in 1988~\cite{G88}. Multiplication with Cauchy matrices corresponds to operations on rational polynomials; Cauchy matrices naturally fit in with the other three types of matrices.
The four classes of matrices are all subsumed by the notion of displacement rank introduced by Kailath, Kung, and Morf in 1979~\cite{KKM79} (also see~\cite{FMKL79} and~\cite{HR84}). Kailath et al. initially used displacement rank to define Toeplitz-like matrices, generalizing Toeplitz matrices. Then Morf~\cite{morf1980} and Bitmead and Anderson~\cite{bitmead1980} developed a fast divide-and-conquer approach for solving a Toeplitz-like linear system of equations in quasilinear time, now known as the MBA algorithm. This was extended to Hankel and Vandermonde-like matrices in~\cite{P90}. In 1994, Gohberg and Olshevsky further used displacement rank to define Vandermonde-like, Hankel-like, and Cauchy-like matrices and developed superfast algorithms for vector multiplication~\cite{GO94}. %\ar{I think there was also a parallel paper by Pan, which we should also cite.} \agu{Pan has so many publications about related things that I have no idea what to cite\ldots does his book (already in the bib) not work?}
In 1998, Olshevsky and Pan extended the MBA algorithm to superfast inversion, in a unified way, of these four main classes of displacement structured matrices~\cite{OP98}.
These matrix classes were unified and generalized by Olshevsky and Shokrollahi in 2000~\cite{OS00} with a class of matrices they named confluent Cauchy-like, which have low displacement rank with respect to Jordan form matrices. In this work, we extend these results by investigating matrices with low displacement rank with respect to any triangular $t$-band matrices, which we define to be matrices whose non-zero elements all appear in $t$ consecutive diagonals. General and unified algorithms for the most popular classes of matrices with displacement structure have continued to be refined for practicality and precision, such as in~\cite{HLS17} and~\cite{PT17}.
%; we refer to Pan and Tsigaridas in 2014 for an overview and for an algorithm with explicit bounds on precision~\cite{PT14}.

A second strand of research that inspired our work is the study of orthogonal polynomial transforms, especially that of Driscoll, Healy, and Rockmore~\cite{driscoll}. Orthogonal polynomials are widely used and well worth studying in their own right: for an introduction to the area, see the classic book of Chihara~\cite{chihara}. We present applications for some specific orthogonal polynomials. Chebyshev polynomials are used for numerical stability (see e.g.\ the ChebFun package~\cite{chebfun}) as well as approximation theory (see e.g.\ Chebyshev approximation~\cite{cheb-approx}). Jacobi polynomials form solutions of certain differential equations~\cite{jacobi}. Zernike polynomials have applications in optics and physics~\cite{zernike}. In fact, our investigation into structured matrix-vector multiplication problems started with some applied work on Zernike polynomials, and our results applied to fast Zernike transforms have been used in improved cancer imaging~\cite{yu2016}.
%\ar{We should probably make more noise on applications here.}
Driscoll et al. rely heavily on the three-term recurrence satisfied by orthogonal polynomials to devise a divide-and-conquer algorithm for computing matrix-vector multiplication. One main result of this work is a direct generalization of the recurrence, and we rely heavily on the recurrence to formulate our own divide and conquer algorithm. The basic algorithm in Section~\ref{sec:transpose} is modeled after previous works on orthogonal polynomials~\cite{driscoll,potts1998fast,bostan2010op}, which are themselves reminiscent of classic recursive doubling techniques such as that proposed by Kogge and Stone~\cite{KS73}.

A third significant strand of research is the study of rank-structured matrices. The most well-known example is the class of semiseparable matrices, for which we refer to an extensive survey by Vandebril, Van Barel, Golub, and Mastronardi~\cite{semisep} for a detailed overview of the body of work. Many variants and generalizations exist including the generator representable semiseparable matrices, sequentially separable mentions, and quasiseparable matrices, which all share the defining feature of certain sub-matrices (such as those contained above or below the diagonal) being low rank~\cite{VVM08,EGH13}. These types of matrices turn out to be highly useful for devising fast practical solutions of certain structured systems of equations. Quasiseparable matrices in particular, which appear in our Theorem~\ref{thm:intro-dr-quasi}, are actively being researched with recent fast representations and algorithms in~\cite{pernet2016,pernet2017}. Just as how the four main displacement structures are closely tied to polynomial operations~\cite{P01}, the work of Bella, Eidelman, Gohberg, and Olshevsky show deep connections between computations with rank-structured matrices and with polynomials~\cite{BEGO08}. Indeed at this point connections between structured matrices and polynomials is well established~\cite{P90,BP94,P01}.
%The most straightforward class of semiseparable matrices are the generator representable semiseparable matrices, which are matrices whose upper triangular and lower triangular portions are both of (low) rank.
%Our results can straightforwardly recover generator semiseparable matrices whose triangular portions are of rank $1$. However, semiseparable matrices are defined more generally with respect to the rank of matrix sub-blocks. The idea of utilizing the ranks of matrix sub-blocks has been generalized many times, and we refer to Bella, Eidelman, Gohberg, and Olshevsky's work on $(H,m)$-quasiseparable matrices~\cite{quasi} for a relatively recent exploration of various generalizations. This generalization actually has deep connections to polynomials that satisfy recurrences; if we define $p_i(X)$ as the characteristic polynomial of the upper left $i \times i$ submatrix, the $p_i$ form a family that satisfy recurrence relations. Connecting our notion of recurrence - where the rows of our matrix form a recursive polynomial family - to that of quasiseparable matrices will be explored in future work. As far as we are aware, there is no fully general superfast matrix-vector multiplication algorithm for $(H,m)$-quasiseparable matrices. We note that this is an area of great active research; we point the reader to Bella et al.'s survey on computing with quasiseparable matrices~\cite{BEGO08} for an overview of the wide array of work.

As discussed in Section~\ref{sec:intro}, we view the connection of these many strands of work as our strongest conceptual contribution.

\subsection{Known Results}
\label{subsec:known}
%\ar{Collect results on polynomials and matrix vector multiplication that we use as a given in the paper. It might have to be moved to the appendix.}

Multiplication of two degree $N$ univariate polynomials can be performed in $\tO(N)$ operations: The classic FFT computes it in $O(N\log N)$ operations for certain fields~\cite{fft}, and generalizations of the Sch\"{o}nhage-Strassen algorithm compute it in $O(N\log N\log\log N)$ ring operations in general~\cite{cantor1991fast}. Computing polynomial divisors and remainders, i.e.\ $p(X) \pmod{q(X)}$ with $\deg(p(X)), \deg(q(X)) = O(N)$ can be done with the same number of operations~\cite{AHU74}.

If matrix-matrix multiplication by two $N \times N$ matrices can be performed with $O(N^\omega)$ operations for some $\omega$, then an $N \times N$ matrix can be inverted with $O(N^\omega)$ operations~\cite{CLRS}.

The matrix-vector product by matrices $\vA,\vA^T$ (and $\vA^{-1},\vA^{-T}$ when they exist) for Toeplitz and Hankel matrices $\vA$ can be computed in $O(N\log N)$ operations~\cite{P01}.
When $\vA$ is a Vandermonde matrix, these products takes $O(N\log^2 N)$ operations~\cite{P01}. The same holds for confluent Vandermonde matrices, defined as

\begin{defn}
  \label{defn:vandermonde}
  Given a set of points $p_0, \dots, p_{N-1}$, let $n_i = \displaystyle\sum_{j=0}^{i-1} \mathbb{1}(p_j = p_i)$ be the number of preceding points identical to $p_i$. Then the $N\times M$ {\em (confluent) Vandermonde matrix}\footnote{The term Vandermonde is usually applied when the $p_i$ are all distinct, otherwise it is known as a confluent Vandermonde matrix. However, either way the matrix is uniquely specified by the ordered sequence $p_i$ so we drop this distinction when the context is clear}, denoted $\vV_{p_0, \dots, p_{N-1}}$, is defined such that \[\vV[i,j] = \begin{cases} 0, & j < n_i \\  \frac{j!}{(j-n_i)!(n_i)!} {p_i}^{j-n_i}, &j \ge n_i\end{cases}\]  for $0 \le i < N, 0 \le j < M$.\footnote{Some sources define the confluent Vandermonde matrix with the factor of $n_i!$ in the denominator~\cite{P01}, and others do not. It makes no real difference, but is slightly more convenient for us to use the former notation.}
\end{defn}
%\agu{The confluent Vandermonde is sometimes defined with the $n_i!$ in the denominator and sometimes without. Wikipedia defines it the latter way but we want the former. Might be worth a footnote}
Furthermore, multiplication by these matrices encode polynomial evaluation. For any vector $\vy$, $\vV_{p_0, \dots, p_{N-1}} \vy$ is equivalent to evaluating the polynomial $v(X) = \displaystyle\sum_{i=0}^{N-1} \vy[i] X^i$ at $v^{(n_i)}(p_i)$ for each $i$.

%Furthermore, multiplication by a Vandermonde matrix encodes polynomial evaluation: If $\vV[i,j] = \alpha_i^j$, then $(\vV\vb)[i] = f(\alpha_i)$, where $f(X) = \sum b[i] X^i$. A {\em confluent Vandermonde} matrix
%\agu{Move defn of confluent Vandermonde from 9.5 to here}

The characteristic polynomial $c_{\vM}(X)$ of a matrix $\vM$ is equal to the determinant of $X\vI-\vM$. When $\vM$ is triangular, it is equal to $\prod (X-\vM[i,i])$. By the Cayley-Hamilton Theorem, every matrix satisfies its own characteristic equation, i.e.\ $c_{\vM}(\vM) = 0$.

A $\lambda$-{\em Jordan block} is a square matrix of the form $\lambda\vI + \vS$, where $\vS$ is the shift matrix which is $1$ on the superdiagonal and $0$ elsewhere. A matrix in {\em Jordan normal form} is a direct sum of Jordan blocks. The minimal polynomial of a matrix is equal to the product $\prod_{\lambda} (X-\lambda)^{n_{\lambda}}$ where $n_{\lambda}$ is the size of the largest Jordan block for $\lambda$. Consequently, if a matrix has equal minimal and characteristic polynomials, then it has only one Jordan block per eigenvalue.

\section{Sparse Product Width}
\label{sec:spw}

In this section, we define the notion of \emph{sparse product width}. This notion and its properties have been known in various equivalent forms, most closely as the shortest linear straight-line program describing a matrix. We collect and describe these properties in the language of sparse product width, because we find this notion more intuitive in the context of structured dense matrices.

\subsection{Formal Definition}

\begin{defn}
  A \emph{factorization} of a matrix $\vA \in \F^{m \times n}$ is a sequence of matrices $\vA_1, \vA_2, \ldots, \vA_p$ (for some $p \in \N$) such that for some sequence $k_0, k_1, \ldots, k_p \in \N$, the matrices are each of shape $\vA_i \in \F^{k_i \times k_{i-1}}$, and $k_0 = n$ and $k_p = m$, and
  \[
    \vA = \vA_p \cdot \vA_{p-1} \cdots \vA_2 \cdot \vA_1 = \prod_{i=1}^p \vA_i.
  \]

  A \emph{proper factorization} of $\vA$ is a factorization of $\vA$ with the additional constraints
  \begin{itemize}
    \item None of the $\vA_i$, except possibly $\vA_p$, contains a zero row
    \item None of the $\vA_i$, except possibly $\vA_1$, contains a zero column
  \end{itemize}
\end{defn}

\begin{defn}
  \label{defn:idim}
  The \emph{inner dimension} of a pair of matrices $\vA \in \F^{m \times n}$ and $\vB \in \F^{n \times p}$ is their shared dimension
  \[ \idim(\vA,\vB) = n \]
\end{defn}

\begin{defn}
  \label{defn:spw}
  The \emph{sparse product width} of a factorization of a matrix $\vA$ is the value
  \[
    \spw(\vA_1, \ldots, \vA_p)
    =
    \sum_{i=1}^p \norm{\vA_i}_0
    -
    \sum_{i=1}^{p-1} \idim(\vA_{i+1}, \vA_i).
  \]
  The sparse product width of a matrix $\vA$ is the minimum sparse product width over all its proper factorizations, and is denoted $\spw(\vA)$.
\end{defn}

The sparse product width is equivalent to other notions usually associated with circuit complexity. One direct connection is to the Shortest Linear Program (SLP), which is to minimize the number of linear operations needed to compute a set of linear forms ~\cite{boyar2013logic}.
\begin{prop}
\label{prop:spw-equiv}
  The sparse product width of a matrix $\vA$ is equal to the output dimension (the number of non-zero rows of $\vA$) plus the length of the shortest linear program computing $\vA$. 
\end{prop}
This equivalence is quite direct. For simplicity assume that $\vA$ is well-behaved so it has no zero rows or columns. Note that the sparse product width minus output dimension is equal to the minimum over factorizations of
\[ \sum_{i=1}^p \norm{\vA_i}_0 - R(\vA_i) = \sum_{i=1}^p \sum_{r=0}^{R(\vA_i)-1} \norm{\vA_i[r,:]}_0-1 \]
where $R(\vA)$ is the number of rows of a matrix. We only need to show that the above quantity matches the length of the SLP.

Given a proper factorization $(\vA_1, \dots, \vA_p)$ of $\vA$, the factorization can be viewed as computing the linear transformation $\vx \mapsto \vA\vx$ through composing the linear transformations $\vx \mapsto \vA_i \vx, i = 1, \dots, p$. Consider any row $R$ of any matrix in the factorization, and suppose the row has $r$ non-zero entries. Multiplication by this row is the same as computing some linear function $x_\ell \leftarrow c_1 x_{i_1} + \dots + c_r x_{i_r}$, where the $c_j$ are constants depending on this factorization (i.e.\ $c_j$ are the entries of $\F$), and the $x_{i_j}$ refer to previously computed values. This can be broken into $r-1$ binary linear operations of the form $z \gets ax + by$ for any previously computed values $x,y$. Thus computation of the row $R$ can be represented by $r-1$ instructions in a linear straight-line program. Summing this over all rows in the factorization gives the value in Definition~\ref{defn:spw}.

The converse reduction is also easy to show. The $k$th linear operation of a straight line program $z \gets ax + by$ can be represented as a matrix $\vA_k \in \F^{(k+1) \times k}$ where $\vA_k[0:k, 0:k] = \vI_k$ and $\vA_k[k,:]$ is a row containing only the elements $a$ and $b$. Thus multiplying by $\vA_k$ reflects keeping all previously computed values and computing one new value $ax+by$. Thus a straight line program of length $s$ can be represented as a product of such matrices, with a sparse product width equal to $s$ plus the output dimension.

Therefore a factorization $\cA$ of $\vA$, viewed as a linear transformation, is equivalent to a linear straight-line program of length $\spw(\cA)$. This is also equivalent to the size of a linear circuit.

We also remark that an advantage of the sparse product width over notions like SLP is that it provides a natural language to describe existing classes of structured matrices. Notably, the two most canonical types, sparse matrices and low rank matrices, manifest as the simplest types of sparse factorizations - a $1$-matrix factorization and a $2$-matrix dense factorization, respectively. Furthermore, there is a large class of structured matrices including the Discrete Fourier/Sine/Cosine Transform and classical matrices of low displacement rank (Toeplitz-like, Hankel-like, Vandermonde-like, Cauchy-like, etc.) that have similar sparse factorizations - consisting of $O(\log{N})$ matrices that are roughly block diagonal - that can be grouped together under this parameterization.

\begin{rmk}
  In Definition~\ref{defn:spw}, the subtracted $\idim$ terms are necessary; as mentioned, it reflects that multiplying by a $1 \times m$ matrix requires $m-1$ linear operations. The minimum value of $\sum \norm{\vA_i}_0$ can be more than constant factor away from Definition~\ref{defn:spw}. An example is the \emph{prefix sum} matrix $\vA$ consisting of $1$s on the diagonal and below and $0$ above the diagonal. This matrix clearly requires $O(n)$ operations to multiply so $\spw(\vA) \le  O(n)$. However, suppose there was a proper factorization satisfying $\sum \norm{\vA_i}_0 = O(n)$. Note that for a proper factorization, we must have $\norm{\vA_i}_0\ge n$ for all $\vA_i$ except perhaps the first and last one. This implies that the factorization must consist of $O(1)$ matrices, each of which have $O(1)$ sparsity per row. With $O(n)$ processors, multiplication by each matrix takes constant time, so multiplication by $\vA$ can be done in $O(1)$ time. However, there is no $O(1)$ algorithm for computing prefix sums in parallel, a contradiction.\ar{(1) Need a citation for the parallel prefix sum lower bound. (2) I forget the communication model in these models. Is the lower bound for shared memory model? Basically in the reduction, you do need to communicate the intermediate result corresponding to each $\vA_i$.}
\end{rmk}

Because of Proposition~\ref{prop:spw-equiv}, the following properties of sparse product width follow from results about linear circuits.

\subsection{Properties}
\begin{lmm}
  \label{lmm:properties}
  The following properties hold for all matrices $\vA,\vB$ of the appropriate dimensions.
  \begin{description}
    \item[Product]
      \label{lmm:spw-product}
      $
        \mathsf{spw}(\vA\vB) \le \mathsf{spw}(\vA) + \mathsf{spw}(B) - \idim(\vA, \vB)
      $
    \item[Block Composition]
      \label{lmm:spw-block}
      The $\spw$ of a block matrix is the sum of the $\spw$ of the blocks.
    \item[Sum]
      \label{lmm:spw-sum}
      $
        \spw(\vA + \vB) \le \spw(\vA) + \spw(\vB).
      $
    \item[Kronecker Product]
      Let $\vA$ be an $n\times n$ matrix and $\vB$ be $m\times m$. Then $\spw(\vA \otimes \vB) \leq \spw( (\vA \otimes \vI_m)(\vI_n \otimes \vB) ) \leq m\spw(\vA) + n\spw(\vB)$.
  \end{description}
\end{lmm}

It is known that over infinite fields, linear straight line programs are optimal to within a constant factor compared to general straight-line programs (that is, those that allow multiplication of variables as well, thus not necessarily always computing linear functions of the inputs)~\cite{burgisser2013algebraic}. Correspondingly, the sparse product width is an optimal descriptor of the algorithmic complexity of matrix-vector multiplication on these types of models.
\begin{thm}
  \label{thm:spw-optimal}
  For any matrix $\vA \in \R^{m\times n}$, let $\mathcal{T}(\vA)$ denote the runtime of the fastest matrix-vector multiplication algorithm for this matrix on any arithmetic circuit. Then,
  \[
    \mathcal{T}(\vA) = \Theta(\mathsf{spw}(\vA)).
  \]
  for all $\vA \in \R^{m\times n}$.
\end{thm}

From Definition~\ref{defn:spw}, it is evident that $\spw(\vA) = \spw(\vA^T)$, which implies the classic \emph{transposition principle}. This principle has a rich history dating to Bordewijk~\cite{bordewijk1956}.
\begin{thm}
  \label{thm:transposition}
  For any matrix $\vA$, the runtimes of optimal algorithms computing $\vA$ and $\vA^T$ are within a constant factor of each other.
\end{thm}

Additionally, the sparse product width characterization nicely captures the parameterization of many types of structured matrices. Here are bounds for some important classes.
\begin{lmm}
  \begin{enumerate}[label=(\arabic*)]
    \item A matrix $\vA \in \F^{m \times n}$ that has sparsity $s$ satisfies $\spw(\vA) \leq s$.
    \item A matrix $\vA \in \F^{m \times n}$ that has rank $r$ satisfies $\spw(\vA) \leq r(m+n-1)$.
    \item A matrix $\vA \in \F^{m \times n}$ that has rigidity $\operatorname{Rig}_\vA(r) \leq s$ satisfies $\spw(\vA) \leq r(m+n-1) + s$. 
    \item A matrix $\vA \in \F^{n \times n}$ that is a Discrete Fourier Transform, Hankel, or Toeplitz matrix satisfies $\spw(\vA) \leq O(n\log n)$. A Vandermonde matrix satisfies $\spw(\vA) \leq O(n\log^2 n)$.
  \end{enumerate}
\end{lmm}
% TODO property references to labels here
\begin{proof}
  (1) and (2) follow from definitions of sparsity and rank; these sparse factorizations consist of $1$ matrix and $2$ matrices respectively. (3) follows from the subadditive property of sparse product width. The matrices of (4) have classically factorizations of length $O(\log{N})$ or $O(\log^2{N})$ where each matrix has $O(N)$ entries. For example, it is well-known that the DFT matrix can be factored into $\log{N}$ matrices with each having $2$ non-zero entries per row, corresponding to the Cooley-Tukey Fast Fourier Transform algorithm~\cite{van1992computational}; Hankel and Toeplitz matrices can be viewed as convolutions and hence a constant number of DFT applications; Vandermonde matrices can be factored into $\log{N}$ block-Toeplitz matrices~\cite{driscoll}.
\end{proof}

The basic low-recurrence width matrices~\eqref{eq:intro-recur-poly} also admit a similar sparse factorization, consisting of $O(\log{N})$ roughly block-diagonal matrices. For details in the orthogonal polynomial case, see~\cite{bostan2010op}; generalizing from recurrences of width $2$ to $t$ is easy.

\subsection{Hardness}
\label{subsec:spw:hard}
The sparse product width is a useful way to describe how complex a given linear operator is. However, given a matrix, it is not easy to actually recover its sparse product width. It is known that over any field $\F$, the question of finding the shortest linear straight-line program is NP-hard~\cite{boyar2013logic}. Furthermore, there is no polynomial time approximation scheme for it~\cite{boyar2013logic}. The best known approximation ratio is $O(N/\log{N})$~\cite{peralta} (which follows from Lupanov's upper bound on general matrix vector multiplication~\cite{lupanov}).

\section{Details on the Structure Lemma}
\label{sec:structure}

\begin{lmm}
  \label{lmm:structure-general-deg}
  For any commutative ring $\cR$, let $\vG \in \cR^{N \times N}$ be a $t$-band lower triangular matrix, in particular
  \[
    \vG[i,j] =
    \begin{cases}
      g_{i,0} & i=j \\
      -g_{i,i-j} & j < i \le j+t \\
      0 & \text{otherwise}
    \end{cases}
  \]
  Let $\vT_i: i = 0, \dots, n-2$ be the companion matrix 
  \[\vT_i=
    \begin{pmatrix}
    0 & 1& \cdots &0  &0 \\
    \vdots & \vdots & \cdots &\vdots &\vdots\\
    0 & 0 &\cdots &1 &0\\
    0 & 0 &\cdots &0 &1\\
    g'_{i+1,t} & g'_{i+1,t-1} & \cdots & g'_{i+1,2} & g'_{i+1,1} \\
    \end{pmatrix}.
  \]
  where $g'_{i,j} = g_{i,j}\prod_{k=i-j+1}^i g_{k,0}$ (we say $g_{i,j} = 0$ for $j > i$).

  Then
  \[
    \vG^{-1}[i,j] = (g_{j,0}\dots g_{i,0})^{-1} \vT_{[j:i]}[t-1,t-1]
  \]
  (recall $\vT_{[j:i]}[t-1,t-1]$ is the bottom right corner of $\vT_{[j:i]}$).
\end{lmm}
\begin{proof}
  We first show it assuming $g_{i,0} = 1$ for all $i$. In this case, $g'_{i,j} = g_{i,j}$ for all $i,j$. 

  Consider $\vx,\vy \in \cR^{N}$ (symbolically) such that $\vG\vx = \vy$. This is equivalent to the recurrence relation
  \[
    \vx[i] = \sum_{j=1}^t g_{i,j}\vx[i-j] + \vy[i]
  \]
  or
  \[
    \begin{bmatrix} \vx[i-t+1] \\ \vdots \\ \vx[i] \end{bmatrix}
    =
    \vT_{i-1}
    \begin{bmatrix} \vx[i-t] \\ \vdots \\ \vx[i-1] \end{bmatrix}
    +
    \begin{bmatrix} \vzero^T \\ \vdots \\ \ve_i^T \end{bmatrix}
    \vy
  \]
  for all $i \geq 1$. As a base case of the recurrence, we have
  \[
    \begin{bmatrix} \vx[-t+1] \\ \vdots \\ \vx[0] \end{bmatrix}
    =
    \begin{bmatrix} \vzero^T \\ \vdots \\ \ve_0^T \end{bmatrix}
    \vy
  \]
  Iterating this relation gives
  \[
    \begin{bmatrix} \vx[i-t+1] \\ \vdots \\ \vx[i] \end{bmatrix}
    =
    \vT_{i-1}\dots\vT_0
    \begin{bmatrix} \vzero^T \\ \vdots \\ \ve_0^T \end{bmatrix}
    \vy
    + \dots +
    \vT_{i-1}
    \begin{bmatrix} \vzero^T \\ \vdots \\ \ve_{i-1}^T \end{bmatrix}
    \vy +
    \begin{bmatrix} \vzero^T \\ \vdots \\ \ve_{i}^T \end{bmatrix}
    \vy
    =
    \sum_{j=0}^i
    \vT_{[j:i]}
    \begin{bmatrix} \vzero^T \\ \vdots \\ \ve_{j}^T \end{bmatrix}
    \vy
  \]
  Thus
  \[
    \vx[i] = \sum_{j=0}^i \vT_{[j:i]}[t-1,t-1] \vy[j].
  \]
  But symbolically, $\vx = \vG^{-1}\vy$, so by matching coefficients we see
  \[
    \vG^{-1}[i,j] = \vT_{[j:i]}[t-1,t-1],
  \]
  in other words the bottom right element of $\vT_{[j:i]}$.

  We now show it in general. Notice that multiplying $\vG$ on the left by $\vL = \diag(1, g_{0,0}, g_{0,0}g_{1,0}, \dots, g_{0,0}\dots g_{n-2,0})$ and on the right by $\vR = \diag( g_{0,0}^{-1}, (g_{0,0}g_{1,0})^{-1}, \dots, (g_{0,0}\dots g_{n-1,0})^{-1})$ yields the matrix
  \[
    \vG'[i,j] =
    \begin{cases}
      1 & i=j \\
      -g'_{i,i-j} & j < i \le j+t \\
      0 & \text{otherwise}
    \end{cases}
  \]
  Thus $\vG^{-1} = (\vL^{-1}\vL\vG\vR\vR^{-1})^{-1} = \vR\vG'^{-1}\vL$, where $\vG'^{-1}$ can be found by the previous case.
\end{proof}

\begin{proof}[Proof of Lemma~\ref{lmm:structure2}]
  Apply Lemma~\ref{lmm:structure-general-deg}. When the recurrence is of degree $(d,\bard)$, multiply on the left and right by appropriate diagonal matrices as in Lemma~\ref{lmm:structure-general-deg}.
\end{proof}

%
%\begin{proof}[Proof of Lemma~\ref{lmm:work-horse}]
%Note that~\eqref{eq:recur-non-comm} can be written as
%\[ \vf_{i+1} = \begin{bmatrix} g_{i,0}(X) & \cdots & g_{i,t}(X) \end{bmatrix} \otimes \begin{bmatrix} \vf_{i} & \cdots & \vf_{i-t} \end{bmatrix}^T \]
%or
%\[ \begin{bmatrix} \vf_{i+1} & \cdots & \vf_{i-t+1} \end{bmatrix}^T = \vT_i \otimes \begin{bmatrix} \vf_{i} & \cdots & \vf_{i-t} \end{bmatrix}^T \]
%By composing this, we get
%\begin{align*}
%  \begin{bmatrix} \vf_{k+i} & \cdots & \vf_{k+i-t} \end{bmatrix}^T &= \vT_{k+i-1} \otimes \left( \cdots \otimes \left( \vT_{k} \otimes \begin{bmatrix} \vf_{k+t} & \cdots & \vf_{k} \end{bmatrix}^T \right) \right) \\
%  &=\left(\vT_{k+i-1}\cdot \vT_{k+i-2}\cdots \vT_k\right) \otimes \begin{bmatrix} \vf_{k+t} & \cdots & \vf_{k} \end{bmatrix}^T\\
%  &=\vT_{[k:k+i]} \otimes \begin{bmatrix} \vf_{k+t} & \cdots & \vf_{k} \end{bmatrix}^T
%\end{align*}
%The first part of the lemma is equivalent to the top row of this equation, where we define $h_{i,j}^{(k)}$ to be $(1,j)$-th entry of $\vT_{[k:k+i]}$. The second part of the lemma is equivalent to the top row of $\vT_{[k:k+i+1]} = \vT_{k+i}\vT_{[k:k+i]}$.
%\end{proof}

%\agu{explicit lemma for degree sizes in general case?}
%\begin{lmm}
%\label{lmm:T-range-sizes2}
%Let the recurrence in~\eqref{eq:recur-non-comm} be $(d,\bar{d})$-nice. Then
%for any $0 \leq \ell \leq r < N$, the matrix $\vT_{[\ell:r]} = D_{\ell,r}(X)^{-1}\vT'_{\ell,r}$ where $\deg(D_{\ell,r}) \leq \bar{d}(r-\ell)$ and for all $0\le i,j\le t$,
%\[\deg(\vT'_{[\ell:r]}[i,j]) \le \bar{d}(r-\ell) + d\max((r-\ell+j-i),0).\]
%\end{lmm}

\begin{proof}[Proof of Lemma~\ref{lmm:T-range-sizes}]
  We prove this when $(d,\bar{d}) = (1,0)$. For general degrees, the transition matrices $\vT_i$ have degrees uniformly scaled by $(d+\bard)$, as shown in Lemma~\ref{lmm:structure-general-deg}.

  Fix any arbitrary $\ell$. We will prove the statement by induction on $r-\ell$. For the base case (i.e.\ when we are considering $\vT_{\ell}$), the bounds follow from the definition of $\vT_{\ell}$ and the our assumption on the sizes of $g_{\ell,j}(X)$ for $0\le j\le t$. In fact when $r \leq \ell+t$ it can be shown inductively that the last $r-\ell$ rows satisfy the desired degree constraints and the rest looks like a shift matrix, i.e.\ $\vT_{[\ell:r]}[i,j] = \delta_{r-\ell+i-j}$ for $i < t-(r-\ell)$. Now assume the result is true for $r-\ell=\Delta \ge t+1$.

Now consider the case of $r=\ell+\Delta+1$. %In the inductive hypothesis we can drop the $\max$ by the bound on $r-\ell$. 
 In this case note that $\vT_{[\ell:r]}=\vT_{r-1}\cdot \vT_{[\ell:r-1]}$. By the action of $\vT_{r-1}$, the first $t-1$ rows of $\vT_{[\ell:r]}$ are the last $t-1$ rows of $\vT_{[\ell:r-1]}$ and the size claims for entries in those rows follows from the inductive hypothesis. Now note that for any $0\le j < t$, we have
 \[
     \vT_{[\ell:r]}[t-1,j]=\sum_{k=1}^{t} g_{r,k}(X)\cdot \vT_{[\ell:r-1]}[t-k,j].
 \]
By the inductive hypothesis, we have that
\[
    \deg\vT_{[\ell:r]}[t-1,j] \le \max_{0\le k < t} \left( \deg g_{r,k} + \deg\vT_{[\ell:r-1]}[t-k,j] \right) \le \max_k (k + ((r-1)-\ell+(t-k)-j)) = r-\ell+(t-1)-j,
\]
as desired.
\end{proof}

\section{Explicit Algorithms for Computing $\vA\vb$}
\label{sec:Ab}

In this section we elaborate on the problem of computing
\[\vc=\vA\vb,\]
by providing an explicit algorithm. We note that an alternate way to derive such an algorithm is through the transposition principle. In particular, by analyzing the algorithm for $\vA^T\vb$ as a linear computation on $\vb$, a circuit or sparse factorization (Appendix~\ref{sec:spw}) of $\vA$ can be deduced. The transpose of this factorization yields an algorithm calculating $\vA\vb$. Such an approach was taken in~\cite{bostan2010op} in the case of orthogonal polynomial recurrences, and a generalization of their sparse factorization from $2$ to $t$ can be used to deduce a $\vA\vb$ algorithm for a restricted setting (polynomial recurrences with no modulus or error). However, here we provide an algorithm with a natural interpretation in terms of the behavior of the linear operator $\vA$. It also turns out that this algorithm is the natural transpose of Algorithm~\ref{algo:transpose-mult}, although we do not show the details here.

We start with a basic polynomial recurrence~\eqref{eq:intro-recur-poly} where the $g_{i,j}(X)$ are degree $(1,0)$. For the purposes of later generalizing to modular recurrences, we will actually consider a minor generalization where the recurrence has degree $(d,\bard)$ and the polynomials are not taken in a mod, so they define a matrix of dimensions $N \times (d+\bard)N$. Namely, consider the polynomials $f_i(X)$ defined by recurrence
\begin{equation}
\label{eq:recur-rational-equiv-new}
D_{i}(X)f_{i+1}(X) =\sum_{j=0}^t n_{i,j}(X)\cdot f_{i-j}(X).
\end{equation}
where $\deg(D_i(X)) = \bar{d}$ and $\deg(n_{i,j}(X)) \leq d(j+1) + \bar{d}$. Furthermore, assume the starting conditions $f_i(X) : 0 \leq i \leq t$ satisfy
\begin{equation}
  \label{eq:rational-assumption}
  \prod_{j=0}^{N-1} D_j(X) \mid f_i(X)
\end{equation}
(this condition is only to ensure that the recurrence generates polynomials, which will be shown).

Assume that the $f_i(X)$ have degrees bounded by $\bar{N} = (d+\bar{d})N$, and we will consider the problem of computing $\vA\vb$ where the $\vA[i,j]$ is the coefficient of $X^j$ in $f_i(X)$ (note that $\vA \in \F^{N \times \bar{N}}, \vb \in \F^{\bar{N}}$).

We again emphasize that the setting when $(d,\bar{d}) = (1,0)$ corresponds to the basic polynomial recurrence and should be considered the prototypical example for this section. In this case $\bar{N} = N$, $D_i(X) = 1$ for all $i$, and the divisibility assumptions on the starting polynomials~\eqref{eq:rational-assumption} are degenerate.

The main idea of the algorithm is to use the following observation to compute $\vA\vb$. 
For any vector $\vu\in\F^N$, define the polynomial
\[\vu(X)=\sum_{i=0}^{N-1} u_i\cdot X^i.\]
Also define
\[\vu^R=\vJ\cdot \vu,\]
for the {\em reverse} vector, where $\vJ$ is the `reverse identity' matrix.
Finally, define $\coef_i(p(X))$ to be the coefficient of the term $X^i$ in the polynomial $p(X)$.
With the above notations we have
\begin{lmm}
\label{lmm:reverse-convolution}
For any vector $\vu,\vv\in\F^N$, we have $\ip{\vu}{\vv} \left(\ip{\vu}{\vv^R}\right) =  \coef_{N-1}(\vu(X)\cdot \vv^R(X))$ ($\vu(X)\cdot \vv(X)$ resp.).
\end{lmm}
\begin{proof}
The proof follows from noting that the coefficient of $X^{N-1}$ in $\vu(X)\cdot \vv^R(X)$ is given by
\[\sum_{j=0}^{N-1} \vu[i]\cdot \vv^R[N-1-i] =\sum_{i=0}^{N-1} \vu[i]\cdot \vv[i] =\ip{\vu}{\vv},\]
where we used that fact that $\left(\vv^R\right)^R=\vv$.
\end{proof}

For notational convenience, define $D_{[i:j]}(X) = \prod_{k=i}^{j-1} D_k(X)$ ($i,j$ can be out of the range $[t:N]$ with the convention $D_i(X) = 1$ for $i$ outside the range). By Lemma~\ref{lmm:reverse-convolution}, we know that $\vc[i] = \coef_{\bar{N}-1}(f_i(X)\vb^R(X))$. By Lemma~\ref{lmm:structure2}, this means that $\vc[i]$ is the coefficient of $X^{\bar{N}-1}$ of
\[ \ve \vT_{[t:i+t]} D_{[t:N+t]}(X) \vF, \]
where $\vF \in \F^{t+1}$ is defined by
\begin{equation}
  \label{eq:Ab-initialization}
  \vF[t-i] = \frac{f_i(X)\vb^R(X)}{D_{[t:N+t]}(X)},
\end{equation}
and $\ve$ is the row vector $\begin{bmatrix} 0 & \cdots & 0 & 1 \end{bmatrix}$ (i.e.\ so we are considering the last row of $\vT_{[t:i+t]}$). Note that this expression is well-defined because $\vT_i D_i(X)$ is a matrix of polynomials, and $\vF$ is a vector of polynomials by~\eqref{eq:rational-assumption}, so the resulting product is a polynomial.

Note that the first half of these expressions can be written as $(\ve \vT_{[t:i+t]} D_{[t:N/2+t]}(X)) (D_{[N/2+t:N+t]}(X) \vF)$ for $0 \leq i < N/2$. By Lemma~\ref{lmm:T-range-sizes}, the left term $\vT_{[t:i+t]} D_{[t:N/2+t]}(X)$ is a matrix with polynomial entries where the last row has degrees bounded by $di + \bar{d}N/2 < \bar{N}/2$. So the coefficient of $\bar{N}$ in the whole product depends only on the higher-order $\bar{N}/2$ coefficients of the right term $D_{[N/2+t:N+t]}(X) \vF$. For convenience, define this operator which reduces a polynomial to one on its higher order coefficients
\[ \operatorname{Reduce}(f(X), n) = \frac{\left( f(X)\mod{X^n} \right) - \left( f(X)\mod{X^{n/2}} \right)}{X^{n/2}} \]
Thus if we define $\vF_\ell[i] = \operatorname{Reduce}( D_{[N/2+t:N+t]}(X)\vF[i], \bar{N})$, then $\vc[i]$ for $0\leq i < N/2$ is the coefficient of $X^{\bar{N}/2 - 1}$ in
\[ \ve \vT_{[t:i+t]} D_{[t:N/2+t]}(X) \vF_\ell.\]
Note that this has the same form as the original problem, but with every term of half the size.

Similarly, we examine the second half of the answer. We want the coefficient of $X^{\bar{N}-1}$ in $ \ve \vT_{[t:N/2+i+t]} D_{[t:N+t]}(X) \vF$ for $0 \leq i < N/2$, which can be written as $(\ve \vT_{[N/2+t:N/2+i+t]} D_{[N/2+t:N+t]}(X)) (\vT_{[t:N/2+t]} D_{[t:N/2+t]}(X) \vF)$. Note once again that the left matrix has degrees bounded by $\bar{N}/2$, so we only need the higher order $\bar{N}/2$ coefficients of the polynomials in the right term. Thus defining $\vF_r[i] = \operatorname{Reduce}( \vT_{[t:N/2+t]} D_{[t:N/2+t]}(X) \vF[i], \bar{N})$, then $\vc[N/2+i] : 0\leq i < N/2$ is the coefficient of $X^{\bar{N}/2-1}$ in
\[ \ve \vT_{[N/2+t:N/2+i+t]} D_{[N/2+t:N+t]}(X) \vF_r \]
which is once again a problem of half the size of the original.

The algorithm is formalized in Algorithm~\ref{algo:vect-mult-rat-new}, the correctness of which follows from the above discussion. The initial call is to $\Amult(\vF,m,0)$, where we assume that $N=2^{m}$.
\begin{algorithm}
\caption{\Amult($\vF,a,k$)}
\label{algo:vect-mult-rat-new}
\begin{algorithmic}[1]
\renewcommand{\algorithmicrequire}{\textbf{Input:}}
\Require{$\vT_{\left[\frac{bN}{2^d}:\frac{bN}{2^d}+\frac{N}{2^{d+1}}\right]}$ for $0\le d<m, 0\le b<2^d$}
\Require{$\vF,a,k$, such that $\deg(\vF[i]) \leq (d+\bar{d})2^a$ for $0\leq i \leq t$}
\renewcommand{\algorithmicensure}{\textbf{Output:}}
\Ensure{$\vc[i] = \coef_{X^{(d+\bar{d})2^a-1}} (\ve\vT_{[k:k+i]} D_{[k:k+2^a]}(X) \vF)$, for $0\le i<2^a$}
\State $n\gets 2^a$
\If{$n \le t$}\Comment{Base case}
\label{step:rat-base-case}
\State $\vc[k+i] \gets \coef_{X^{(d+\bar{d})n-1}} D_{[k:k+2^a]}(X)\vF[t-i]$ for $0 \leq i < n$
%coefficient of 
%	\For{every $0\le i<n$}
%        \State $q_i(X)\gets \sum_{j=0}^t h^{(k)}_{i,j}(X)p_j(X)D_{[k+i+1:k+n]}(X) = p_j(X)D_{[k+i+1:k+n]}$
%		\State $\vc[i]\gets$ coefficient of $X^{\Delta-1}$ in $q_i(X)$
%	\EndFor
\EndIf
%\For{ every $0\le j\le t$}\Comment{Do computation for the first recursive call}
%        \State $p'_j(X) \gets \operatorname{Extract}(p_j(X) D_{[k+n/2:k+n]}(X), \barN)$
%\EndFor
\State $\vF_\ell \gets \operatorname{Reduce}(D_{[k+n/2:k+n]}(X)\vF, (d+\bar{d})n)$
\State $\vc[k:k+n/2]\gets \Amult(\vF_\ell,a-1,k)$
\State $\vF_r \gets \operatorname{Reduce}(\vT_{[k:k+n/2]} D_{[k:k+n/2]}(X) \vF, (d+\bar{d})n)$
\label{step:rat-algo-mv}
\State $\vc[k+n/2:k+n]\gets \Amult(\vF_r,a-1,k+n/2)$
%\State $\left(p_{n/2+t}(X),\dots,p_{n/2}(X)\right)^T\gets \vT_{[k:k+n/2]}\left(p_t(X)\cdots,p_0(X)\right)^T$ \Comment{Do computation for second recursive call}
%\For{ every $0\le j\le t$}
%        \State $p'_{n/2+j}\gets \operatorname{Extract}(p_{n/2+j}(X), \Delta)$
%\EndFor
%\State $\vc[k+n/2:k+n]\gets \Amult(p'_{n/2}(X),\dots,p'_{n/2+t}(X),\Delta/2,a-1,k+n/2)$
%\State \Return{$\vc[k:k+n]$}
\end{algorithmic}
\end{algorithm}

We argue that Algorithm runs efficiently. In particular,
\begin{lmm}
\label{lmm:rat-algo-runtime}
A call to $\Amult(\vF,m,t)$ takes $O(t^2(\bar{d}+d)\cM(N)\log{N})$ many operations.
\end{lmm}
\begin{proof}
  First we note that $D_{[bN/2^d : bN/2^d + N/2^{d-1}]}(X)$ can be computed for all $0 \leq d < m, 0 \leq b < 2^d$ in $O(\bar{d}\cM(N)\log N)$ operations by a straightforward divide and conquer, so we compute these first within the time bound.

  Also note that the base case is just $t$ multiplications of polynomials of size $(d+\bar{d})t$, which takes $O( (d+\bar{d})t \cM(t))$ operations to compute.

If $T(n)$ is the number of operations needed for a call to Algorithm~\ref{algo:vect-mult-rat-new} with input size $a = \log_2 n$, then we will show that
\[ T(n) \le 2T(n/2)+O(t^2 (d+\bar{d})\cM(n)). \]
This will prove the claimed runtime.

The first recursive call only requires computing $D_{[k+n/2:k+n]}(X)\vF$ which consists of $t+1$ multiplications of degree $(d+\bar{d})n$ polynomials; this takes $O(t(d+\bar{d})\cM(n)\log{n})$ operations. For the second recursive call, the runtime is dominated by Step~\ref{step:rat-algo-mv} is matrix vector multiplication with dimension $t+1$ where each entry is a polynomial of degree $O( (d+\bar{d})n)$. Hence, this step takes $O(t^2(d+\bar{d})\cM(n)\log{n})$ many operations, as desired.
\end{proof}

We remark that the analysis here is essentially equivalent to that of Lemma~\ref{lmm:transpose-mult-comp}; both algorithms are bottlenecked by multiplication of a ranged transition matrix. The only difference is that the transition matrices for a $(d,\bar{d})$-degree recurrence have degrees scaled by a factor of $(d+\bar{d})$ as shown in Lemma~\ref{lmm:T-range-sizes}. Similarly, a simple modification of Lemma~\ref{cor:transpose-mult-pre} shows that the pre-processing step of computing $\vT_{[bN/2^d:bN/2^d + N/2^{d+1}]}$ for all $0 \leq d < m, 0 \leq b < 2^d$ takes $O(t^{\omega} (d+\bar{d}) \cM(N)\log N)$ operations here. Thus, we have argued the following result:
\begin{thm}
\label{thm:main-result-Ab-rat-new}
For any $(d,\bar{d})$-degree recurrence as in~\eqref{eq:recur-rational-equiv-new} satisfying \eqref{eq:rational-assumption}, with  $O(t^{\omega}(\bar{d}+d)\cM(N) \log N)$ pre-processing operations, any $\vA\vb$ can be computed with $O(t^2(\bar{d}+d)\cM(N)\log{N})$ operations over $\F$.
\end{thm}
\begin{cor}
  For any $(1,0)$-degree recurrence~\eqref{eq:intro-recur-poly}, with  $O(t^{\omega}\cM(N) \log N)$ pre-processing operations, any $\vA\vb$ can be computed with $O(t^2\cM(N)\log{N})$ operations over $\F$.
\end{cor}

\paragraph{$\vA\vb$ for modular recurrence}
Now we consider computing $\vA\vb$ where $\vA$ is square and defined by a degree $(d,\bard)$ recurrence $\pmod{M(X)}$ for some degree $N$ polynomial $M(X)$.
To compute this product, we will factor the matrix $\vA$ into matrices that we have already shown admit fast matrix-vector multiplication.

As usual $\vA$ is associated with polynomials $f_i(X) = \sum_{j=0}^{N-1} \vA[i,j]X^j$. By Lemma~\ref{lmm:structure2}, $\vf_i$ is the last element of $\vT_{[t:i+t]} \vF$, where this product is done $\pmod{M(X)}$. Now let the roots of $M(X)$ be $\alpha_0, \dots, \alpha_{N-1}$. Consider the $N\times N$ matrix $\vZ$ such that $\vZ[i,j] = f_i(\alpha_j)$. This can be factored into $\vZ = \vA\vV^T$, where $\vV$ is the $N\times N$ Vandermonde matrix on the $\alpha_j$ (Definition~\ref{defn:vandermonde}).

On the other hand, we can factor $\vZ$ in a different way. Let $D(X) = \prod D_i(X)$ and let $C(X)$ be its inverse $\pmod{M(X)}$. For all $i$, define the polynomial $g_i(X)$ which is the last element of $\vT_{[t:i+t]} (D(X)C(X)\vF)$, where this is computed over $\F[X]$. Note that these polynomials have degree at most $(d+\bar{d})N$; furthermore, $D(X)$ can be computed in time $O(dN\log^2{dN})$ by divide-and-conquer, and $C(X)$ can be computed in $\widetilde{O}(N)$ operations using the fast Euclidean Algorithm~\cite{yap2000fundamental}. But the $g_i(X)$ exactly satisfy recurrence~\eqref{eq:recur-rational-equiv-new} with constraint~\eqref{eq:rational-assumption}, so that we can run Algorithm~\ref{algo:vect-mult-rat-new}. Thus by Theorem~\ref{thm:main-result-Ab-rat-new}, we can efficiently multiply by the $N \times (d+\bar{d})N$ matrix $\vA'$ containing the coefficients of the $g_i(X)$. Finally, note that $g_i(\alpha_j) = f_i(\alpha_j)$ by their equivalence modulo $M(X)$. Therefore $\vZ = \vA' \vV'^T$, where $\vV'$ is the $N \times (d+\bar{d})N$ Vandermonde matrix on the $\alpha_j$.

Thus we have the factorization $\vA = \vA'\vV'^T\vV^{-T}$. The Vandermonde and inverse Vandermonde matrices have dimensions at most $(d+\bar{d})N$ and thus can be multiplied in $O((d+\bar{d})N\log^2 N)$ time, so each component of this factorization admits matrix-vector multiplication in order $O( t^2(d+\bar{d})N\log^2 N)$ operations. %Note that the only place we have used the assumption that $M(X)$ has distinct roots is in ensuring that $\vV^T$ has an inverse; this assumption is easily removed by letting $\vV'$ and $\vV$ be a confluent Vandermonde matrix when there are repeated roots.

We remark that this factorization of $\vA$ can be used to perform the multiplication $\vA^T\vb$ as well, with the same asymptotic runtime but a worse constant factor than directly running Algorithm~\ref{algo:transpose-mult}.

%This finishes the proof of Theorem~\ref{thm:recur-mod}.

\paragraph{$\vA\vb$ for matrix recurrence}
The derivation for equation~\eqref{eq:recur-matrix-reduction} actually shows that $\vA$ has the form $\sum_{j=0}^{r-1} \vH_j \vK_j$ where $\vH_j$ is a basic recurrence~\eqref{eq:intro-recur-poly} and $\vK_j$ is a Krylov matrix.
If we define $\vb_j = \vK_j \vb$, it suffices to compute $\sum_{j=0}^t \vH_j \vb_j$. Because each $\vH_j$ satisfies the same recurrence, it suffices to modify the initialization step of Algorithm~\ref{algo:vect-mult-rat-new} (note how this is dual to $\vA^T\vb$, where we only need to modify the post-processing step). We only need to replace the numerator of equation~\eqref{eq:Ab-initialization}, which becomes
\[ \vF[t-i] = \frac{\sum_{j=0}^t h_{i,j}^{(0)}(X)\vb_j^R(X)}{D_{[t:N+t]}(X)}. \]
By Lemma~\ref{lmm:structure2}, $h_{i,j}^{(0)}(X) = \delta_{i,j}$ so this initialization has no additional cost if the $\vb_j^R(X)$ are known.
%\frac{f_i(X)\vb^R(X)}{D_{[t:N+t]}(X)},

Thus just as in the $\vA^T\vb$ case, the product $\vA\vb$ for a matrix recurrence~\eqref{eq:recur} reduces directly to the $\vA\vb$ algorithm for the same recurrence but in a modulus. Thus this algorithm has the same complexity as in Theorem~\ref{thm:recur-general}.

\section{Succinct Representations and Multivariate Polynomials}
\label{sec:app-multi}

Here we present details and additional connections to the work in Section~\ref{sec:multi}.

First we recollect known algorithms for multipoint evaluation of multivariate polynomials. Then we show the omitted proofs of results in Section~\ref{sec:multi}. Finally we show an additional example of a matrix from coding theory that has low recurrence width.

\subsection{Algorithms for Multipoint Evaluation of Multivariate Polynomials}
Recall the multipoint evaluation problem:

\begin{defn}
Given an $m$-variate polynomial $f(X_1,\dots,X_m)$ such that each variable has degree at most $d-1$ and $N=d^m$ distinct points $\vx(i)=(x(i)_1,\dots,x(i)_m)$ for $1\le i\le N$, output the vector $\left(f(\vx(i))\right)_{i=1}^N$.
\end{defn}

We will use the following reduction from~\cite{KU11,bivariate}. Let $\alpha_1,\dots,\alpha_N$ be distinct points. For $i\in [m]$, define the polynomial $g_i(X)$ of degree at most $N-1$ such that for every $j\in [N]$, we have
\[g_i(\alpha_j)=x(j)_i.\]
Then as shown in~\cite{KU11}, the multipoint evaluation algorithm is equivalent to computing the following polynomial:
\[c(X)=f(g_1(X),\dots,g_m(X))\mod{h(X)},\]
where $h(X)=\prod_{j=1}^N (X-\alpha_i)$. In particular, we have $c(\alpha_i)=f(\vx(i))$ for every $i\in [N]$. Thus, we aim to solve the following problem:
\begin{defn}[Modular Composition]
Given a polynomial $f(X_1,\dots,X_m)$ with individual degree at most $d-1$ and $m+1$ polynomials $g_1(X),\dots,g_m(X),h(X)$ all of degree at most $N-1$ (where $N\eqdef d^m$), compute the polynomial
\[f(g_1(X),\dots,g_m(X))\mod{h(X)}.\]
\end{defn}

As was noted in~\cite{bivariate}, if for the multipoint evaluation problem all the $x(j)_1$ for $j\in [N]$ are distinct, then we can take $\alpha_j=x(j)_1$ and in this case $\deg(g_1)=1$ since we can assume that $g_1(X)=X$. We will see that this allows for slight improvement in the runtime. Also in what follows, we will assume that an $n\times n$ and $n\times n^2$ matrix can be multiplied with $O(n^{\omega_2})$ operations.

\subsubsection{Algorithm for the general case}

Consider Algorithm~\ref{algo:mod-comp-gen} (which is a straightforward generalization of the algorithm for $m=1$ from~\cite{BK78}). 

We will use $X^{\vim}$ for any $\vim=(i_1,\dots,i_m)$ to denote the monomial $\prod_{\ell=1}^m X_{\ell}^{i_{\ell}}$. Let $k$ be any integer that divides $d$ and define
\[q=\frac{d}{k}.\]
With this notation, write down $f$ as follows
\begin{equation}
\label{eq:f-decompose-gen}
f(X_1,\dots,X_m)=\sum_{\vjm\in\Z_q^m} \left(\sum_{\vim\in\Z_k^m} f_{\vjm,\vim}\cdot X^{\vim}\right)\cdot X^{\vjm\cdot k},
\end{equation}
where $f_{\vjm,\vim}$ are constants.

\begin{algorithm}
\caption{Algorithm for Modular Composition: general case}
\label{algo:mod-comp-gen}
\begin{algorithmic}[1]
\renewcommand{\algorithmicrequire}{\textbf{Input:}}
\Require{$f(X_1,\dots,X_m)$ in the form of~\eqref{eq:f-decompose-gen} and $g_1(X),\dots,g_m(X),h(X)$ of degree at most $N-1$ with $N=d^m$}
\renewcommand{\algorithmicensure}{\textbf{Output:}}
\Ensure{\[f(g_1(X),\dots,g_m(X))\mod{h(X)}\]}
\State Let $k$ be an integer that divides $d$\Comment{We will use $k=\sqrt{d}$}
\State $q\gets \frac{d}{k}$
\For{ every $\vim=(i_1,\dots,i_m)\in\Z_k^m$}
\label{step:mod-comp-gen-gi}
	\State $g_{\vim}(X)\gets \prod_{\ell=1}^m \left(g_{\ell}(X)\right)^{i_{\ell}}\mod{h(X)}$.
\EndFor
\For{ every $\vjm=(j_1,\dots,j_m)\in\Z_q^m$}
\label{step:mod-comp-gen-gj}
	\State $g^{\vjm}(X)\gets \prod_{\ell=1}^m \left(g_{\ell}(X)\right)^{j_{\ell}\cdot k}\mod{h(X)}$.
\EndFor
\For{ every $\vjm\in\Z_q^{m}$}
\label{step:mod-comp-gen-ai}
	\State $a_{\vjm}(X)\gets \sum_{\vim\in\Z_k^m} f_{\vjm,\vim}\cdot g_{\vim}(X)$
\EndFor
\State \Return{$\sum_{\vjm\in\Z_q^{m}} a_{\vjm}(X)\cdot g^{\vjm}(X)\mod{h(X)}$}
\label{step:mod-comp-gen-return}
\end{algorithmic}
\end{algorithm}

Algorithm~\ref{algo:mod-comp-gen} presents the algorithm to solving the modular composition problem. The correctness of the algorithm follows from definition. We now argue its runtime.

Note that for a fixed $\vim\in\Z_k^m$, the polynomial $g_{\vim}(X)$ can be computed in $\tO(mN)$ operations since it involves $m$ exponentiations and $m-1$ product of polynomials of degree at most $N-1$ mod $h(X)$. Thus, Step~\ref{step:mod-comp-gen-gi} overall takes $\tO(m\cdot k^m\cdot N)$ many operations. By a similar argument Step~\ref{step:mod-comp-gen-gj} takes $\tO(m\cdot q^m\cdot N)$ operations. Step~\ref{step:mod-comp-gen-return} needs $q^m$ polynomial multiplication (mod $h(X)$) and $q^{m}-1$ polynomial multiplication where all polynomial are of degree at most $N-1$ and hence, this step takes $\tO(q^m\cdot N)$ operations. So all these steps overall take $\tO(m\cdot \max(k,q)^m\cdot d^m)$ many operations.

So the only step we need to analyze is Step~\ref{step:mod-comp-gen-ai}. Towards this end note that for any $\vjm\in\Z_q^m$
\begin{align*}
a_{\vjm}(X) & = \sum_{\vim\in\Z_k^m} f_{\vjm,\vim}\cdot g_{\vim}(X)\\
& = \sum_{\vim\in\Z_k^m} f_{\vjm,\vim}\cdot \sum_{\ell=0}^{N-1} g_{\vim}[\ell]\cdot X^{\ell}\\
& = \sum_{\ell=0}^{N-1} \left(\sum_{\vim\in\Z_k^m} f_{\vjm,\vim}\cdot g_{\vim}[\ell]\right) X^{\ell}.
\end{align*}
Thus, if we think of the $q^m\times d^m$ matrix $\vA$, where $\vA[\vjm,:]$ has the coefficients of $a_{\vjm}(X)$, then we have
\[\vA=\vF\times \vG,\]
where $\vF$ is an $q^m\times k^m$ matrix with $F_{\vjm,\vim}=f_{\vjm,\vim}$ and $\vG$ is an $k^m\times d^m$ matrix with $G_{\vim,\ell}=g_{\vim}[\ell]$.  Let $\omega(r,s,t)$ be defined so that one can multiply an $n^r\times n^s$ with an $n^s\times n^t$ matrix with $n^{\omega(r,s,t)}$ operations.  If we set $k=d^{\eps}$ for some $0\le \eps\le 1$, we have that Algorithm~\ref{algo:mod-comp-gen} can be implemented with
\[\tO\left( m\cdot d^{m(\max(\eps,1-\eps)+1)}+ d^{m\cdot \omega(1-\eps,\eps,1)}\right)\]
many operations. It turns out that the expression above is optimized at $\eps=\frac{1}{2}$, which leads to an overall (assuming $m$ is a constant) $\tO\left(d^{\omega_2m/2}\right)$ many operations. Thus, we have argued that
\begin{thm}
\label{thm:mod-comp-gen}
The modular composition problem with parameters $d$ and $m$ can be solved with $\tO\left(d^{\omega_2m/2}\right)$ many operations.
\end{thm}

\subsubsection{A `direct' algorithm for the multipoint evaluation case}

We now note that one can convert Algorithm~\ref{algo:mod-comp-gen} into a ``direct" algorithm for the multipoint evaluation problem. Algorithm~\ref{algo:multipoint-eval} has the details.

\begin{algorithm}
\caption{Algorithm for Multipoint Evaluation}
\label{algo:multipoint-eval}
\begin{algorithmic}[1]
\renewcommand{\algorithmicrequire}{\textbf{Input:}}\Require{$f(X_1,\dots,X_m)$ in the form of~\eqref{eq:f-decompose-gen} and evaluation points $\va(i)$ for $i\in [N]$}
\renewcommand{\algorithmicensure}{\textbf{Output:}}
\Ensure{\[\left( f(\va(i)\right)_{i\in [N]}\]}
\For{ every $\vim=(i_1,\dots,i_{m/2})\in\Z_d^{m/2}$}
	%\For{ every $k\in [N]$}
        	\State $\vg_{\vim}\gets \left(\prod_{\ell=1}^{m/2} \left(a(k)_{\ell}\right)^{i_{\ell}}\right)_{k\in [N]}$
	%\EndFor
\EndFor
\For{ every $\vjm=(j_1,\dots,j_m)\in\Z_d^{m/2}$}
%\label{step:mod-comp-gen-gj}
	%\For{ every $k\in [N]$}
        	\State $\vg^{\vjm}\gets \left(\prod_{\ell=m/2+1}^m \left(a(k)_{\ell}\right)^{i_{\ell}}\right)_{k\in [N]}$
	%\EndFor
\EndFor
\For{ every $\vjm\in\Z_d^{m/2}$}
%\label{step:mod-comp-gen-ai}
	%\For{ every $k\in [N]$}
        	\State $\vb_{\vjm}\gets \left(\sum_{\vim\in\Z_d^{m/2}} f_{\vjm,\vim}\cdot \vg_{\vim}(k)\right)_{k\in [N]}$
	%\EndFor
\EndFor
\State \Return{$\sum_{\vjm\in\Z_d^{m/2}} \left\langle \vb_{\vjm}, \vg^{\vjm}\right\rangle$}
%\label{step:mod-comp-gen-return}
\end{algorithmic}
\end{algorithm}

The correctness of the algorithm again follows from definition. It is easy to check that the computation of $\vg_{\vim}, \vg_{\vjm}$ and the output vectors can be accomplished with $\tO(d^{3m/2})$ many operations. Finally, the computation of the $\vb_{\vjm}$ can be done with $\tO(d^{\omega_2m/2})$ many operations using fast rectangular matrix multiplication.
\subsubsection{Algorithm for the distinct first coordinate case}

We now consider the case when all the $x(j)_1$ for $j\in [N]$ are distinct: i.e.\ we assume that $g_1(X)=X$. In this case we re-write $f$ as follows:
\begin{equation}
\label{eq:f-decompose-diff}
f(X_1,\dots,X_m)=\sum_{\vjm\in\Z_q^{m-1}} \left(\sum_{\vim\in\Z_k^{m-1}} f_{\vjm,\vim}(X_1)\cdot X_{-1}^{\vim}\right)\cdot X_{-1}^{\vjm\cdot k},
\end{equation}
where $X_{-1}^{\vim}$ denotes the monomial on the variables $X_2,\dots, X_m$ and each $f_{\vjm,\vim}(X_1)$ is of degree at most $d-1$.

\begin{algorithm}
\caption{Algorithm for Modular Composition: distinct first coordinate case}
\label{algo:mod-comp-diff}
\begin{algorithmic}[1]
\renewcommand{\algorithmicrequire}{\textbf{Input:}}
\Require{$f(X_1,\dots,X_m)$ in the form of~\eqref{eq:f-decompose-diff} and $g_2(X),\dots,g_m(X),h(X)$ of degree at most $N-1$ with $N=d^m$}
\renewcommand{\algorithmicensure}{\textbf{Output:}}
\Ensure{\[f(X,g_2(X),\dots,g_m(X))\mod{h(X)}\]}
\State Let $k$ be an integer that divides $d$\Comment{We will use $k=\sqrt{d}$}
\State $q\gets \frac{d}{k}$
\For{ every $\vim=(i_2,\dots,i_m)\in\Z_k^{m-1}$}
\label{step:mod-comp-diff-gi}
        \State $g_{\vim}(X)\gets \prod_{\ell=2}^m \left(g_{\ell}(X)\right)^{i_{\ell}}\mod{h(X)}$.
\EndFor
\For{ every $\vjm=(j_2,\dots,j_m)\in\Z_q^{m-1}$}
\label{step:mod-comp-diff-gj}
        \State $g^{\vjm}(X)\gets \prod_{\ell=2}^m \left(g_{\ell}(X)\right)^{j_{\ell}\cdot k}\mod{h(X)}$.
\EndFor
\For{ every $\vjm\in\Z_q^{m-1}$}
\label{step:mod-comp-diff-ai}
        \State $a_{\vjm}(X)\gets \sum_{\vim\in\Z_k^{m-1}} f_{\vjm,\vim}(X)\cdot g_{\vim}(X)\mod{h(X)}$ 
\EndFor
\State \Return{$\sum_{\vjm\in\Z_q^{m-1}} a_{\vjm}(X)\cdot g^{\vjm}(X)\mod{h(X)}$}
\label{step:mod-comp-diff-return}
\end{algorithmic}
\end{algorithm}

Algorithm~\ref{algo:mod-comp-diff} shows how to update Algorithm~\ref{algo:mod-comp-gen} to handle this special case. Again the correctness of this algorithm follows from the definitions.

We next quickly outline how the analysis of Algorithm~\ref{algo:mod-comp-diff} differs from that of Algorithm~\ref{algo:mod-comp-gen}. First, the same argument we used earlier can be used to show that Steps~\ref{step:mod-comp-diff-gi},~\ref{step:mod-comp-diff-gj} and~\ref{step:mod-comp-diff-return} can be accomplished with $\tO(m\cdot \max(k,q)^{m-1}\cdot d^m)$ many operations.

As before the runtime is dominated by the number of operations needed for Step~\ref{step:mod-comp-diff-ai}. Towards this end note that
\begin{align*}
a_{\vjm}(X) & = \sum_{\vim\in\Z_k^{m-1}} f_{\vjm,\vim}(X)\cdot g_{\vim}(X)\\
& = \sum_{\vim\in\Z_k^{m-1}} f_{\vjm,\vim}(X)\cdot \sum_{\ell=0}^{d^{m-1}-1} g_{\vim}[\ell](X)\cdot (X^d)^{\ell}\\
& = \sum_{\ell=0}^{d^{m-1}-1} \left(\sum_{\vim\in\Z_k^{m-1}} f_{\vjm,\vim}(X)\cdot  g_{\vim}[\ell](X)\right) (X^d)^{\ell}.
\end{align*}
Note that in the above we have decomposed $g_{\vim}$ as a polynomial in powers of $X^d$ (instead of $X$ for Algorithm~\ref{algo:mod-comp-gen}). In particular, this implies that all of $f_{\vjm,\vim}(X)$ and $g_{\vim}[\ell](X)$ are polynomials of degree at most $d-1$. If we think of $a_{\vjm}(X)$ as polynomials in $X^d$ (with coefficients being polynomials of degree at most $2d-2$), we can represent the above as
\[\vA=\vF\times \vG,\]
where the $q^{m-1}\times k^{m-1}$ matrix $\vF$ is defined by $F_{\vjm,\vim}=f_{\vjm,\vim}(X)$ and the $k^{m-1}\times d^{m-1}$ matrix $\vG$ is defined by $g_{\vim,\ell}=g_{\vim}[\ell](X)$. Again if we set $k=n^{\eps}$, then the run time of Algorithm~\ref{algo:mod-comp-diff} is given by (we use the fact that each multiplication and addition in $\vF\times \vG$ can be implemented with $\tO(d)$ operations):
\[\tO\left( m\cdot d^{(m-1)\max(\eps,1-\eps)+m}+ d\cdot d^{(m-1)\cdot \omega(1-\eps,\eps,1)}\right)\]
many operations. It turns out that the expression above is optimized at $\eps=\frac{1}{2}$, which leads to an overall (assuming $m$ is a constant) $\tO\left(d^{1+\omega_2(m-1)/2}\right)$ many operations. Thus, we have argued that
\begin{thm}
\label{thm:mod-comp-gen-distinct}
The modular composition problem with parameters $d$ and $m$  for the case of $g_1(X)=X$  can be solved with $\tO\left(d^{1+\omega_2(m-1)/2}\right)$ many operations.
\end{thm}

\subsection{Omitted Proofs}

\subsubsection{Proof of Lemma~\ref{lmm:disp-rank-multivariate}}

%\ar{The proof below is NOT full fleshed out (It only shows that the rank is $d^{m-1}$. I need to try and figure it out on paper first. If it works out, I will put in the text here.}

We now prove Lemma~\ref{lmm:disp-rank-multivariate}. We will argue that all but at most $2^{m/2}d^{m/2-1}$ columns of $\vB^{(0)}$ are $\vzero$, which would prove the result. Towards this end we will use the expression  in~\eqref{eq:multivariate-recur-equiv}.

%For the proof it will be convenient to index the columns of the matrices by indices $\vi\in\Z_d^m$. Further, we will assume that the columns are sorted in increasing lexical order of their indices. 

For notational convenience for any index $\vi\in\Z_d^m$, we will use $X^{\vi}$ to denote the monomial $\prod_{j=1}^m X_j^{i_j}$. Then note that $\vA^{(m)}[:,\vi]$ is just the evaluation of the monomial $X^{\vi}$ on the points $\vx(1),\dots,\vx(N)$. Further, it can be checked (e.g.\ by induction) that that exists polynomials $P_{\vi}(X_1,\dots,X_m)$ such that $\vB^{(0)}[:,\vi]$ is the evaluation of $P_{\vi}(X_1,\dots,X_m)$ on the points $\vx(1),\dots,\vx(N)$.

To simplify subsequent expressions, we introduce few more notation. For any $S\subseteq [m/2,m]$, define the matrix
\begin{equation}
\label{eq:M-S}
\vM_S = \left(\prod_{j\in S}\vD_{X_j}^{-1}\right)\vA^{(m)}\left(\prod_{j\in [m/2,m]\setminus S} \vZ^{d^{m-j}}\right).
\end{equation}
We can again argue by induction that for every $\vi\in\Z_d^m$, we have that $\vM_s[:,\vi]$ is the evaluation of a polynomial $Q_S^{\vi}(X_1,\dots,X_m)$ on the points $\vx(1),\dots,\vx(N)$. Note that this along with~\eqref{eq:multivariate-recur-equiv}, implies that
\begin{equation}
\label{eq:P-def}
P_{\vi}(X_1,\dots,X_m) = \sum_{S\subseteq [m/2,m]} (-1)^{m/2+1-|S|} Q_S^{\vi}(X_1,\dots,X_m).
\end{equation}

Now we claim that
\begin{claim}
\label{clm:mnonmial-cancel}
For every $\vi\in\Z_d^m$ that has an index $m/2\le j^*\le m$ such that $i_{j^*}\ge 2$ and $S\subseteq [m/2,m]\setminus \{j^*\}$, the following holds:
\[Q_S^{\vi} (X_1,\dots X_m)= Q_{S\cup\{j^*\}}^{\vi}(X_1,\dots,X_m).\]
\end{claim}

We first argue why the claim above completes the proof. Fix any $\vi\in\Z_d^m$ that has an index $m/2\le j^*\le m$ such that $i_{j^*}\ge 2$. Indeed by pairing up all $S\subseteq [m/2,m]\setminus \{j^*\}$ with $S\cup\{j^*\}$, Claim~\ref{clm:mnonmial-cancel} along with~\eqref{eq:P-def} implies that
\[P_{\vi}(X_1,\dots,X_m)\equiv 0.\]
Note that this implies that $\vB^{(0)}[:,\vi]=\vzero$ if there exists an index $m/2\le j^*\le m$ such that $i_{j^*}\ge 2$. Note that there are at least $d^m-2^{m/2}\cdot d^{m/2-1}$ such indices, which implies that at most $2^{m/2}\cdot d^{m/2-1}$ non-zero columns in $\vB^{(0)}$, as desired.

\begin{proof}[Proof of Claim~\ref{clm:mnonmial-cancel}] 
For any subset $T\subseteq [m]$, recall that $\ve_T$ is the characteristic vector of $T$ in $\{0,1\}^m$. Then note that for any $S\subseteq [m/2,m]$, we have
\[Q_S^{\vi}(X_1,\dots,X_m)=X^{(\vi\ominus \ve_{[m/2,m]\setminus S})-\ve_S},\]
where $\ominus$ is subtraction over $\Z_d^m$ (and $-$ is the usual subtraction over $\Z^m$). Indeed the $\ominus \ve_{[m/2,m]\setminus S}$ term corresponds to the matrix $\left(\prod_{j\in [m/2,m]\setminus S} \vZ^{d^{m-j}}\right)$ and the $-\ve_S$ term corresponds to the matrix $\left(\prod_{j\in S}\vD_{X_j}^{-1}\right)$ in~\eqref{eq:M-S}.

Fix a $\vi\in\Z_d^m$ that has an index $m/2\le j^*\le m$ such that $i_{j^*}\ge 2$ and $S\subseteq [m/2,m]\setminus \{j^*\}$. Define $S'=S\cup \{j^*\}$. Then we claim that
\[(\vi\ominus \ve_{[m/2,m]\setminus S})-\ve_S = (\vi\ominus \ve_{[m/2,m]\setminus S'})-\ve_{S'},\]
which is enough to prove the claim.
To prove the above, we will argue that
\[(\vi\ominus \ve_{[m/2,m]\setminus S}) = (\vi\ominus \ve_{[m/2,m]\setminus S'})-\ve_{j^*}.\]
Note that the above is sufficient by definition of $S'$. Now note that $\ve_{[m/2,m]\setminus S}= \ve_{[m/2,m]\setminus S'}+\ve_{j^*}$. In particular, this implies that it is sufficient to prove
\begin{equation}
\label{eq:no-diff-minuses}
(\vi\ominus \ve_{[m/2,m]\setminus S'})\ominus \ve_{j^*} = (\vi\ominus \ve_{[m/2,m]\setminus S'})-\ve_{j^*}.
\end{equation}
By the assumption that $i_{j^*}\ge 2$, we have that 
\[(\vi\ominus \ve_{[m/2,m]\setminus S'})_{j^*}\ge 1,\]
which implies that both $\ominus \ve_{j^*}$ and $-\ve_{j^*}$ have the same effect on $(\vi\ominus \ve_{[m/2,m]\setminus S'})$, which proves~\eqref{eq:no-diff-minuses}, as desired.
%\ar{If there is a subtle bug it is probably here so please double-check!}
\end{proof}

\subsection{Multipoint evaluation of multivariate polynomials and their derivatives}
\label{subsec:multi-derivative}

In this section, we present another example of a matrix with our general notion of recurrence that has been studied in coding theory. We would like to stress that currently this section does not yield any conditional ``lower bounds" along the lines of Theorem~\ref{thm:lb-1} or~\ref{thm:lb-2}.

We begin by setting up the notation for the derivative of a multivariate polynomial $f(X_1,\dots, X_m)$. In particular, given an $\vim=(i_1,\dots,i_m)$, we denote the $\vim$th derivative of $f$ as follows:
\[f^{(\vim)}(X_1,\dots,X_m)=\frac{\partial^{i_1}}{\partial X_1}\cdots\frac{\partial^{i_m}}{\partial X_m} f,\]
where $\frac{\partial^0}{\partial X} f =f$. If the underlying field $\F$ is a finite field, then the derivatives are defined as the {\em Hasse} derivatives.

We consider the following problem.
\begin{defn}
Given an $m$-variate polynomial $f(X_1,\dots,X_m)$ such that each variable has degree at most $d-1$, an integer $0\le r<d$ and $n=d^m$ distinct points $\va(i)=(a(i)_1,\dots,a(i)_m)$ for $1\le i\le N$, output the vector 
\[\left(f^{(\vim)}(\va(j)) \right)_{j\in [n], \vim\in\Z_r^m}.\]
 %where $f^{(\vim)}$ is the $\vim$th derivative of $f$.
\end{defn}
The above corresponds to (puncturing) of multivariate multiplicity codes, which have been studied recently in coding theory~\cite{derivative-code-1,derivative-code-2,derivative-code-3}. These codes have excellent local and list decoding properties. We note that in definition of multiplicity codes, $d$ and $r$ are limits on the total degree and the total order of derivatives while in our case these are bounds for individual variables. However, these change the problem size by only a factor that just depends on $m$ (i.e.\ we have at most a factor $m!$ more rows and columns). Since we think of $m$ as constant, we ignore this difference. For such codes in~\cite{derivative-code-1,derivative-code-2} the order of the derivatives $r$ is assumed to be a constant. However, the multiplicity code used in~\cite{derivative-code-3}, $r$ is non-constant. We would like to mention that these matrices turn up in list decoding of Reed-Solomon and related codes (this was also observed in~\cite{OS99}).\footnote{However in the list decoding applications $\vA$ is not square and one is interested in obtaining a non-zero element $\vf$ from its kernel: i.e.\ a non-zero $\vf$ such that $\vA\vf=\vzero$.} In particular, for the Reed-Solomon list decoder of Guruswami and Sudan~\cite{GS98}, the algorithm needs $\vA$ with $m=2$ and $r$ being a polynomial in $n$.
%\ar{Need to send Swastik an email to see if our params make sense in some other settings.}

We note that the problem above is the same as $\vA\cdot \vf$, where $\vf$ is the vector of coefficients of $f(X_1,\dots,X_m)=\sum_{\vjm\in\Z_d^m} f_{\vjm} X^{\vjm}$, where as before $X^{\vjm}$ is the monomial $\prod_{\ell=1}^m X_j^{j_{\ell}}$ and the matrix $\vA$ is defined as follows:

\[A_{(k,\vim),\vjm}=\prod_{\ell=1}^m \binom{j_{\ell}}{i_{\ell}} \left(a(k)_{\ell}\right)^{j_{\ell}-i_{\ell}},\]
for $k\in [n]$, $\vjm=(j_1,\dots,j_m)\in\Z_d^m$ and $\vim=(i_1,\dots,i_m)\in\Z_r^m$. We note that the definition holds over all fields. %, whether finite or not.

We use the convention that $\binom{b}{c}=0$ if $b<c$.

The aim of the rest of the section is to show that the matrix $\vA$ satisfies a recurrence that we can handle.

For notational simplicity, let us fix $k\in [n]$ and we drop the dependence on $k$ from the indices. In particular, consider the (submatrix):
\[A_{\vim,\vjm}=\prod_{\ell=1}^m \binom{j_{\ell}}{i_{\ell}} \left(a_{\ell}\right)^{j_{\ell}-i_{\ell}}.\]
Think of $\vim=(\vim^0,\vim^1)$, where $\vim^0=(i_1,\dots,i_{m'})$ and $\vim^1=(i_{m'+1},\dots,i_m)$ for some $1\le m'<m$ to be determined. Now fix an $\vim=(\vim^0,\vim^1)$ such that $\vim^1\neq\vzero$ and define \[S_{\vim}=\{\ell\in\{m'+1,\dots,m\}|i_{\ell}\neq 0\}.\]
%Further fix a $\vjm$ such that $\vim^1\preceq \vjm^1$ (i.e.\ for every $\ell\in (m',m]$, we have $j_{\ell}\ge i_{\ell}$).
Consider the following sequence of relations:
{\allowdisplaybreaks
\begin{align}
\label{eq:derivative-s1}
A_{\vim,\vjm}&=\left(\prod_{\ell\in [m]\setminus S_{\vim}} \binom{j_{\ell}}{i_{\ell}} \left(a_{\ell}\right)^{j_{\ell}-i_{\ell}}\right)\cdot\left(\prod_{\ell\in S_{\vim}}\left(\binom{j_{\ell}-1}{i_{\ell}-1}+\binom{j_{\ell}-1}{i_{\ell}}\right) a_{\ell}^{j_{\ell}-i_{\ell}}\right)\\
\label{eq:derivative-s2}
&=\sum_{T\subseteq S_{\vim}} \prod_{\ell=1}^m \binom{j_{\ell}-\ind{\ell\in S_{\vim}}}{i_{\ell}-\ind{\ell\in T}} a_{\ell}^{j_{\ell}-i_{\ell}}\\
&=\sum_{T\subseteq S_{\vim}} \left(\prod_{\ell=1}^m a_{\ell}^{\ind{\ell\in S_{\vim}}-\ind{\ell\in T}}\right)\cdot\left(\prod_{\ell=1}^m \binom{j_{\ell}-\ind{\ell\in S_{\vim}}}{i_{\ell}-\ind{\ell\in T}} a_{\ell}^{j_{\ell}-\ind{\ell\in S_{\vim}}-(i_{\ell}-\ind{\ell\in T})}\right)\notag\\
\label{eq:derivative-s3}
&=\sum_{T\subseteq S_{\vim}} \va^{\ve_{S_{\vim}}-\ve_T}\cdot A_{\vim-\ve_T,\vjm-\ve_{S_{\vim}}}.
\end{align}
}
In the above~\eqref{eq:derivative-s1} follows from the following equality for integers $b\ge 0$ and $c\ge 1$, $\binom{b}{c}=\binom{b-1}{c-1}+\binom{b-1}{c}$ while~\eqref{eq:derivative-s2} follows from the notation that $\ind{P}$ is the indicator value for the predicate $P$.

Consider the matrix $\vE$ defined as follows:
\begin{equation}
\label{eq:derivative-recur}
\vE[(k,\vim),:]=\vA[(k,\vim),:]- \sum_{T\subseteq S_{\vim}} \va^{\ve_{S_{\vim}}-\ve_T}\cdot \vA[(k,\vim-\ve_T),:]\cdot \vZ^{\sum_{\ell\in S_{\vim}} d^{m-\ell}}.
\end{equation}

By~\eqref{eq:derivative-s3}, we have that for every $\vim=(\vim^0,\vim^1)$ such that $\vim^1\neq\vzero$ %and every $\vjm=(\vjm^0,\vjm^1)$ with $\vim^1\preceq \vjm^1$, we have that
\[E[(k,\vim),\vjm]=0.\]
%Since $\vjm^1\in\Z_d^{m-m'}\setminus \Z_r^{m-m'}$ implies that $\vim^1\preceq \vjm^1$ for any $\vim\in\Z_r^m$, we have that all the non-zero elements of $\vE$ are covered by the columns $\vE[:,\vjm]$ with $\vjm^1\in \Z_r^{m-m'}$ and 
In other words, the  non-zero rows  $(k,\vim)$ must satisfy $\vim^1=\vzero$. Thus, we have argued that
\begin{lmm}
\label{lmm:rank-err-derivative}
$\vE$ as defined in~\eqref{eq:derivative-recur} has rank at most $nr^{m'}$.
\end{lmm}
\begin{proof}
This follows from the fact that there are $n\cdot r^{m'}$ many indices $(k,\vim)$ with $\vim^1=\vzero$. % and there are $r^{m-m'}$ many indices $\vjm$ with $\vjm^1\in \Z_r^{m-m'}$.
\end{proof}

This in turn implies the following:
\begin{lmm}
\label{lmm:derivative-recur-param}
The recurrence in~\eqref{eq:derivative-recur} has \width\ of  $\left(t= \frac{r^{m-m'+1}-1}{r-1},nr^{m'}\right)$ and degree $\left(D=\frac{d^{m-m'+1}-1}{d-1},D\right)$.
\end{lmm}

\subsubsection{Instantiation of parameters}

We now consider few instantiation of parameters to get a feel for Lemma~\ref{lmm:derivative-recur-param}.

We start with the case of $n=1$: note that in this case we have $r=d$ (since we want $N=n\cdot r^m=d^m$). In this case the choice of $m'$ that makes all the parameters roughly equal is $m'=m/2$. In this case we get that~\eqref{lmm:derivative-recur-param} has \width $(2d^{m/2},d^{m/2})$ and is degree $(2d^{m/2},2d^{m/2})$. However, this does not give anything algorithmically interesting since the input size for such a recurrence is already $\Omega(d^{m/2}\cdot d^{m/2}\cdot N+d^{m/2}\cdot N)=\Omega(N^2)$. 

Recall that a recurrence with \width\ $(T,R)$ that is degree $(D,D)$ has an input size of $O((TD+R)N)$. Thus, to decrease the input size of the recurrence in~\eqref{eq:derivative-recur}, we need to pick $m'$ such that $2(m-m')=m'$. This implies that we pick $m'=2m/3$ and thus we have a recurrence with \width\ $(2d^{m/3},2d^{m/3})$ that is degree $(2d^{m/3},2d^{m/3})$ for an overall input size of $O(d^{5m/6})=O(N^{5/3})$.

%Now assume that we have any algorithm that solves 

%\ar{Am not sure if this recurrence with our algorithms is going to lead to anything non-trivial. So probably not much here except for an example of a natural matrix that fits our general recurrence definition.}

\section{Miscellaneous Results}
\subsection{Efficient Jordan Decomposition}
\label{subsec:jordan}

Every linear operator admits a Jordan normal form, which in some sense is the smallest and most uniform representation of the operator under any basis. Knowing the Jordan decomposition of a matrix permits many useful applications. For example, being able to quickly describe and manipulate the change of basis matrix, which is a set of generalized eigenvectors, subsumes calculating and operating on the eigenvectors of the matrix.

Another use of the Jordan decomposition is being able to understand the evaluation of any analytic function on the original matrix. We highlight this application because of its connection to Krylov efficiency in Section~\ref{sec:Krylov}, which can be solved by evaluating the matrix at a certain polynomial function. Computing matrix functions has widespread uses, and there is a significant body of research on the design and analysis of algorithms for various matrix functions, including the logarithm, matrix sign, $p$th root, sine and cosine~\cite{higham2008}. Perhaps the most prominent example of a matrix function is the matrix exponential, which is tied to the theory of differential equations - systems are well-approximated around equilibria by a linearization $\vx'(t) = \vM\vx(t)$, which is solved by the exponential $\vx(t) = e^{t\vM}\vx(0)$~\cite{hartman1960}. We note that much work has been done on computing the matrix exponential in particular without going through the Jordan decomposition, since it is generally hard to compute~\cite{putzer1966}.
 \agu{import of old min=char result for Jordan decomposition. pretty unhappy about how this section looks right now.. will clean up if time permits\ldots}
The Jordan form can be used to easily compute matrix functions. Consider a matrix $\vM$ with a Jordan decomposition $\vM = \vA\vJ\vA^{-1}$. We say that $\vM$ is $(\alpha,\beta)$-Jordan efficient if $\vA$ and $\vA^{-1}$ admit super-fast matrix-vector multiplication in $O(\beta)$ operations with $O(\alpha)$ pre-processing steps.
 %\documentclass{article}[11pt]
%
%\input{macros}
%\usepackage{arydshln}
%\usepackage{comment}
%
%
%\begin{document}

%\subsection{Efficient Jordan Decompositions}

Now consider a triangular $\Delta$-banded $N\times N$ matrix $\vM$ whose minimal polynomial has full degree. Throughout this section, we assume $\vM$ is upper triangular. We will prove that $\vM$ is $(\Delta^{\omega_\calR}, \Delta^2)$-Jordan efficient in this case.

\subsubsection{Recurrence of $A$}
\label{sec:recurrence}
$\vM$ is similar to $\vF_\vM$ where ${\vF_\vM}^T$ is the Frobenius companion matrix for $c_\vM(X)$. In particular, suppose $c_\vM(X) = X^N - c_{N-1} X^{N-1} \dots - c_0$, then $\vF_\vM$ will be
\[\begin{bmatrix}
  0 & 1 & 0 & \dots & 0 \\
  0 & 0 & 1 & \dots & 0 \\
  \vdots & \vdots & \vdots & \ddots & \vdots \\
  0 & 0 & 0 & \dots & 1 \\
  c_0 & c_1 & c_2 & \dots & c_{N-1} 
\end{bmatrix}\]
This matrix ${\vF_\vM}^T$ is simply the Frobenius normal form, i.e., the canonical rational form, of $\vM^T$.

For a matrix $\vA$, we define $a_i(X)$ as the polynomial corresponding to the $i^{th}$ row of $\vA$, i.e., $a_i(X) = \sum_{j=0}^N \vA[i,j] X^j$. 

\begin{lmm}
  For some matrix $\vA$, $\vM \vA = \vA \vF_\vM$ if and only if $(\vM- X \vI) \begin{bmatrix} a_0(X) \\ \vdots \\ a_{N-1}(X) \end{bmatrix} = 0 \mod c_\vM(X)$.
\end{lmm}
\begin{proof}
  Multiplying a vector by $\vF_\vM$ in $\calR$ is isomorphic to multiplying the corresponding polynomial by $X$ in $\calR[X]/(c_\vM(X))$. This implies that $\vM \begin{bmatrix} a_0(X) \\ \vdots \\ a_{N-1}(X) \end{bmatrix} = X \begin{bmatrix} a_0(X) \\ \vdots \\ a_{N-1}(X) \end{bmatrix} \pmod{c_\vM(X)}$. 
\end{proof}

\begin{cor}
  \label{cor:jordan-recurrence}
  There exists a $\Delta-1$ width matrix $\vA$ such that $\vM \vA = \vA \vF_\vM$. More specifically, the rows of $\vA$ satisfy
  \begin{align}
	(X - \vM[i,i]) a_i(X) = \sum_{j=1}^{\min(\Delta, N-i)} \vM[i,i+j] a_{i+j}(X)+d_i c_\vM(X)\label{eqn:delta-recur}
  \end{align}
  where $d_i\in\mathbb F$. 
\end{cor}
\begin{proof}
  Note that the matrix $\vM$ we are interested in is an upper-triangular, $\Delta$-banded matrix. Our previous lemma implies that that $\vM \vA = \vA \vF_\vM$ if and only if $(\vM - X\vI)  \begin{bmatrix} a_0(X) \\ \vdots \\ a_{N-1}(X) \end{bmatrix} = 0 \pmod{c_\vM(X)}$.

  We remove the modulus and treat the $a_i(X)$ as polynomials over $\mathbb F[X]$ of degree less than $N$. The above condition becomes the system of equations~\ref{eqn:delta-recur}.
Since the left side has degree $N$, the $d_i$ must have degree $0$ so they are scalars. Thus if a family of polynomials is defined to satisfy the recurrence above, the associated $\Delta-1$ width matrix $A$ must satisfy $\vM \vA = \vA \vF_\vM$. 
\end{proof}

We say that a recurrence of form (\ref{eqn:delta-recur}) has \emph{size $N$} if the total length of the recurrence has length $N$ and $c(X)$ has degree $N$, so that all polynomials are bounded by degree $N-1$. We call the scalars $\vM[\cdot,\cdot]$ the recurrence coefficients, and the $d_i$ the error coefficients.

From now on, our use of $\vA$ will denote such a $\Delta-1$ width matrix. Note that independent of the $d_i$ coefficients, the following divisibility lemma holds.

\begin{lmm}
  \label{lmm:triangular}
  $\prod_{j=0}^{i-1} (X-\vM[j,j]) | a_i(X)$ for $0 \le i \le N-1$. 
\end{lmm}
\begin{proof}
  Proof by induction. As a base case, this is true for $a_{N-1}(X)$ by equation (\ref{eqn:delta-recur}) since $\prod_{j=0}^{N-1}(x-\vM[j,j]) = c_{\vM}(X)$. And (\ref{eqn:delta-recur}) gives us the inductive step.
\end{proof}

\subsubsection{Conditions on the Error Coefficients}
\label{sec:error-coef}

We showed in Section~\ref{sec:recurrence} that the equation $\vM\vA = \vA\vF$ is equivalent to a recurrence (\ref{eqn:delta-recur}). In order to complete the task of finding $\vA$ satisfying $\vM = \vA\vF\vA^{-1}$, we must find scalars $d_i$ in the recurrence that lead to an invertible matrix $\vA$. The goal of this subsection is in proving a strong sufficiency condition on such sequences $d_i$.

We will first show some equivalent conditions to $\vA$ being invertible. We will make heavy use of (confluent) Vandermonde matrices here (Definition~\ref{defn:vandermonde}), and for convenience define $\vV_\vM$ to be the Vandermonde matrix $\vV_{\vM[0,0], \dots, \vM[N-1, N-1]}$.
%. For convenience, we re-define the confluent Vandermonde matrix (Definition~\ref{defn:vandermonde}) here.
%\begin{defn}
%  Given a set of points $p_0, \dots, p_{N-1}$, let $n_i = \displaystyle\sum_{j=0}^{i-1} \mathbb{1}(p_j = p_i)$ be the number of preceding points identical to $p_i$. Then the confluent Vandermonde matrix, denoted $\vV_{p_0, \dots, p_{N-1}}$, is defined such that \[\vV[i,j] = \begin{cases} 0, & j < n_i \\  \frac{j!}{(j-n_i)!(n_i)!} {p_i}^{j-n_i}, &j \ge n_i\end{cases}\]  for $0 \le i, j \le N-1$. \footnote{Some sources define the confluent Vandermonde matrix with the factor of $n_i!$ in the denominator~\cite{pan-book}, and others do not. It makes no real difference, but is slightly more convenient for us to use the former notation.}
%\end{defn}
%\agu{The confluent Vandermonde is sometimes defined with the $n_i!$ in the denominator and sometimes without. Wikipedia defines it the latter way but we want the former. Might be worth a footnote}

%Note that for any vector $\vy$, $\vV_{p_0, \dots, p_{N-1}} \vy$ is equivalent to evaluating the polynomial $v(X) = \displaystyle\sum_{i=0}^{N-1} \vy[i] X^i$ at $v^{(n_i)}(p_i)$ for each $i$. We define $\vV_\vM$ to be the confluent Vandermonde matrix $\vV_{\vM[0,0], \dots, \vM[N-1, N-1]}$.

\begin{lmm}
  \label{lmm:triangular-equiv}
  The following are true for $\vA$ defined by recurrence (\ref{eqn:delta-recur}), for any $d_i$.
  \begin{enumerate}[label=(\alph*)]
	\item $\vA\vV_{\vM}^T$ is upper triangular.
	\item $a_i^{(n_j)}(\vM[j,j]) = 0$ for all $i$ and $j < i$.
	\item $\prod_{j<i}(X-\vM[j,j]) \,|\, a_i(X)$ for all $i$.
  \end{enumerate}
\end{lmm}
\begin{proof}
  We show the equivalence of the conditions. This is sufficient since $(c)$ is true by Lemma~\ref{lmm:triangular}.
  \begin{description}
	\item[$(a) \iff (b)$]
	  As noted above, $\vA \vV_\vM^T[i,j]$ can be understood as a Hermetian evaluations (evaluation a polynomial and its derivatives) and is equal to $a_i^{(n_j)}(\vM[j,j])$. The equivalence follows since upper triangularity the same as saying $\vA\vV_\vM^T[i,j] = 0$ for all $j<i$.
	\item[$(b) \iff (c)$]
	  Fix an $i$. For every $\lambda$, let $n_\lambda$ be its multiplicity in $\prod_{j<i} (X-\vM[j,j])$. Note that condition (c) is equivalent to saying $(X-\lambda)^{n_\lambda} \,|\, a_i(X)$ for all $\lambda$. For a fixed $\lambda$, the divisibility condition $(X-\lambda)^{n_\lambda} \,|\,a_i(X)$ is equivalent to $a_i^{(k)}(\lambda) = 0$ for $k < n_\lambda$. This can be seen, for example, by considering the Taylor expansion of $a_i$ around $\lambda$. Finally, the union of these equations over all $\lambda$ is exactly (b), completing the equivalence.
  \end{description}
\end{proof}

The equivalence of the following conditions follows easily from the same reasoning.
\begin{cor}
  \label{lmm:invertible-equiv}
  The following are equivalent for $\vA$ defined by recurrence (\ref{eqn:delta-recur}).
  \begin{enumerate}[label=(\alph*)]
	\item $\vA$ is invertible.
	\item $\vA\vV_{\vM}^T$ is invertible.
	\item $a_i^{(n_i)}(\vM[i,i]) != 0$.
	\item $\prod_{j\leq i}(X-\vM[j,j]) \,\not|\, a_i(X)$\label{cond4}
  \end{enumerate}
\end{cor}

We use Corollary~\ref{lmm:invertible-equiv} and in particular \ref{cond4} to find conditions when $\vA$ is invertible. For convenience, define $p_i(X) = \prod_{j < i} \vM[j,j]$.

Consider a fixed $i$. We will show that if $d_{i+1},\dots,d_{N-1}$ are fixed such that \ref{cond4} is satisfied for all $a_{[i+1:N]}(X)$, then we can choose $d_i$ such that \ref{cond4} is satisfied for $a_i(X)$ as well.

\begin{description}
  \item{Case 1}: There does not exist $j > i$ such that $\vM[j,j] = \vM[i,i]$.

	By the recurrence,
	\[ a_i(X) = \frac{\sum_{j>i} \vM[i,j]a_j(X)}{X-\vM[i,i]} + d_i \frac{p_N(X)}{X-\vM[i,i]} \]

	Note that the first term on the RHS is a polynomial and also is a multiple of $p_i(X)$ by inductively invoking Lemma~\ref{lmm:triangular}. Note that the second term is a multiple of $p_i(X)$ but not $p_{i+1}(X)$  by the assumption for this case. Thus there exists some $d_i$ such that $\prod_{j\leq i}(X-\vM[j,j]) \,\not|\, a_i(X)$ holds - in fact, there is only one $d_i$ such that it doesn't hold.

  \item{Case 2}: There exists $j > i$ such that $\vM[j,j] = \vM[i,i]$. Choose $j$ to be the largest such index.

	\begin{claim}
	  Fix $a_{j}$ and follow the recurrence (1) without adding multiples of the modulus - i.e. let $d_{[i:j]} = 0$. Let this family of polynomials be $a'_{i}(X)$. Then $p_{i+1}(X) | a_i(X)$ if and only if $p_{i+1}(X) | a'_i(X)$.
	\end{claim}
	\begin{proof}
	  $a_i(X)$ and $a'_i(X)$ differ only through the terms $c_{i}p_N(X), \cdots, c_{j-1}p_N(X)$ added while following (1) with the modulus. All of these terms contribute a multiple of $p_{i+1}(X)$ to row $i$ so $a_i(X) \equiv a'_i(X) \pmod{p_{i+1}(X)}$.
	\end{proof}

	\begin{claim}
	  $p_{i+1}(X) | a'_i(X)$ if and only if $p_{j+1}(X) | a_j(X)$.
	\end{claim}
	\begin{proof}
	  The recurrence without the mod can be written as the product of $\Delta\times\Delta$ transition matrices, and in particular we can express $a'_i(X)$ via the equation
	  \[ \frac{a'_i(X)}{p_i(X)} = h_{i,j}(X)\frac{a_j}{p_j} + h_{i,j+1}(X-\vM[j,j])\frac{a_{j+1}}{p_{j+1}} + \dots + h_{i,j+\Delta}(X-\vM[j,j])\cdots(X-\vM[j+\Delta-1,j+\Delta-1])\frac{a_{j+\Delta}}{p_{j+\Delta}} \]
	  % TODO? \agu{This seems like it would benefit from a work-horse lemma that says $a_i$ depends on $a_j$ through rational functions whose denominators are $p_{\lambda_{[i:j]}}$}

	  Multiplying through again to isolate $a'_i$ and taking this mod $p_{i+1}(X)$, we see that only the first term affects whether $p_{i+1}(X) | a'_i(X)$.

			 %Elaboration: Because $[a'_i \dots a'_{i+\Delta}] = T_{[i:j]} [a_j \dots a_{j+\Delta}] = \frac{1}{(x-M_{ii})\dots(x-M_{j-1,j-1})}T' [a_j \dots a_{j+\Delta}]$, where $T' \in \mathbb{F}[X]^{\Delta\times\Delta}$. Dividing through by $p_i(X)$, $a'_i/p_i(X) = h_{i,j}(X)(a_j/p_j) + \dots + h_{i,j+1}(X-M_{jj})(a_{j+1}/p_{j+1}) + \dots$. Multiplying through again to isolate $a'_i$ and taking this mod $p_{i+1}(X)$, we see that only the first term contributes.

          Thus the claim is equivalent to saying that $h_{i,j}(X)$ is not a multiple of $(X-\vM[i,i])$. If it was, then note that $a'_i(X)$ will \emph{always} be a multiple of $p_{i+1}(X)$ no matter what $a_j(X)$ is. Then $a_i(X)$ will be a multiple of $p_{i+1}(X)$ no matter what $a_{i+1}, \cdots, a_{N-1}$ are, and there is \emph{no} family of polynomials satisfying \ref{cond4} of Corollary~\ref{lmm:invertible-equiv}. This is a contradiction since there is some $\vA$ satisfying recurrence~\eqref{eqn:delta-recur} that is invertible - since there is a change of basis matrix $\vA$ such that $\vM = \vA\vF\vA^{-1}$, and by Corollary~\ref{cor:jordan-recurrence} it satisfies~\eqref{eqn:delta-recur}.
	\end{proof}
\end{description}

The results above thus imply the following result.
\begin{thm}
  \label{thm:greedy-error-coef}
  Let $\vA$ be a matrix satisfying (\ref{eqn:delta-recur}).
  \begin{itemize}
	\item $\vA \vV_\vM^T$ is upper triangular for any sequences $d_i$.
	\item We can pick $d_{N-1}, \ldots, d_{0}$ in order, such that each $d_i$ chosen to satisfy the local condition $a_i^{(n_i)}(M[i,i]) != 0$. Then the matrix $\vA \vV_\vM^T$ will be invertible. Furthermore, if there exists $j>i$ such that $M[i,i] = M[j,j]$ then $d_i$ can be anything, in particular $0$.
  \end{itemize}
\end{thm}

\subsubsection{Finding the Error Coefficients and Inverting $\vA {\vV_\vM}^T$}
\label{sec:error-coef2}

The main goal of this section is to prove a structure result on $\vA \vV_\vM^T$, that will easily allow us to
\begin{itemize}
  \item Given a recurrence with unspecified error coefficients, pick the $d_i$ such that $\vA\vV_\vM^T$ is invertible.
  \item Given a fully specified recurrence such that $\vA\vV_\vM^T$ is invertible, multiply $(\vA\vV_\vM^T)^{-1}\vb$ fast.
\end{itemize}

Suppose we have fixed error coefficients $d_{[0:N]}$ and consider matrix $\vA$ corresponding to the polynomials generated by a recurrence of the form (\ref{eqn:delta-recur})
\begin{equation}
  \label{eqn:delta-recur2}
  (X - \lambda_i) a_i(X) = \sum_{j=1}^{\min(\Delta, N-i)} c_{i,j} a_{i+j}(X)+d_i p_{\lambda_{[0:N]}}(X)
\end{equation}

(We will use the shorthand notation $p_{\alpha,\beta,\ldots}(X) = (X-\alpha)(X-\beta)\cdots$.)

%\agu{I changed the notation slightly to use $\lambda_i$ and $p(X)$ in place of $\vM[i,i]$ and $c_\vM, c_L, c_R$. The purpose was twofold: First, to abstract this away from a specific matrix $\vM$. Second, because I find it slightly easier to reason about recursion with general notation that applies to all levels, instead of the notation $c_\vM$ which refers to one particular polynomial at the top level. These points are pretty minor though and in retrospect might not be optimal\ldots I'm fine with whatever people want}

By triangularity, we can partition the matrix as follows
\[
  \vA \vV_{\lambda_{[0:N]}}^T =
  \left[\begin{array}{c:c}
	  \vT_L & \vB \\ \hdashline
	  0 & \vT_R
  \end{array}\right].
\]

The core result is that the two triangular sub-blocks have the same structure as itself.

\begin{lmm}
  \label{lmm:self-similarity}
  Given an $N$-size recurrence (\ref{eqn:delta-recur2}) that defines a matrix $\vA$, there exist $N/2$-size recurrences of the form (\ref{eqn:delta-recur2}) producing matrices $\vA_L,\vA_R$ such that
  \[ \vT_L = \vA_L \vV_{\lambda_{[0:N/2]}}^T \quad \text{and} \quad \vT_R = \vA_R \vV_{\lambda_{[N/2:N]}}^T \vE \]
  for a matrix $E$ that is a direct sum of upper-triangular Toeplitz matrices and invertible.

  Furthermore, the coefficients of these recurrences can be found with $O(\Delta^2 N\log^3 N)$ operations.
\end{lmm}
\begin{proof}
%Let $c_L(X) = \prod_{i=0}^{N/2-1} (X-M[i, i])$ and let $c_R(X) = \prod_{i=N/2}^{N-1} (X-M[i,i])$. Let $A_L,A_R$ be the unique $N/2\times N/2$ matrices corresponding to polynomials
%\[ \begin{bmatrix} a_0(X) \mod c_L(X) \\ a_1(X) \mod c_L(X) \\ \vdots \\ a_{N/2-1}(X) \mod c_L(X)\end{bmatrix} \quad \text{and} \quad \begin{bmatrix} a_{N/2}(X) / c_L(X) \\ a_{N/2+1}(X) / c_L(X) \\ \vdots \\ a_{N-1}(X) / c_L(X) \end{bmatrix} \]
  Define $c_L(X) = \prod_{i = 0}^{N/2-1} (X-\lambda_i)$ and $c_R(X) = \prod_{i=N/2}^{N-1} (X-\lambda_i)$, corresponding to the left and right halves of the diagonal elements.

  \textbf{Structure of $\vT_R$} 

  $\vT_R$ consists of the higher order (derivative) evaluations of $a_{[N/2:N]}$. However, we would like to express it instead as low order evaluations (namely, corresponding to $\vV_{\lambda_{[N/2:N]}}$) of low-degree polynomials. Fix a root $\lambda$ of $p_{\lambda_{[N/2:N]}}$ and suppose that $\vT_R$ ``contains'' the evaluations $a_i^{(j)}(\lambda), \dots, a_i^{(k)}$ for $k \ge j \ge 0$. Equivalently, $j$ is the multiplicity of $\lambda$ in $\lambda_{[0:N/2]}$ and $k-j$ is the multiplicity in $\lambda_{[N/2:N]}$.

  Let $q_i(X) = a_i(X) / c_L(X)$ and consider the Taylor expansions of $a_i(X),q_i(X),c_L(X)$ around $\lambda$
  \begin{equation*}
	a_i(X) = \displaystyle\sum_{\ell = j}^{\infty} \frac{\alpha_i^{(\ell)}(\lambda)}{\ell!} (X - \lambda)^\ell \qquad
	q_i(X) = \displaystyle\sum_{\ell = 0}^{\infty} \frac{q_i^{(\ell)}(\lambda)}{\ell!}(X - \lambda)^\ell \qquad
	c_L(X) = \displaystyle\sum_{\ell = j}^{\infty} \frac{c_L^{(\ell)}(\lambda)}{\ell!}(X - \lambda)^\ell \qquad
  \end{equation*}

  Since polynomial multiplication corresponds to a convolution of coefficients,
%The expression $a_i(X) = \displaystyle\left(\sum_{\ell=0}^\infty \frac{q_i^{(\ell)}(\lambda)}{\ell!} (X - \lambda)^\ell \right)\left(\sum_{\ell = j}^{\infty} \frac{c_L^{(\ell)}(\lambda)}{\ell!}(X - \lambda)^\ell\right)$ implies
  \[
	\begin{bmatrix} a_i^{(j)}(\lambda)/j! \\ \vdots \\ a_i^{(k)}(\lambda)/k! \end{bmatrix}
	= 
	\begin{bmatrix} q_i^{(0)}(\lambda)/0! \\ \vdots \\ q_i^{(k-j)}(\lambda)/(k-j)! \end{bmatrix}
	\ast
	\begin{bmatrix} c_L^{(j)}(\lambda)/j! \\ \vdots \\ c_L^{(k)}(\lambda)/k! \end{bmatrix}
  \]
  where $*$ denotes the convolution. Consider the matrix $\vA_R$ where row $i$ consists of the coefficients of $q_i = a_i/c_L$. If $S$ is the indices of the subsequence of $\lambda_{[N/2:N]}$ corresponding to $\lambda$, then the above equation can be equivalently written
  \[
	(\vT_R)_{i,S} = (\vA_R\vV_{\lambda_{[N/2:N]}}^T)_{i,S}
	\begin{bmatrix}
	c_L^{(j)}(\lambda)/j! & \cdots & c_L^{(k)}(\lambda)/k! \\
	\vdots & \ddots & \vdots \\
	0 & \cdots & c_L^{(j)}(\lambda)/j! \\
	\end{bmatrix}
  \]
  There is an analogous upper-triangular Toeplitz matrix for each $\lambda$, so define $\vE$ to be the $N\times N$ matrix that is the appropriate direct sum of these Toeplitz matrices, such that $(\vT_R)_{i} = (\vA_R\vV_{\lambda_{[N/2:N]}}^T)_{i}\vE$ holds. Note that the unique entries of $\vE$ correspond to the evaluation of $c_L$ at the second half of $\vV_{\lambda_{[0:N]}}$ which can be computed in $O(N\log^2 N)$. Also, $\vE$ is upper-triangular with non-zero diagonal. Since the above equation holds for all $i \in [N/2:N]$, we can factor
  \[
	\vT_R = (\vA_R\vV_{\lambda_{[N/2:N]}}^T) \vE
  \]
  Finally, the polynomials that $\vA_R$ correspond to satisfy a recurrence
  \begin{equation*}
	(X - \lambda_i) q_i(X) = \sum_{j=1}^{\min(\Delta, N-i)} c_{i,j} q_{i+j}(X)+d_i p_{\lambda_{[N/2:N]}}(X)
  \end{equation*}
  which directly follows from dividing (\ref{eqn:delta-recur2}) by $p_{\lambda_{[0:N/2]}}$.

%Note that $\vA_R \vV_{c_R}$ computes $q_i^{(0)}(\lambda), \dots, q_i^{(k-j)}(\lambda)$ for the necessary $\lambda$. To recover the evaluations in $\vT_R$, we need to scale the values of $\vA_R \vV_{c_R}$ (right multiplying by a diagonal matrix), convolve them with the appropriate vector (right multiplying by a block Toeplitz), and then scale them again (another right multiplication by a diagonal matrix). We thus have that $\vT_R = \vA_R \vV_{c_R} \vE \vP_R$ where $\vE$ is the product of diagonal, block Toeplitz, and diagonal matrices. 

  \textbf{Structure of $\vT_L$}

  $\vT_L$ corresponds to evaluations of $a_{[0:N/2]}(X)$ and their derivatives at $\lambda_{[0:N/2]}$. Note that we can subtract $c_L(X)$ from any polynomial without changing these evaluations. So if we define the polynomials $q_i(X): i \in [0:N/2]$ to be the unique polynomials of degree less than $N/2$ such that $q_i = a_i \pmod{c_L}$ and $\vA_L$ to be the corresponding matrix of coefficients of $q_i$, then $\vT_L = \vA_L \vV_{\lambda_{[0:N/2]}}^T$.

  It remains to show that the $q_i$ satisfy a recurrence of the form (\ref{eqn:delta-recur2}) and to specify its coefficients. Then $\vA_L$ will be parameterized the same way as $\vA$, completing the self-similarity result that $\vT_L$ has the same structure as $\vA\vV_{\lambda_{[0:N/2]}}^T$ but of half the size.

  Consider an $i \in [0:N/2]$. Note that (\ref{eqn:delta-recur2}) implies that
  \begin{align}
	(X-\lambda_i)a_i(X) &= \sum c_{i,j}a_{i+j}(X) \pmod{p_{\lambda_{[0:N]}}(X)} \notag\\
	\label{eq:atri-q}
	(X-\lambda_i)a_i(X) &= \sum c_{i,j}a_{i+j}(X) \pmod{p_{\lambda_{[0:N/2]}}(X)} \\
	(X-\lambda_i)q_i(X) &= \sum c_{i,j}q_{i+j}(X) \pmod{p_{\lambda_{[0:N/2]}}(X)} \notag
  \end{align}
  so that we know there exists scalars $d'_i$ such that
  \begin{equation}
	\label{eqn:delta-recur2-reduced}
	(X - \lambda_i) q_i(X) = \sum_{j=1}^{\min(\Delta, N-i)} c_{i,j} q_{i+j}(X)+d'_i p_{\lambda_{[0:N/2]}}(X)
  \end{equation}
  and the goal is to find these scalars. Define $b_i(X) = (a_i-q_i)/p_{\lambda_{[0:N/2]}}$ which is a polynomial from the definition of $q_i$. Subtract the previous equation from (\ref{eqn:delta-recur2}) and divide out $p_{\lambda_{[0:N/2]}}$ to obtain
  \[
	(X - \lambda_i) b_i(X) = \sum_{j=1}^{\min(\Delta, N-i)} c_{i,j} b_{i+j}(X) + (d_i p_{\lambda_{[N/2:N]}}(X) - d'_i)
  \]
  Therefore $d'_i$ is the unique element of $\mathbb F$ that makes the RHS of the above equation divisible by $(X-\lambda_i)$. This element is just $\sum c_{i,j} b_{i+j}(\lambda_i) + d_i p_{\lambda_{[N/2:N]}}(\lambda_i)$. Note that $p_{\lambda_{[N/2:N]}}$ can be evaluated at all $\lambda_{[0:N/2]}$ in time $O(N\log^2 N)$ time, so that term can be treated separately. Then define $d''_i := d'_i - d_i p_{\lambda_{[N/2:N]}}(\lambda_i)$ which is equivalently defined as the unique number making $\sum_{j=1}^{\min(\Delta, N-i)} c_{i,j} b_{i+j}(X) - d''_i$ divisible by $(X-\lambda_i)$. It suffices to find the $d''_i$.

  \begin{prop}
	Suppose $b_{N}, \ldots, b_{N+\Delta}$ are polynomials of degree $N$ and $\lambda_i, c_{i,j}$ are some fixed scalars. For $i = N-1, \ldots, 0$, define $b_i$ to be unique polynomial such that
	\[ (X - \lambda_i) b_i(X) = \sum_{j=1}^{\min(\Delta, N-i)} c_{i,j} b_{i+j}(X)- d''_i \]
	holds for some $d''_i$ (which is also unique). Then we can find all the $d''_i$ in time $O(\Delta^2 N\log^3 N)$.
  \end{prop}
  \begin{proof}
	If $N = O(\Delta)$, we can explicit compute the $b_i$ and $d''_i$ naively in time at most $O(\Delta^3)$. Otherwise we will find the $d''_i$ with a divide-and-conquer algorithm.

        As previously noted, the value of $d''_i$ depends only on the values of $b_{[i+1:i+\Delta]}$ evaluated at $\lambda_i$. Conversely, the polynomial $b_i$ contributes to the result only through its evaluations at $\lambda_{[i-\Delta:i-1]}$. In particular, $b_{[N/2+\Delta:N+\Delta]}$ can be reduced $\pmod{\prod_{[N/2:N]}(X-\lambda_j)}$ without affecting the $d''_i$. So the remainders of the starting conditions $b_{[N:N+\Delta]}$ mod $\prod_{[N/2:N]}(X-\lambda_j)$ suffice to recursively compute $d''_{[N/2:N]}$. Using these, we can use the ranged transition matrix to compute $b_{[N/2:N/2+\Delta]}$. Reducing these $\pmod{\prod_{[0:N/2]}(X-\lambda_j)}$ and recursing gives $d''_{[0:N/2]}$.
        %\agu{the reason why we can reduce mod the product might need to be elaborated}.
	%\begin{enumerate}
	%  \item Reduce starting conditions $$.
	%  \item Recurse on the $N/2$-size subproblem with these starting conditions to find $d''_{[N/2:N]}$.
	%  \item Using $d''_{[N/2:N]}$, jump the recurrence to compute $b_{[N/2:N/2+\Delta]}$. Reduce these $\pmod{\prod_{[0:N/2]}(X-\lambda_j)}$ and recurse which computes $d''_{[0:N/2]}$.
	%\end{enumerate}
	The base cases need $O(\Delta^3\cdot N/\Delta )$ operations which is dominated by the main recursion, which has runtime $T(N) = 2T(N/2) + \Delta^2 N\log^2 N$ resolving to $O(\Delta^2 N\log^3 N)$ assuming that the ranged transition matrices for the jumps are pre-computed.
  \end{proof}
  This auxiliary algorithm gives us the following result.
  \begin{prop}
	Given the coefficients of the recurrence (\ref{eqn:delta-recur2}) and starting conditions $a_{[N/2:N/2+\Delta]}$, we can find the coefficients of recurrence (\ref{eqn:delta-recur2-reduced}) in $O(\Delta^2 N\log^3 N)$ operations.
  \end{prop}
  \begin{proof}
    First find $b_{[N/2:N/2+\Delta]}$ which reduces $a_i$ mod $p_{\lambda_{[0:N/2]}}$. Then run the above algorithm to find all $d''_{[0:N/2]}$. Finally, set $d'_i = d''_i + d_i p_{\lambda_{[N/2:N]}}(\lambda_i)$. The bottleneck is the auxilliary algorithm which requires $O(\Delta^2 N\log^3 N)$ time.
    %\begin{enumerate}
    %  \item Find $b_{[N/2:N/2+\Delta]}$ which reduces $a_i$ mod $p_{\lambda_{[0:N/2]}}$. Time: $O(\Delta N\log N)$ operations

    %  \item Run algorithm [cite above prop] to find all $d''_{[0:N/2]}$. Time: $O(\Delta^2 N\log^3 N)$.
    %  \item Set $d'_i = d''_i + d_i p_{\lambda_{[N/2:N]}}(\lambda_i)$. Time: $O(N\log^2 N)$.
    %\end{enumerate}
  \end{proof}
  \begin{cor}
    Given the coefficients of recurrence (\ref{eqn:delta-recur2-reduced}), we can find coefficients of (\ref{eqn:delta-recur2}) in $O(\Delta^2 N\log^3 N)$ operations.
  \end{cor}
  \begin{proof}
    Same as above, but the last step is reversed. Given $d'_i$, we compute $d_i$ as $d_i = (d'_i-d''_i)/p_{\lambda_{[N/2:N]}}(\lambda_i)$. This is well-defined if $p_{\lambda_{[N/2:N]}}(\lambda_i)$ is non-zero. Otherwise, $\lambda_i$ appears later in the sequence and we can just set $d_i = 0$ by Theorem~\ref{thm:greedy-error-coef}.
  \end{proof}

%Our permutation matrix $\vP$ does have some structure with respect to its action on $\vV$: no two columns corresponding to the same eigenvalue are inverted. In particular, $P$ ensures that the $i^{th}$ column of $\vV \vP$ corresponds to the evaluation of a polynomial $q(x)$ at $q^{(n_i)}(M[i,i])$. This structure immediately implies that $\vT_L = \begin{bmatrix}a_0(X) \\ \vdots \\ a_{N/2-1}(X)\end{bmatrix} \vV_{c_L}^T \vP_L$ for some permutation matrix $\vP_L$. We note that taking these polynomials $\mod c_L(X)$ necessarily leaves the required evaluations unchanged.

  This completes the proof of the self-similarity structure of sub-blocks of $\vA\vV_{\lambda_{[0:N/2]}}^T$.
\end{proof}

Using Lemma~\ref{lmm:self-similarity}, we can pick coefficients $d_i$ for (\ref{eqn:delta-recur2}) that generate an invertible $\vA$, and also multiply a vector by $\vA^{-1}$.
\begin{cor}
  Given a size-$N$ recurrence (\ref{eqn:delta-recur2}) with fixed recurrence coefficients and starting conditions, and unspecified error coefficients $d_i$, we can pick the $d_i$ such that the resulting matrix $\vA$ is invertible in time $O(\Delta^2 N\log^4 N)$.
\end{cor}
\begin{proof}
  If $N = O(\Delta)$, then we can explicitly compute $a_{N-1},\ldots,a_{0}$ in order using the recurrence while picking $d_{N-1},\ldots,d_0$ appropriately; the results of Section~\ref{sec:error-coef} ensure that this is possible. Otherwise, we use Lemma~\ref{lmm:self-similarity} to recurse.

  Since $\vA_R$ satisfies a size-$N/2$ recurrence with known recurrence coefficients and starting conditions and unknown error coefficients, we can recursively find $d_{[N/2:N]}$ such that $\vA_R$ and hence $\vT_R$ is invertible.

  For the other half, we can jump the recurrence using $d_{[N/2:N]}$ to find $a_{[N/2:N/2+\Delta]}$ (time $O(\Delta^2 N\log^2 N)$). Now $\vA_L$ satisfies a size-$N/2$ recurrence with known recurrence coefficients and unknown error coefficients, and we can use [corollary inside the lemma] to find $d_i$ such that $\vA_L$ is invertible.

  Thus we have picked $d_i$ such that $\vT_L,\vT_R$ are invertible, and by triangularity so is $\vA\vV_{\lambda_{[0:N/2]}}^T$.

  Reducing to the subproblems is $O(\Delta^2 N\log^3 N)$ work by Lemma~\ref{lmm:self-similarity}.
\end{proof}

\begin{cor}
  \label{cor:inverse-multiply}
  Given a fully specified recurrence (\ref{eqn:delta-recur2}) such that $\vA\vV_{\lambda_{[0:N/2]}}^T$ is invertible, we can compute $(\vA\vV_{\lambda_{[0:N/2]}}^T)^{-1}\vb$ for any vector $\vb$ in $O(\Delta^2 N\log^4 N)$ operations.
\end{cor}
\begin{proof}
  Note that inversion of a triangular block matrix can be expressed using the Schur complement as
  \[ \begin{bmatrix} \vA & \vB \\ \vzero & \vC \end{bmatrix}^{-1} = \begin{bmatrix} \vA^{-1} & -\vA^{-1}\vB\vC^{-1} \\ \vzero & \vD^{-1} \end{bmatrix} \]

  So by Lemma~\ref{lmm:self-similarity}, the desired product can be written as
  \[
	(\vA\vV_{\lambda_{[0:N/2]}}^T)^{-1}\vb =
	\begin{bmatrix}
	(\vA_L\vV_{\lambda_{[0:N/2]}}^T)^{-1} & -(\vA_L\vV_{\lambda_{[0:N/2]}}^T)^{-1}\vB(\vA_R\vV_{\lambda_{[N/2:N]}}^T\vE)^{-1} \\
	\vzero & (\vA_R\vV_{\lambda_{[N/2:N]}}^T\vE)^{-1}
	\end{bmatrix} \vb
	=
	\begin{bmatrix}
	(\vA_L\vV_{\lambda_{[0:N/2]}}^T)^{-1} \left( \vb_L - \vB\vE^{-1}(\vA_R\vV_{\lambda_{[N/2:N]}}^T)^{-1}\vb_R \right) \\
	\vE^{-1}(\vA_R\vV_{\lambda_{[N/2:N]}}^T)^{-1}\vb_R
	\end{bmatrix}
  \]
  where $\vb_L = \vb_{[0:N/2]}, \vb_R = \vb_{[N/2:N]}$.

  So $(\vA\vV_{\lambda_{[0:N/2]}}^T)^{-1}\vb$ can be computed with three matrix-vector multiplications: by $(\vA_R\vV_{\lambda_{[N/2:N]}}^T)^{-1}$, $\vE^{-1}$, $\vB$, and $(\vA_L\vV_{\lambda_{[0:N/2]}}^T)^{-1}$, in order. Note that multiplying by $\vE^{-1}$ is easy since it is a direct sum of upper-triangular Toeplitz matrices. $\vB$ is a submatrix of $\vA$ which we can multiply in $\Delta^2 (N \log^3 N)$ operations. Finally, the first and last multiplications are identical subproblems of half the size. 
\end{proof}

Thus we have shown:
\begin{thm}
  Let $\vM$ be an upper-triangular, $\Delta$-banded matrix whose minimal polynomial equals its characteristic polyomial. Then $M$ is $(\Delta^{\omega_\calR}, \Delta^2)$ Jordan efficient. More precisely, there exists a Jordan decomposition $\vM = \vA\vJ\vA^{-1}$ such that with $O(\Delta^\omega_\calR N\log^2 N + \Delta^2 N \log^4 N)$ pre-processing time, multiplication by $\vA,\vA^{-1}$ can be computed in $O(\Delta^2 N \log^4 N)$ operations.
\end{thm}

Such matrices satisfy the property that their evaluation at any analytic function $f$ also admits super-fast matrix-vector multiplication.
\begin{lmm}[\cite{higham2008}]
  \label{lmm:Jordan-matrix-function}
  Let $\vJ$ be a Jordan block with diagonal $\lambda$ and $f(z)$ be a function that is analytic at $\lambda$. Then
  \[
    f(\vJ)
    =
    \begin{pmatrix}f(\lambda) & f'(\lambda) & \cdots & \frac{f^{(n-1)}(\lambda)}{(n-1)!} \\ 0 & f(\lambda) & \cdots & \frac{f^{(n-2)}(\lambda)}{(n-2)!} \\ \vdots & \vdots & \ddots & \vdots \\ 0 & 0 & \cdots & f(\lambda)
    \end{pmatrix}
  \]
\end{lmm}
\begin{proof}
  The proof follows from considering the Taylor series expansion $f(z) = \sum \frac{f^{(i)}(\lambda)}{i!} (z-\lambda)^i$ and plugging in $\vJ = \lambda\vI + \vS^T$.
\end{proof}

\begin{lmm}
  \label{lmm:Jordan-mult}
  Let $\vM$ be $(\alpha,\beta)$-Jordan efficient and $f(z)$ be a function that is analytic at the eigenvalues of $\vM$. Then $f(\vM)$ admits super-fast matrix multiplication in $O(\beta)$ operations, with $O(\alpha)$ pre-processing.
\end{lmm}
\begin{proof}
  We can write a matrix-vector product as
  \[ f(\vM)\vb = f(\vA\vJ\vA^{-1})\vb = \vA f(\vJ) \vA^{-1} \vb \]
  By Lemma~\ref{lmm:Jordan-matrix-function}, assuming access to the multi-point Hermite-type evaluations of $f$ at the eigenvalues of $\vJ$, the product $f(\vJ)\vb$ is a series of Toeplitz multiplications which can be computed in time $O(N\log N)$. Thus this product is bottlenecked by the time it takes to multiply by $\vA$ and $\vA^{-1}$, and the result follows.
\end{proof}

For certain classes of matrices, our techniques facilitate fast computation and succinct description of the Jordan decomposition. We list some cases of matrices that are Jordan efficient because the change of basis matrix $\vA$ has low recurrence width. Note that these matrices are even more structured than general matrices of recurrence width $\Delta$, allowing us to multiply by the inverse as well.
\agu{I'm no longer sure about how to present these claims\ldots I'm not sure if the $\vA$ that arise here are included in the matrices we can invert\ldots again I feel pretty iffy about the Jordan efficiency stuff without heavy cleanup}
\begin{enumerate}
  \item Suppose that $\vM$ is a Jacobi matrix. Then it is diagonalizable with $\vM = \vA\vD\vA^{-1}$, and $\vA$ is an orthogonal polynomial matrix. In this case, multiplication by $\vA$ is our standard result and multiplication by $\vA^{-1}$ can be done using some properties of orthogonal polynomials, i.e.\ detailed in~\cite{bostan2010op}.
  \item Suppose that $\vM$ is triangular $\Delta$-band, and $\vM$'s minimal polynomial equals its characteristic polynomial. In this case, it turns out that the change of basis matrix can be expressed as $\vA\vV^T\vP$, where $\vA$ is a $\Delta$-width recurrence, $\vV$ is a confluent Vandermonde matrix, and $\vP$ is a permutation matrix, and furthermore $\vA\vV^T$ is triangular. We can show how to describe $\vA$ with $\widetilde{O}(\Delta^\omega N)$ operations, as well as perform matrix-vector multiplication by $\vA$ and $\vA^{-1}$ in $\widetilde{O}(\Delta^2 N)$ operations. The full proof is in Appendix~\ref{subsec:jordan}.
  \item When $\vM$ is triangular $\Delta$-band and diagonal, we can also compute and multiply by $\vA$ efficiently using similar techniques to the above case.
\end{enumerate}

Now we show how Krylov efficiency can be reduced to Jordan efficiency. Suppose that $\vR$ is Jordan efficient. Observe that the Krylov multiplication can be expressed as
\begin{align*}
  \kry(\vR, \vy)\vx &= \sum_{i=0}^{N-1} \vx[i] (\vR^i \vy) \\
  &= \left( \sum_{i=0}^{N-1} \vx[i] \vR^i \right) \vy \\
  &= \vx(\vR) \vy
\end{align*}
where we define the function $\vx(X) = \sum \vx[i] X^i$.

Note that $\vx(X)$ is analytic and the multi-point Hermite-type evaluation problem on it is computable in $O(N\log^2{N})$ time~\cite{OS00,P01}. Thus the Krylov multiplication can be performed using Lemma~\ref{lmm:Jordan-mult}, and Krylov efficiency of $\vR$ reduces exactly to Jordan efficiency with the same complexity bounds. Therefore all Jordan-efficient matrices are also Krylov-efficient.

\subsection{Bernoulli Polynomials}
\label{sec:bernoulli}

%\paragraph{Computing sequences.} Second 
We observe that the recurrences that we consider for bounded \width\ matrices actually have {\em holonomic} sequences as a special case, which are some of the very well studied sequences and many algorithms for such sequences have been implemented in algebra packages~\cite{holonomic}. Additionally, the computation of many well-known sequences of numbers, including Stirling numbers of the second kind and Bernoulli numbers, are also very well studied problems. 
%\ar{Are there any good refs for these named sequences? Also even thugh these are not holonomic are they of low displacement rank?} 
There have been some recent ad-hoc algorithms to compute such numbers (e.g.\ the current best algorithm to compute the Bernoulli numbers appears in~\cite{bernoulli}). Despite these sequences not being holonomic, our general purpose algorithms can still be used to compute them. Next, we show how to compute the first $N$ Bernoulli numbers in $O(N\log^2 N)$ operations, which recovers the same algorithmic bit complexity (potentially with a log factor loss) as the algorithms that are specific to Bernoulli numbers.

%\agu{above moved from intro}

Bernoulli polynomials appear in various topics of mathematics including the Euler-Maclaurin formulae relating sums to integrals, and formulations of the Riemann zeta function~\cite{apostol}. In this section we will demonstrate how our techniques are flexible enough to calculate quantities such as the Bernoulli numbers, even though they do not ostensibly fall into the framework of bounded-width linear recurrences.

Bernoulli polynomials are traditionally defined by the recursive formula
\[B_i(X) = X^i - \sum_{k=0}^{i-1} \binom{i}{k}\frac{B_k(X)}{i-k+1}\]
with $B_0(X) = 1$~\cite{bernoulli-facts}. This recurrence seemingly cannot be captured by our model of recurrences, since each polynomial depends on \emph{every previous polynomial}. However, with some additional work, a recurrence with bounded \width\ can also capture this larger recurrence, thereby facilitating superfast matrix-vector multiplication involving $B$. 

We start out by computing $B_i(0)$ for each $i$. For notational convenience, we may drop the $0$ and use $B_i$ to denote $B_i(0)$. In particular $$B_i = \sum_{k=0}^i (-1)^k \frac{k!}{k+1} {i \brace k}$$ where ${i \brace k}$ denotes the Stirling numbers of the second kind~\cite{bernoulli-facts}. To compute $B_i$ for $0 \le i \le N$, we define the Stirling matrix $\vW$ such that $\vW[i,j] = {i \brace j}$ and a vector  a diagonal matrix $\vD$ such that $\vD[i,i] = i!$, and a vector $\vx$ such that $\vx[k] = \frac{(-1)^k k!}{k+1}$; the vector $\vb = \vW \vx$ will contain $B_i(0)$ as $\vb[i]$.

\begin{lmm}
  \label{lmm:bernoulli-W}
If $\vW[i,j] = {i \brace j}$, we can multiply $\vW \vx$ or $\vW^T \vx$ for any vector $\vx$ with $O(N \log^2 N)$ operations.
\end{lmm}
\begin{proof}
Note that the Stirling numbers satisfy the recurrence~\cite{bernoulli-facts}
\[ {i+1 \brace k}  = k {i \brace k} + {i \brace k-1}. \]
If we let $\vf_i = \vW[i, :]^T$, the recurrence gives us that $\vf_{i+1} = (\vD + \vS) \vf_i$. Note that by Theorem~\ref{thm:ut-band-ke}, $(\vD + \vS)$ is $(1,1)$-Krylov efficient. Theorem~\ref{thm:mainATb} completes the proof.
\end{proof}
\bcor
We can compute $B_i(0)$ for all $i$ with $O(N \log^2 N)$ operations.
\ecor

We note that the $B_i(0)$ are actually the Bernoulli numbers, and our algorithm facilitates the computation of the Bernoulli numbers in the same runtime as recent state-of-the-art ad hoc approaches~\cite{bernoulli}. (See Section~\ref{sec:numerical} for more.)
%\ar{Should we state that computing $\vB\vz$ allows us to compute $N$ moments of a Poisson distribution simultaneously?}

We now focus on a matrix $\vZ$ such that $\vZ[i,j] = B_i(\alpha_j)$ for $\alpha_0, \dots, \alpha_{N-1} \in \F$. Note that a superfast multiplication algorithm for $\vZ$ immediately implies one for $\vB$ (the matrix of coefficients of $B_i$) since $\vZ = \vB\vV_{\alpha}^T$, where $\vV_{\alpha}$ is the Vandermonde matrix defined by evaluation points $\alpha_0,\dots,\alpha_{N-1}$. We take advantage of the following identity relating $B_i(X)$ to $B_i$: 
\[B_{i+1}(X) = B_{i+1} + \sum_{k=0}^i \frac{i+1}{k+1} {i \brace k} (X)_{k+1}\]
where $(X)_{k+1} = X (X-1) \cdots (X-k)$ denotes the falling factorial~\cite{bernoulli-facts}.

\begin{lmm}
  \label{lmm:bernoulli-F}
Suppose a matrix $\vF$ is defined such that $f_i(X) = \sum_{j=0}^{N-1} \vF[i,j] X^j = (X)_i$. Then for any vector $\vx$, we can multiply $\vF \vx$ and $\vF^T \vx$ in $O(N \log^2 N)$.
\end{lmm}
\begin{proof}
By definition $f_{i+1}(X) = (X-i) f_i(X)$, and hence, Theorem~\ref{thm:mainATb} imply that we can multiply $\vF$ and $\vF^T$ by vectors in $O(N \log^2 N)$.
\end{proof}
\agu{In two places above, add reference to $\vA\vb$ algorithm.}

\begin{thm}
For any vector $\vb$, we can compute $\vZ \vb$ or $\vZ^T \vb$ with $O(N \log^2 N)$ operations.
\end{thm}
\begin{proof}
  Suppose we define $\vC_{i+1,j} = \frac{i+1}{j+1} {i \brace j}$. Note that $\vC$ is just the Stirling matrix $\vW$ multiplied by two diagonal matrices, implying it supports $O(N \log^2 N)$ vector multiplication. Furthermore, as in our previous lemma, define $\vF$ to be a matrix corresponding to the falling factorials. Then $(\vC \vV_{\alpha} \vF^T)[i,j] = B_i(\alpha_j) - B_i$. So if we define a matrix $\vM$ such that $\vM[i,j] = B_i$, $\vZ = \vC \vV_{\alpha} \vF^T + \vM$. Thus multiplying by $\vZ$ or $\vZ^T$ reduces to multiplying by these components, which by Lemmas~\ref{lmm:bernoulli-W} and~\ref{lmm:bernoulli-F} can be done in $O(N\log^2 N)$ operations.
\end{proof}

As discussed previously, this theorem immediately implies that we can compute matrix-vector multiplications involving $\vB$ in $O(N \log^2 N)$ operations as well. 
%We note that the numerical errors that may arise throughout these computations come from rounding and not having enough bits to represent the data; there are no catastrophic cancellation errors. \ar{Can this statement be justified a bit better? After all there are some negative numbers and as such they can cancel out some stuff.}\rpu{Not sure why this true...do you have any thoughts Chris?}
  As such, we only need to analyze the bit complexity of the algorithm to understand its numerical properties (which we do in Section~\ref{sec:numerical}). 

\subsection{Bit Complexity}
\label{sec:numerical}
%\label{sec:bit-comp}

It turns out that it is fairly easy to analyze the bit-complexity of our algorithms. In particular, we note that all the basic operations in our algorithms reduce to operations on polynomials. In particular, consider a matrix with \width\ of $t$ over the set of integers $\Z$. (We note that in most of the problems of computing sequences the input is indeed over $\Z$.) Note that our results on the efficiency of our algorithms are stated in terms of the number of operations of $\Z$. When dealing with the bit complexity of our algorithms, we have to worry about the size of the integers. It can be verified that given the input recurrence $(\vG,\vF)$, all the integers are of size $\left(\max(\|\vF\|_{\infty},\|\vG\|_{\infty}\right)^{O(N)}$. Since two $n$-bits integers are can multiplied with $O(n\log{n}\log\log{n})$-bit operations, this implies to obtain the bit complexity of our algorithms, we just need to multiply our bounds on the number of operations by $\tO(N\cdot \max(\|\vF\|_{\infty},\|\vG\|_{\infty}))$.

In particular, the above implies that the Bernoulli numbers can be computed with $\tO(N^2)$-bit operations. This is within $\tO(1)$ factors of the best known algorithm developed specifically for this problem in~\cite{bernoulli}.

%\ar{I will add in a para here: there should be a theorem for Section~\ref{sec:bernoulli} to refer to.}

\subsection{Preliminary Experimental Results}
\label{sec:experiments}

We coded the basic Algorithm~\ref{algo:transpose-mult} in C++ and compared with a naive brute force algorithm. We would like to stress that in this section, we only present preliminary experimental results with the goal of showing that in principle the constants in our algorithms' runtimes are reasonable (at least to the extent that preliminary implementations of our algorithms beat the naive brute force algorithm).

We compare 3 facets of the algorithms: the preprocessing step, computing $\vA \vb$, and computing $\vA^T \vb$. We consider matrices $\vA$ whose rows are the coefficients of polynomials that follow our basic recurrence: 
\[ f_{i+1}(X) = \sum_{i=0}^t g_{i,j}(X) f_i(X) \]
where $\deg(g_{i,j} = j, \deg(f_i) = i$. We generate the coefficients of $f_i(X)$ for $i \le t$, the coefficients of $g_{i,j}(X)$ for $i > t$, and the elements of $\vb$ pseudo-randomly in the range [-1, 1] using the C++ \texttt{rand()} function. The input to the algorithms are the $f_i(X)$ for $i \le t$, $g_{i,j}(X)$ for $i > t$, and $\vb$. 

The brute force's preprocessing step is to explicitly compute the $\binom{N+1}{2}$ polynomial coefficients of $f_i(X)$ for all $i$, thereby computing the non-zero element of $\vA$. The vector multiplication is the straightforward $O(N^2)$ algorithm. Our approach's preprocessing step uses the naive cubic matrix multiplication algorithm. And it uses the open-source library FFTW to compute FFTs. 
%\ar{Should we keep the last line in since the current algorithm statement does not explicitly use FFTs? I guess the more general point is that this implementation is for a previous version of our algorithm, right? I'll leave it to Chris to figure out a way to finesse this :-)}
%\rpu{You still need FFTs to compute polynomial multiplication fast. Also, I think the newer algorithm is still the same computation, even if its expressed different/more elegantly.}

The experiments below were run on a 2-year old laptop with an i7-4500U processor (up to 3 GHz, 4 MB cache) and 8 GB of RAM.  All of the numbers below express times in seconds.

\begin{tabular} {|c|cc|cc|cc|}
\hline
& \multicolumn{2}{|c|}{Preprocessing} & \multicolumn{2}{|c|}{$\vA \vb$} & \multicolumn{2}{|c|}{$\vA^T \vb$} \\
& Brute Force & Our Approach & Brute Force & Our Approach & Brute Force & Our Approach \\ \hline
N = 100, t = 1 & 0.02 & \textbf{0.01} & \textbf{0.0004} & 0.003 & \textbf{0.0006} & 0.003 \\
N = 100, t = 4 & \textbf{0.04} & 0.07 & \textbf{0.0005} & 0.009 & \textbf{0.0005} & 0.009 \\ \hline
N = 1000, t = 1 & 1.3 & \textbf{0.13} & 0.05 & \textbf{0.04} & 0.06 & \textbf{0.04} \\
N = 1000, t = 4 & 2.8 & \textbf{1.2} & \textbf{0.04} & 0.15 & \textbf{0.05} & 0.15 \\ \hline
N = 10000, t = 1 & 127 & \textbf{1.7} & 4.8 & \textbf{0.4} & 5.6 & \textbf{0.5} \\
N = 10000, t = 4 & 322.8 & \textbf{18.2} & 4.2 & \textbf{1.8} & 5.3 & \textbf{1.8} \\ \hline
\end{tabular}

At $N=100$, the brute force clearly outclasses ours, but notably our preprocessing isn't much slower than what brute force requires. Our approach starts to perform better at $N=1000$ where the multiplication algorithms are similar in runtime to the brute force multiplication algorithms and are significantly faster than the brute force preprocessing. And at $N=10000$, our approach universally outperforms the brute force approach.

\end{document}